\documentclass[12pt]{article}

\usepackage{xr}
\usepackage{subfiles}
\usepackage{amsthm}
\newtheorem{assumption}{Assumption}
\newtheorem*{wkassumption}{Assumption \continuation}
\newcommand{\continuation}{??}
\newenvironment{weakassumption}[1]
 {\renewcommand{\continuation}{\ref{#1}*}\wkassumption}
 {\endwkassumption}

\newtheorem{theorem}{Theorem}
\usepackage{bm}
\usepackage{amsmath}
\usepackage{amssymb,hyperref}
\usepackage{graphicx}
\usepackage[margin=0.65in]{geometry}
\usepackage[authoryear]{natbib}
\usepackage{bbm}
\usepackage{setspace}
\usepackage{enumerate}
\usepackage{multirow}
\usepackage{multicol}
\usepackage{colortbl}
\usepackage{hhline}
\usepackage{longtable}
\usepackage{array}
\usepackage{hyperref}
\usepackage{float}
\usepackage{wrapfig}
\usepackage{xcolor}
\usepackage{booktabs} % for nicer rulers
\usepackage{xpatch}
\usepackage{tabulary}
\makeatletter
\xpatchcmd{\@thm}{\thm@headpunct{.}}{\thm@headpunct{}}{}{}
\makeatother
\usepackage{tikz}
\newcommand*\circled[1]{\tikz[baseline=(char.base)]{
            \node[shape=circle,draw,inner sep=2pt] (char) {#1};}}

            \def\indep{{\,\perp \!\!\! \perp\,}}

\allowdisplaybreaks

\newtheorem{lemma}{Lemma}

\newcommand{\blocktheorem}[1]{%
  \csletcs{old#1}{#1}% Store \begin
  \csletcs{endold#1}{end#1}% Store \end
  \RenewDocumentEnvironment{#1}{o}
    {\par\addvspace{1.5ex}
     \noindent\begin{minipage}{\textwidth}
     \IfNoValueTF{##1}
       {\csuse{old#1}}
       {\csuse{old#1}[##1]}}
    {\csuse{endold#1}
     \end{minipage}
     \par\addvspace{1.5ex}}
}

\newcommand{\ts}{\textsuperscript}

\usepackage{tikz}
\usetikzlibrary{shapes.geometric, arrows.meta, positioning, shapes, calc, intersections,decorations.text}

\colorlet{shadecolor10}{gray!10}
\colorlet{shadecolor20}{gray!20}
\colorlet{shadecolor40}{gray!40}
\colorlet{shadecolor50}{gray!50}
\colorlet{shadecolor60}{gray!60}
\colorlet{shadecolor80}{gray!80}
\colorlet{shadecolor70}{gray!70}
\colorlet{shadecolor90}{gray!90}

\colorlet{gray10}{black!10}
\colorlet{gray20}{black!20}
\colorlet{gray40}{black!40}
\colorlet{gray50}{black!50}
\colorlet{gray60}{black!60}
\colorlet{gray70}{black!70}
\colorlet{gray80}{black!80}
\colorlet{gray90}{black!90}

\begin{document}

\title{\textbf{Transportability of Principal Causal Effects}}
 
\author{Justin M. Clark\thanks{\textbf{Email}: \href{mailto:clar2272@umn.edu}{clar2272@umn.edu}, \textbf{Mailing Address}: 2221 University Ave SE, Suite 200
Minneapolis, MN 55414, USA} \hspace{1cm} Kollin W. Rott \hspace{1cm} James S. Hodges \hspace{1cm} Jared D. Huling\thanks{\textbf{Email}:\href{mailto:huling@umn.edu}{ huling@umn.edu}}\\ \\
\normalsize Division of Biostatistics and Health Data Science \\ \normalsize University of Minnesota School of Public Health, MN, USA}
\date{}

\maketitle

\begin{abstract}
Recent research in causal inference has made important progress in addressing challenges to the external validity of trial findings. Such methods weight trial participant data to more closely resemble the distribution of effect-modifying covariates in a well-defined target population. In the presence of participant non-adherence to study medication, these methods effectively transport an intention-to-treat effect that averages over heterogeneous compliance behaviors. In this paper, we develop a principal stratification framework to identify causal effects conditioning on both compliance behavior and membership in the target population. We also develop non-parametric efficiency theory for and construct efficient estimators of such ``transported'' principal causal effects and characterize their finite-sample performance in simulation experiments. While this work focuses on treatment non-adherence, the framework is applicable to a broad class of estimands that target effects in clinically-relevant, possibly latent subsets of a target population.

\noindent \textbf{Keywords:} causal inference, doubly robust estimation, generalizability, noncompliance, principal stratification
\end{abstract}
\newpage{}

\section{Introduction}
\label{introduction}

Randomized controlled trials (RCTs) are often characterized as the gold standard for evaluating treatment efficacy, owing to their internal validity \citep{juni_assessing_2001}. RCTs may nonetheless lack external validity if trial participants are not drawn from target populations relevant to clinicians and policy makers. Recent advances in causal inference have helped overcome these challenges by combining data from an RCT and a target population to produce effect estimates that are interpretable for the target population \citep{degtiar_review_2023}. Intercurrent events may, however, complicate interpretation of transported effects; without further adjustment, these complications carry over into transported treatment effects \citep{ich_guidelines}. Responding to these issues, we propose causal estimands and accompanying efficient estimators that generalize causal effects in the presence of post-randomization events.

Variation in treatment adherence is a common post-randomization event with important implications for interpretation of trial findings. According to the intention-to-treat (ITT) principle, treatment groups in RCTs are defined by treatment assigned rather than treatment taken, so reported average treatment effects reflect the impact of assignment rather than the treatment itself. Thus, generalizability analyses based on RCT data typically transport ITT effects, which average over distinct compliance behaviors \citep{dahabreh_itt}. Such an analysis implicitly assumes that the distribution of compliance behaviors in the RCT reflects that in the target population. If the treatment effect varies significantly across compliance patterns, then the reported effect estimate---transported or not---may fail to fully characterize the impact of novel interventions. In particular, this mischaracterization can be large enough to obscure real treatment effects.  

For example, an RCT studying the widely-lauded ``Health Care Hotspotting'' intervention, aimed at improving the delivery of health care for patients with excess health care utilization, gave a null point estimate of the hotspotting intervention's effect on utilization outcomes \citep{finkelstein_health_2020}. However, a recent secondary analysis of the same data found pronounced and statistically significant intervention effects among patients with a higher probability of engagement with their assigned intervention \citep{yang_hospital_2023}. Clinicians or health system administrators whose patients did not participate in the original RCT may respond to these findings with two questions. First, ``do the treatment effects estimated in the initial analysis apply to patients in my health system?'' Second, ``do high engagers in my health system reap the same benefits as high engagers in the RCT population?'' This paper provides a causal framework and statistical methods equipped to answer such questions by simultaneously addressing generalizability and non-adherence behavior. More broadly, our work provides a framework for dealing with any post-randomization event when transporting effects to new populations. Our method extends recent work in principal stratification \citep{ding_principal_2017, jiang_multiply_2022}, to allow transportation of so-called principal causal effects. We further develop a general efficiency theory for this setting.

Challenges to trial generalizability stem in part from the idiosyncratic process of trial recruitment. Even if all trial participants satisfy eligibility criteria, they may still not be a random sample from the population defined by those criteria \citep{dahabreh_extending_2020}. Further, decision-makers commonly seek to assess a treatment's effect in a different population altogether. Thus, effects estimated from an RCT sample may not be relevant to all clinical questions pertaining to a treatment.

The field of causal inference has made important progress in designing approaches that directly address such challenges to the external validity of RCTs. These approaches typically require individual patient data (IPD) from the RCT and from a representative sample of the target population. To estimate treatment effects in the target population, the two data sources are combined to ``transport'' effect estimates from the RCT to the target population \citep{colnet_causal_2024}. Loosely speaking, weighting-based transportation estimators re-weight trial participants so that their covariate data, in the aggregate, more closely resemble that of the target population. Under various assumptions, such weighting estimators consistently estimate the effect of treatment assignment in the target population. Other estimators model the outcome as a function of effect-modifying covariates, or combine weighting and outcome-model estimators in a so-called ``doubly-robust'' estimator \citep{degtiar_review_2023}. 

Transportability analyses typically have a common style. First, analysts define a causal estimand that quantifies an intervention's effect in a well-defined target population. Next, data from trials applying that intervention are combined with information characterizing the target population to marginalize observed effects over the target's distribution of possibly effect-modifying covariates. The resulting estimate is then interpreted under causal assumptions as the expected treatment effect in the target population. Just as differences between trial participants and target populations can obscure a trial's interpretation, so too can intercurrent events, such as nonadherence. In particular, in the presence of nonadherence, the transported treatment effect may or may not reflect the expected treatment effect in the target population due to differences in adherence patterns. 
This paper's aim is to develop methods using that same  framework in the presence of complex post-randomization events in the trial data to produce more interpretable effect estimates.

Several different types of approaches carefully define and identify treatment effects under varied compliance behaviors and complex intercurrent events more generally. Most begin by conceptualizing random variables $C(1)$ and $C(0)$, corresponding to the \textit{potential} treatment received under assignment to treatment and placebo, respectively. We let $C(a)=1$ imply receipt of active treatment under assignment to $a$ and $C(a)=0$ imply non-receipt of treatment under assignment to $a$. These approaches to dealing with compliance differ, however, in both the estimands constructed as a function of $C(1)$ and $C(0)$ and in the assumptions that connect these estimands to observed data.

For example, a mediation approach to compliance views the received treatment $C$ as being on a particular causal pathway by which treatment \textit{assignment} affects outcomes. The effect of treatment assignment $A$ mediated by $C$ is typically called the ``indirect effect" of assignment on outcomes, while the unmediated impact of $A$ is called the ``direct effect." Conceptually, mediation analysis treats $C$ as amenable to intervention, i.e., investigators can ask what outcomes would be observed if compliance were ``set" to a particular value \citep{robins_greenland}. We might consider, for example, the effect of assignment to treatment if compliance were set to its level under placebo \citep{pearl_2001}. 

Other approaches instead treat compliance or treatment receipt behavior $(C(1), C(0))$ as defining latent sub-populations, e.g., participants who would always take active treatment regardless of assignment or who always take the treatment assigned to them. Treatment effects are then defined conditional on membership in such sub-populations \citep{angrist_identification_1996}. While there are formal relationships between effects among compliers and the direct/indirect effects of mediation analysis, in general they are not equal and proceed from different scientific goals \citep{vanderweele_relations}. 
Approaches that target these conditional estimands differ in the types of assumptions needed for causal identification. Instrumental variables (IV) analysis is one such approach, identifying the treatment effect among participants who comply with assigned treatment, typically called the ``local [or complier] average causal effect" \citep{ding_causal_intro}. As defined in \citet{angrist_identification_1996} this identification strategy assumes that treatment assignment affects outcomes only through compliance behavior, an assumption typically termed the ``exclusion restriction." 

A more general approach to addressing post-treatment variables is principal stratification \citep{frangakis_rubin_2002}, which includes complier average causal effects (CACE) as a special case but can also be used to estimate effects conditional on any post-treatment variable the values of which define a subpopulation of interest. This generality and emphasis on subpopulations make principal stratification, in our view, especially well suited to transportability in the presence of complex intercurrent events. As we demonstrate below, principal stratification allows us to define causal effects conditional on membership in the intersection of a target population and the latent groups defined by post-treatment variables. Defining our overarching scientific goal using this kind of transportability problem permits development of tools, techniques, and intuitions from generalizability and transportability.

As described, IV analysis of nonadherence can be viewed as a special case of principal stratification with its own assumptions and inferential targets. Methods for generalizability in the IV setting have been studied previously in \citet{rudolph_robust_2017}. In this work, we instead develop a principal stratification framework with an identification strategy based on alternative assumptions, for two main reasons. First, we hope to provide an alternative to the exclusion restriction that may be more reasonable in certain settings. For example, in complex, community-based interventions like the hotspotting RCT, assignment is unblinded and might plausibly affect outcomes directly through the knowledge of being assigned to the intervention, thereby violating the exclusion restriction. 

Second, we aim to develop a statistical framework applicable to populations beyond those targeted by the complier average causal effect. Principal stratification allows us to define effects in a variety of scientifically-relevant populations. For example, the survivor average causal effect is defined for individuals who would survive regardless of treatment assignment. This population is important in clinical trials where some participants die before their primary outcome is observed \citep{rubin_censoring_death}. Another relevant population comes from the hotspotting trial, in which analyses of engagement focus on post-treatment events in the treatment group alone, which would require marginalizing over different behaviors in the control group. In both cases, effects in a meaningful population are not captured by standard CACE estimands. While this work focuses on compliance, our framework applies to causal effects at the intersection of a target population and any possibly latent principal strata. 

This paper builds on principal ignorability as an alternative to the exclusion restriction; we use this assumption to identify principal stratum membership with principal scores. Principal scores were introduced by \citet{follmann_2000}, generalizing propensity scores to predicting compliance behavior. Principal ignorability is a useful alternative to parametric approaches (see, e.g., \citet{parametric_princ_strata}), the assumptions of which may not be tenable in some situations. This concept has been developed further in, e.g., \citet{jo_stuart_principal} and \citet{feller_principal_2017}. We draw specifically on the theoretical contributions of \citet{ding_principal_2017} and \citet{jiang_multiply_2022}, which extended principal score techniques to estimating principal causal effects beyond those related to compliance and applied such extensions using modern causal inference techniques. 

The rest of this paper is organized as follows. Section \ref{defining-target-estimands} defines notation and target estimands. Section \ref{identification} introduces causal assumptions to identify our novel estimands, combining those of principal ignorability and generalizability. Section \ref{sec: estimation} considers estimation using ideas from nonparametric efficiency theory to improve efficiency over plug-in approaches. Section \ref{sec: simulation} uses simulation experiments to characterize our estimation methods. Section \ref{sec: hotspotting_analysis} applies our framework to data from the Healthcare Hotspotting Trial. Section \ref{sec: discussion} concludes and highlights avenues for future work.

\section{Defining Target Estimands}
\label{defining-target-estimands}

\subsection{Intuition for Transported Principal Causal Effects}
\label{sub: intuitionestimands_princ_strat}

To gain intuition about what is being estimated in our setting, consider how a weighting estimator might be constructed to transport a treatment effect from a trial to compliers in the target population. Given individual patient data (IPD) from the trial and from the target population, we might first use a function of covariates to characterize trial participants with (1) a high probability of compliance to assigned treatment and (2) greater similarity to the target population. These individuals' outcomes are weighted more highly, giving the desired transported effect estimate under some assumptions.

Figure \ref{fig:compliance-diagram} illustrates our approach in a highly stylized setting with two covariates $X_1$ and $X_2$. Supposing that each population --- trial participants, compliers, and the target population --- share a region of common covariate support, we informally highlight areas of high covariate density using the different shapes. For instance, compliers have highest covariate density in the rhombus shape and are more likely to have lower values of $X_2$ and higher values of $X_1$. Our methods focus on treatment effects averaged over the distribution of covariates for compliers in the target population, represented by the shaded region --- again, this region is loosely thought of as the region where compliers in the target population have highest covariate density. If $X_1$ and $X_2$ are effect-modifying, then typical principal causal effects  of compliers in the trial population will not necessarily reflect complier average effects in the target population.

%%%%%%%%%%%%%%%%%%%%%%%%%%%%%%%%%%%%%%%%%%%%%%%%%%%%%%%%%%%%%%%%%%%%%%%%%%%%%%%%%%%%%%%%%%%%%%%%
%%%%%%%%%%%%%%%% stylized depiction of compliers in target population %%%%%%%%%%%%%%%%%%%%%%%%%
%%%%%%%%%%%%%%%%%%%%%%%%%%%%%%%%%%%%%%%%%%%%%%%%%%%%%%%%%%%%%%%%%%%%%%%%%%%%%%%%%%%%%%%%%%%%%%%%
\begin{figure}[h]
    \centering
    \resizebox{0.45\textwidth}{!}{%
    \begin{tikzpicture}[scale=1.25]

  % Axes
  \draw[thick,dashed,->] (0,0) -- (6,0) node[anchor=north east] {\small $X_1$};
  \draw[thick,dashed,->] (0,0) -- (0,5) node[anchor=north east] {\small $X_2$};

  % Intersection area (gray)
  \begin{scope}
    \clip[rotate around={0:(3.5,4)}] (3.5,1.25) -- (5,2) -- (3.5,2.75) -- (2,2) -- cycle;
    \fill[gray40, rotate around={-40:(2,3)}] (2.6,2.5) ellipse (2cm and 0.8cm);
  \end{scope}
  
  %\node[blue, below right] at (4.25,1.75) {All compliers};

  % Target population ( ellipse)
  \draw[thick, gray50, rotate around={-40:(2,3)}, -stealth,
  postaction={decorate,decoration={text effects along path,
    text={ \;\; \;\; \;\; \;\; \;\; \;\; \;\; \;\; \;\; \;\; \;\;\;\;\;\;\; \;\; \;\;\;\;\; Target population}, 
    text align/align=center, 
    %text align={right indent={200pt}}, 
    %reverse path,
    text effects/.cd, font=\small, 
      text along path, % add fill=white, yshift=-0.5ex  to .style below to put over line
      every character/.style={yshift=0.5ex}}}] (2.6,2.5) ellipse (2cm and 0.8cm);
  %\node[red, above] at (0.5,4.5) {Target population};

  % Trial population ( triangle)
  \draw[thick, shadecolor70, rotate around={0:(3.5,4)}, -stealth,
  postaction={decorate,decoration={text effects along path,
    text={ \;\: Trial population}, 
    text align/align=left, 
    %text align={right indent={200pt}}, 
    reverse path,
    text effects/.cd, font=\small, 
      text along path, % add fill=white, yshift=-0.5ex  to .style below to put over line
      every character/.style={yshift=0.5ex}}}] (0.5,0.5) -- (2.75,4.25) -- (4.75,0.75) -- cycle;
  %\node[green, above right] at (5,4) {All high engagers};

  % all compliers (trapezoid)
  \draw[thick, gray80, rotate around={0:(3.5,4)}, -stealth,
  postaction={decorate,decoration={text effects along path,
    text={  \;\;\;\;\;\;\; All compliers }, 
    text align/align=center, 
    %text align={right indent={200pt}}, 
    reverse path,
    text effects/.cd, font=\small, 
      text along path, % add fill=white, yshift=-0.5ex  to .style below to put over line
      every character/.style={yshift=0.5ex}}}] (3.5,1.25) -- (5,2) -- (3.5,2.75) -- (2,2) -- cycle;

  % irregular encircling shape
  \draw[very thick, dashed, black, -stealth,
  postaction={decorate,decoration={text effects along path,
    text={\;\;\;\;\;\;\;\;\;\;\;\;\;\; Region of common support}, 
    text align/align=center, 
    %text align={right indent={200pt}}, 
    reverse path,
    text effects/.cd, font=\small, 
      text along path, % add fill=white, yshift=-0.5ex  to .style below to put over line
      every character/.style={yshift=0.5ex}}}] plot [smooth cycle, tension=0.6] coordinates {(0.5,0.25) (3.5,0.4) (5.5,1)  (4,3.5) (2.5,4.45) (0.5,3.75) (0.5,2)};

\end{tikzpicture}
}
    \caption{Stylized illustration of possible differences between compliers in a trial population and compliers in a new target population. The shapes defined by solid lines are stylized and indicate regions of high covariate density rather than regions of covariate support. The group of compliers in the target population is the solid gray shaded area. }
    \label{fig:compliance-diagram}
\end{figure}
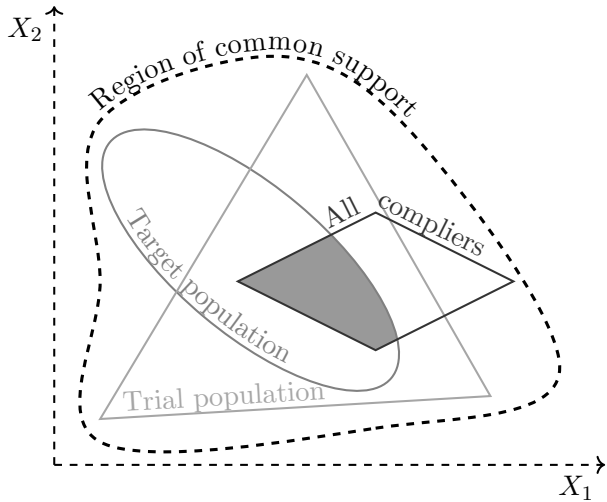

%%%%%%%%%%%%%%%%%%%%%%%%%%%%%%%%%%%%%%%%%%%%%%%%%%%%%%%%%%%%%%%%%%%%%%%%%%%%%%%%%%%%%%%%%%%%%%%%
%%%%%%%%%%%%%%%%%%%%%%%%%%%%%%%%%%%%%%%%%%%%%%%%%%%%%%%%%%%%%%%%%%%%%%%%%%%%%%%%%%%%%%%%%%%%%%%%
%%%%%%%%%%%%%%%%%%%%%%%%%%%%%%%%%%%%%%%%%%%%%%%%%%%%%%%%%%%%%%%%%%%%%%%%%%%%%%%%%%%%%%%%%%%%%%%%

\subsection{Standard Estimands with Principal Strata}
\label{sub: estimands_princ_strat}

We first define notation for data from an RCT alone, excluding any reference to target populations. From a single RCT, we observe i.i.d.~tuples $(Y_i, A_i, X_i, C_i)$ where $Y_i$, $A_i$, and $X_i$ refer to outcome, treatment assignment, and baseline covariates, respectively. We also observe an indicator variable $C_i$, where $C_i=1$ if participant $i$ received active treatment and $C_i=0$ if participant $i$ did not receive active treatment. For example, if $A_i=1$ and $C_i=1$, then participant $i$ adhered to their assigned active treatment. While we frame our work in the context of non-adherence, $C$ could represent other intercurrent events, e.g., survival in the case of truncation by death \citep{lyu_bayesian_2023}. 

We examine compliance through the lens of principal stratification \citep{frangakis_rubin_2002, bornkamp_principal_2021}, which places trial participants in distinct principal strata if they experience distinct intercurrent events, e.g., distinct compliance patterns. Even though we observe such events only after randomization, we conceptualize latent compliance patterns as existing \textit{a priori} and thus can treat them like potential outcomes.

%Just as $Y(a)$ denotes a participant's potential outcome under assignment to treatment $a$, 
Let $C(a)$ denote treatment received under assignment to treatment $A=a$. If $C_i(1)=1$, for instance, then if participant $i$ were assigned to active treatment, they would receive active treatment. Using both potential outcomes and principal strata notation, the average treatment effect among compliers is
\begin{equation}
E[Y(1)-Y(0)|C(1)=1, C(0)=0]. \label{eq: cace}
\end{equation}
A full definition of ``compliance" requires counterfactual knowledge of participant behavior under assignment to both placebo and active treatment. The obstacle to identifying these quantities is that strata defined by observed compliance $C_i$ are mixtures of strata defined by latent potential compliance $C_i(a)$. For example, compliance under assignment to control cannot be observed for patients assigned to treatment. To make progress, we need assumptions about the relationship between observed covariates $X$, observed compliance patterns, and latent compliance patterns. In this work, we apply monotonicity and principal ignorability to identify principal strata; Section \ref{identification} gives more detail.

\subsection{Novel Estimands among Principal Strata in Target Populations}
\label{sub: incorp_target_pop_data}

Now considering generalizability and transportability, we expand the observed data to include baseline covariates in a target population. The observed data then has the form $(R_iY_i, R_iA_i, R_iC_i,  X_i, R_i)$ where the indicator $R_i$ is 1 if individual $i$ is a member of the trial population and 0 if they are a member of the target population. Next, let $U=c_1c_0$ describe the population of individuals with $C_i(1)=c_1$ and $C_i(0)=c_0$, as in \citet{jiang_principal_2016}. Then the target causal quantity is
\begin{equation}
\tau^0_{10}=E[Y(1)-Y(0)|R=0, U=10], \label{eq: target_estimand_compliance}
\end{equation}
the expected treatment effect among target population members who are also compliers. Referring to Figure \ref{fig:compliance-diagram}, this estimand defines a new target population at the intersection of latent compliers with the larger target population. Because the degree of overlap with latent compliers may differ between the trial and target populations, treatment effects among compliers in the trial population may not be representative of compliers in the target population. Our framework addresses this problem directly.

\section{Identification of Principal Effects in the Target Population}
\label{identification}

\subsection{Identification Assumptions}
\label{sub: identification_assumptions}

Our causal quantity $\tau^0_{10}$ is defined in terms of two sets of possibly unobserved variables: potential outcomes $Y(1)$ and $Y(0)$ and potential compliance $C(1)$ and $C(0)$. Unlike the overall average treatment effect in the target population, $E[Y(1)-Y(0)|R=0]$, $\tau^0_{10}$ conditions on an unobserved population, the intersection of $R=0$ with compliers. These challenges require assumptions connecting these unobserved quantities to the observed data. Our assumptions combine those of principal stratification (e.g., Assumptions 1-3 of \citet{jiang_multiply_2022}) and generalizability (e.g., Assumptions 1-6 of \citet{degtiar_review_2023}), and introduce new conditions specific to our setting. We group our assumptions according to these different contexts.  While we focus on identifying effects in compliers, the assumptions below apply generally to any population defined by post-randomization variables $C(1)$ and $C(0)$. Whether such assumptions are appropriate for a given problem depends on the scientific context.

\begin{assumption}[Consistency] \label{consistency}
$Y_i=Y_i(1)A_i+Y_i(0)(1-A_i)$ and $C_i=C_i(1)A_i+C_i(0)(1-A_i)$.
\end{assumption}
\begin{assumption}[Treatment Ignorability] \label{treatment_ignorability} $A\indep (Y(1), Y(0), C(1), C(0))|X, R=1$.
\end{assumption}
Assumptions \ref{consistency} and \ref{treatment_ignorability} are typical in causal inference generally. An important condition embedded in Assumption \ref{consistency} is that participants do not directly affect each others' compliance behavior. Especially in policy-type interventions applied to a single community---as in the hotspotting example---we might expect violations of this assumption. Future work might build on existing literature on interference (e.g., \citet{hudgens_interference}) to incorporate compliance. Assumption \ref{treatment_ignorability} implies that $X$ contains enough information to remove dependence between treatment assignment and potential compliance. This holds by design in RCTs.
\begin{assumption}[Monotonocity] \label{monotonicity} $C(1)\geq C(0)$. \end{assumption}
\begin{assumption}[Principal Ignorability] \label{princignorability}
$E[Y(1)|U=10, R=1, X]=E[Y(1)|U=11, R=1, X]$ and $E[Y(0)|U=00, R=1, X]=E[Y(0)|U=10, R=1, X].$ 
\end{assumption}
Assumptions \ref{monotonicity} and \ref{princignorability} restate conditions common to principal stratification, e.g., Assumptions 2 and 3 in \citet{jiang_multiply_2022}. Principal ignorability is essential to identifying outcomes in a principal stratum when the observed data is a mixture of multiple strata. The first condition in Assumption \ref{princignorability} implies that the conditional mean of potential outcomes under treatment is identical for those with $C(1)=1$ and $C(0)=0$ and those with $C(1)=1$ and $C(0)=1$. Under Assumption \ref{consistency}, we know $C(1)=1$ for those in the treatment group who received treatment. Principal ignorability implies that we do not need to know those participants' behavior under control, i.e., $C(0)$, to identify their conditional average potential outcomes. This strong assumption invokes cross-world conditions applied simultaneously on $Y(a)$, $C(a)$, and $C(1-a)$ that are untestable from data. 

\begin{assumption}[Principal Stratum Exchangeability]\label{stratexch}
$R \indep (C(1), C(0))|X$.
\end{assumption} 
\begin{assumption}[Mean Exchangeability Among Principal Strata]\label{mean_exch}
$E[Y(a)|R=1, U=c_1c_0, X=x]=E[Y(a)|R=0, U=c_1c_0, X=x]$ for $a=1, 0$.
\end{assumption}

We introduce Assumptions \ref{stratexch} and \ref{mean_exch} to connect principal stratification to generalizability and transportability. Assumption \ref{mean_exch} implies the covariates $X$ are rich enough that learning about conditional average potential outcomes among compliers in the trial is as good as learning about such conditional outcomes among compliers in the target population. This extends similar  assumptions in transportability, e.g., Assumption 4 in \citet{dahabreh_extending_2020}. The key difference here is the inclusion of principal stratum membership $U=c_1c_0$. This assumption extends beyond compliance and applies to any potential outcome conditional on principal strata. Similarly, Assumption \ref{stratexch} implies that, conditional on $X$, the distribution of latent compliance patterns is identical in the trial and target populations. That is, $X$ can capture any relationship between latent compliance and the covariate patterns in $R=0$. Of course, this does not imply \textit{marginal} independence between $R$ and $C(1), C(0)$. It is precisely that the distribution of compliers may differ between the target and trial populations that motivates this work.

\subsection{Identification of Transported Principal Effects}
\label{identification_res}

This section gives identification results based on the foregoing assumptions, which connect $\tau^0_{10}$ to observed functions of the data. To define some such functions, let $p_a(X)=P(C=1|A=a, R=1, X)$, $p_a = E_{X}\left[p_a(X)\right]$, $\rho(X)=P(R=1|X)$, $\rho = E_X[P(R=1|X)]$, and $\mu_{ac}(X)=E[Y|A=a, C=c, R=1, X]$. We present three approaches to identifying $\tau^0_{10}$, distinguished in part by their reliance on separate nuisance parameters. Theorem \ref{ident_theorem} gives a simple plug-in identification as a function of $p_a(X)$, $\rho(X)$, and $\mu_{ac}(X)$. Theorems \ref{ipw_identification_thm} and \ref{om_identification_thm} give alternative approaches. Proofs of all theorems are in Section \ref{supp_identification} of the Supplementary Material.

\begin{theorem}[Plug-In Identification]
\label{ident_theorem}
Under Assumptions \ref{consistency} through \ref{mean_exch}, and assuming $\rho(x)>0$, $1-\rho(x)>0$, and $p_a(x)>0$ for $a=1, 0$ and all $x$ in the support of $X$,
\begin{equation}
E[Y(1)-Y(0)|U=10, R=0] = \frac{E\left[\left\{p_1(X)-p_0(X)\right\}\left\{1-\rho(X)\right\}\left(\mu_{11}(X)-\mu_{00}(X)\right)\right]}{E\left[\left\{p_1(X)-p_0(X)\right\}\left\{1-\rho(X)\right\}\right]}. \label{eq: identification_res}
\end{equation}
\end{theorem}

\noindent This result fits into a larger context of weighting-based methods that ``tilt" a given distribution of covariates toward a different distribution; see \citet{fan_li_tilting} for a broader discussion. Here, if we let $f(x)$ denote the marginal distribution of covariates across the trial and target populations, and $g(x)$ the distribution of covariates among compliers in the target population, we can show that $g(x)\propto f(x)(p_1(x)-p_0(x))(1-\rho(x))$. From \citet{jiang_principal_2016}, we know that $p_1(x)-p_0(x)$ identifies the conditional probability of compliance. Since $1-\rho(x)$ is the probability of membership in the target population, $(p_1(x)-p_0(x))(1-\rho(x))$ identifies the conditional probability of being a complier in the target population. Then, for example, the result in (\ref{eq: identification_res}) follows by taking the expectation of $\mu_{11}(X)-\mu_{00}(X)$ with respect to $g(x)$, including the normalizing constant $E[\{p_1(X)-p_0(X)\}\{1-\rho(X)\}]$. 

The two theorems below are alternative identification results relying on different sets of nuisance functions. Theorem \ref{ipw_identification_thm}'s identification result has a form similar to inverse probability weighting, whereas Theorem \ref{om_identification_thm}'s  identification result uses conditional mean outcome functions.

\begin{theorem}[IPW-Based Identification]
\label{ipw_identification_thm}
\textit{Under Assumptions \ref{consistency}-\ref{monotonicity}, and assuming $\rho(x)>0$, $1-\rho(x)>0$, and $p_a(x)>0$ for $a=1, 0$ and all $x$ in the support of $X$, we have
\begin{equation}
E[Y(1)|U=10, R=0] = \frac{1}{D}E\left[C\cdot A\cdot R \cdot \frac{p_1(X)-p_0(X)}{p_1(X)}\cdot\frac{1}{\pi(X)}\cdot \frac{1-\rho(X)}{\rho(X)}\cdot Y\right] \label{eq: ipw_identification_1}
\end{equation}
where $\pi(X)=P(A=1|R=1, X)$ and $D=E[(1-\rho(X))(p_1(X)-p_0(X))]$. Similarly, 
\begin{equation}
E[Y(0)|U=10, R=0] = \frac{1}{D}E\left[(1-C)\cdot (1-A)\cdot R \cdot \frac{p_1(X)-p_0(X)}{1-p_0(X)}\cdot\frac{1}{1-\pi(X)}\cdot \frac{1-\rho(X)}{\rho(X)}\cdot Y\right]. \label{eq: ipw_identification_0}
\end{equation}}
The contrast $\tau^0_{10}$ is identified by taking the difference of the above expressions.
\end{theorem}
\noindent The expressions (\ref{eq: ipw_identification_1}) and (\ref{eq: ipw_identification_0}) imply estimators that are simply weighted averages of the outcomes.
Each term in (\ref{eq: ipw_identification_1}) and (\ref{eq: ipw_identification_0}) serves a distinct purpose in re-weighting outcomes observed among compliers in either treatment arm of the study to resemble those of compliers in the target population. For instance, the weights in (\ref{eq: ipw_identification_1}) break down as follows: 
\begin{equation*}
\underbrace{\frac{p_1(X) - p_0(X)}{p_1(X)}}_{\parbox{4cm}{\center \scriptsize reweight $C=1$ group in arm $A=1$ to compliers}}\times\underbrace{\frac{1}{\pi(X)}}_{\parbox{2.25cm}{\center \scriptsize reweight $A=1$ to study sample}}\times\underbrace{\frac{1-\rho(X)}{\rho(X)}}_{\parbox{2.5cm}{\center \scriptsize transport RCT to target}}.
\end{equation*}
The next theorem uses the outcome model (OM) $\mu_{ac}(X)$ for identification.

\begin{theorem}[OM-Based Identification]
\label{om_identification_thm}
\textit{Under Assumptions \ref{consistency}-\ref{monotonicity}, and assuming $p_a(x)>0$ for $a=1, 0$ and all $x$ in the support of $X$, we have
\begin{equation}
E[Y(1)-Y(0)|U=10, R=0] = \frac{E\left[\{p_1(X)-p_0(X)\}(1-R)(\mu_{11}(X)-\mu_{00}(X))\right]}{E\left[\{p_1(X)-p_0(X)\}(1-R)\right]} .\label{eq: om_identification}
\end{equation}}
\end{theorem}

As is evident from Theorem \ref{om_identification_thm}, the principal stratification and generalizability setting resists straightforward application of, e.g., g-computation approaches \citep{what_if_causal_book}:  even if we had a correctly specified model for conditional mean outcomes among compliers, the distribution of $X$ over which we want to standardize such conditional outcomes is unidentified without further assumptions. Thus, identifying and using the principal scores is essential in \eqref{eq: om_identification}.

\section{Estimation and Efficiency Theory} \label{sec: estimation}

\subsection{Plug-in Estimation} \label{sub: plugins}

\noindent Each of Theorems \ref{ident_theorem}, \ref{ipw_identification_thm} and \ref{om_identification_thm} suggests  an estimator for $\tau^0_{10}$, which we denote $\hat{\tau}_{\text{Plug-In}}$, $\hat{\tau}_{\text{IPW}}$ and $\hat{\tau}_{\text{OM}}$, respectively, where additional superscripts and subscripts are omitted for simplicity:
\begin{align}
\hat{\tau}_{\text{Plug-In}} &= \frac{\mathbb{P}_n\left[\{\hat{p}_1(X)-\hat{p}_0(X)\}\{1-\hat{\rho}(X))\}\{\hat{\mu}_{11}(X)-\hat{\mu}_{00}(X)\}\right]}{\mathbb{P}_n\left[\{\hat{p}_1(X)-\hat{p}_0(X)\}\{1-\hat{\rho}(X)\}\right]}, \label{eq: plug_in_estimator}\\\nonumber\\
\hat{\tau}_{IPW} &= \frac{1}{\hat{D}}\mathbb{P}_n\left[C\cdot A\cdot R \cdot \frac{\hat{p}_1(X)-\hat{p}_0(X)}{\hat{p}_1(X)}\cdot\frac{1}{\hat{\pi}(X)}\cdot \frac{1-\hat{\rho}(X)}{\hat{\rho}(X)}\cdot Y\right] \label{eq: ipw_estimator}\\[8pt]
&\hspace{.75cm} - \frac{1}{\hat{D}}\mathbb{P}_n\left[(1-C)\cdot (1-A)\cdot R \cdot \frac{\hat{p}_1(X)-\hat{p}_0(X)}{1-\hat{p}_0(X)}\cdot\frac{1}{1-\hat{\pi}(X)}\cdot \frac{1-\hat{\rho}(X)}{\hat{\rho}(X)}\cdot Y\right], \text{ and} \nonumber\\ \nonumber\\
\hat{\tau}_{OM} &= \frac{\mathbb{P}_n\left[\{\hat{p}_1(X)-\hat{p}_0(X)\}(1-R)\{\hat{\mu}_{11}(X)-\hat{\mu}_{00}(X)\}\right]}{\mathbb{P}_n\left[\{\hat{p}_1(X)-\hat{p}_0(X)\}(1-R)\right]}, \label{eq: om_estimator}
\end{align}
where $\mathbb{P}_n(f(Z_i))$ denotes the sample average $\frac{1}{n}\sum_{i=1}^n f(Z_i)$ and $\hat{D} = \mathbb{P}_n\left[\{\hat{p}_1(X)-\hat{p}_0(X)\}(1-\hat{\rho}(X))\right]$. The consistency of each estimator relies on correctly specifying a particular set of nuisance parameters. It may be more desirable to construct a single estimator consistent across this range of misspecification scenarios. To do so, we first derive the efficient influence function (EIF) for the parameters identified in Theorem \ref{ident_theorem} and then use the EIF to construct robust, efficient estimators of the treatment effect among compliers in the target population.

\subsection{Efficient Influence Function} \label{sub: eif}

To construct estimators of $\tau^0_{10}$ with desirable efficiency properties, we derive the efficient influence function (EIF) for the parameters identified in Theorem \ref{ident_theorem} and then the EIF for $\tau^0_{10}$. The EIF is especially important in our case because we need to estimate several nuisance functions. Estimators based on the EIF have desirable statistical properties even when some of the nuisance functions are inconsistent or converge at rates slower than $\sqrt{n}$ \citep{kennedy_review}. 

Before presenting the EIF, we first define some additional quantities, all of which have analogues in \citet{jiang_multiply_2022}. These quantities are intermediate expressions used in the EIF for $\tau^0_{10}$ and are a bridge between our setting and the non-transportability setting of \citet{jiang_multiply_2022}.
\begin{align*}
&\psi_{f(Y_{a, r}, C_{a, r}, X)} \equiv \frac{I(A=a)I(R=r)\left[f(Y, C, X)-E\left\{f(Y, C, X)|X, A=a, R=r\right\}\right]}{P(A=a|R=1, X)P(R=r|X)}
\\&\hspace{3cm}+E\left\{f(Y, C, X)|X, A=a, R=r\right\},\\[2pt]
&\psi_{f(R)}=f(R)-E[f(R)|X],\\[6pt]
&\phi^0_{1, 10} \equiv \frac{e_{10}(X)}{p_1(X)}\left\{1-\rho(X)\right\}\psi_{Y_{1, 1}C_{1, 1}}-\mu_{11}(X)\left\{1-\rho(X)\right\}\left\{\psi_{C_{0, 1}}-\frac{p_0(X)}{p_1(X)}\psi_{C_{1, 1}}\right\}\\[2pt]
&\hspace{1.25cm}+e_{10}(X)\psi_{1-R}\mu_{11}(X), \\[6pt]
&\phi^0_{0, 10} \equiv \frac{e_{10}(X)}{1-p_0(X)}\left\{1-\rho(X)\right\}\psi_{Y_{0, 1}(1-C_{0, 1})}-\mu_{00}(X)\left\{1-\rho(X)\right\}\left\{\psi_{1-C_{1, 1}}-\psi_{1-C_{0, 1}}\left(\frac{1-p_1(X)}{1-p_0(X)}\right)\right\}\\[2pt]
&\hspace{1.25cm}+e_{10}(X)\psi_{1-R}\mu_{00}(X), \text{ and}\\
&\lambda_{10} \equiv \left[\psi_{C_{1, 1}}-\psi_{C_{0, 1}}\right]\left\{1-\rho(X)\right\}+e_{10}(X)\psi_{1-R},
\end{align*}
where $e_{10}(X)=p_1(X)-p_0(X)$. Although these expressions are complex, it can be shown that $\phi^0_{1, 10}$, $\phi^0_{0, 10}$, and $\lambda_{10}$ have familiar forms of uncentered influence functions, i.e., mean-zero residual terms plus a parameter of interest. For example, in Equation (\ref{phi_eqn}) in Section \ref{eif_deriv_supp} of the Supplement we show that the mean-zero centered EIF for $E\left[\left\{e_{10}(X)\right\}\left\{1-\rho(X)\right\}\mu_{11}(X)\right]$ is $\phi^0_{1, 10}-E\left[\left\{e_{10}(X)\right\}\left\{1-\rho(X)\right\}\mu_{11}(X)\right]$ so that $\phi^0_{1, 10}$ is centered at $E\left[\left\{e_{10}(X)\right\}\left\{1-\rho(X)\right\}\mu_{11}(X)\right]$. We can use such intermediate results to derive the overall EIF for more complex parameters of interest, namely
\begin{align*}
\psi^0_{10} &= \frac{E\left\{\left\{p_1(X)-p_0(X)\right\}\{1-\rho(X)\}\left[\mu_{11}(X)-\mu_{00}(X)\right]\right\}}{E\left\{\left\{p_1(X)-p_0(X)\right\}\{1-\rho(X)\}\right\}},\\[6pt]
\psi^0_{1, 10} &=  \frac{E\left[\left\{p_1(X)-p_0(X)\right\}\left\{1-\rho(X)\right\}\mu_{11}(X)\right]}{E\left[\left\{p_1(X)-p_0(X)\right\}\left\{1-\rho(X)\right\}\right]}, \text{ and}\\[6pt]
\psi^0_{0, 10} &= \frac{E\left[\left\{p_1(X)-p_0(X)\right\}\left\{1-\rho(X)\right\}\mu_{00}(X)\right]}{E\left[\left\{p_1(X)-p_0(X)\right\}\left\{1-\rho(X)\right\}\right]}.
\end{align*}
Note that $\tau^0_{10} = \psi^0_{10} = \psi^0_{1, 10} - \psi^0_{0, 10}$ under the assumptions in Section \ref{sub: identification_assumptions}.
Theorem \ref{eif_thm} references the above quantities and derives the EIF for our parameter of interest.  

\begin{theorem}[Derivation of Efficient Influence Functions]
\label{eif_thm}
The EIF for $\psi^0_{1, 10}$ is 
\begin{equation}
 \varphi_{\psi^0_{1, 10}}=\frac{\phi^0_{1, 10}}{E\left[e_{10}(X)\left\{1-\rho(X)\right\}\right]}-\psi^0_{1, 10}\left(\frac{\lambda_{10}}{E\left[e_{10}(X)\{1-\rho(X)\}\right]}\right). \label{eq: eif_mu_1}
\end{equation}
Similarly, the EIF for $\psi^0_{0, 10}$ is 
\begin{equation}
\varphi_{\psi^0_{0, 10}} = \frac{\phi^0_{0, 10}}{E\left[e_{10}(X)\left\{1-\rho(X)\right\}\right]}-\psi^0_{0, 10}\left(\frac{\lambda_{10}}{E\left[e_{10}(X)\{1-\rho(X)\}\right]}\right). \label{eq: eif_mu_0}
\end{equation}
Finally, the difference of (\ref{eq: eif_mu_1}) and (\ref{eq: eif_mu_0}) is the EIF for $\psi^0_{10}$, the identified treatment effect:
\begin{equation}
\varphi_{\psi^0_{10}}=\frac{\phi^0_{1, 10}-\phi^0_{0, 10}-\psi^0_{10}\lambda_{10}}{E\left[e_{10}(X)\left\{1-\rho(X)\right\}\right]}. \label{eq: eif_tau}
\end{equation}
\end{theorem}
\noindent The Supplementary Material gives a proof of these results. (Note that $\phi^0_{1, 10}$ and $\phi^0_{0, 10}$ play a role similar to analogous functions in \citet{jiang_multiply_2022} but are not identical to them.) As $\psi^0_{10}$, $\psi^0_{1, 10}$, and $\psi^0_{10}$ are all ratios, it is instructive to note that the forms of Equations (\ref{eq: eif_mu_1}), (\ref{eq: eif_mu_0}) and (\ref{eq: eif_tau}) all have a structure similar to EIFs of ratio parameters. As stated above, $\phi^0_{1, 10}$, $\phi^0_{0, 10}$, and $\lambda_{10}$ are themselves uncentered EIFs for various components of our ratio parameters. The quantities in (\ref{eq: eif_mu_1}), (\ref{eq: eif_mu_0}), and (\ref{eq: eif_tau}) then follow from direct application of theorems describing the form of EIFs for ratio parameters as a function of the EIFs for their constituent parts. Lemma \ref{eif_fraction} of the Supplementary Material gives such a result.

\subsection{EIF-Based Estimation and Asymptotic Results}
\label{sub: eif_estimation}

The EIF in (\ref{eq: eif_tau}) has expectation zero; applying this fact, we can solve for $\psi^0_{10}$ to obtain:

\begin{equation}
%\psi^0_{10}=\frac{E[\phi^0_{1, 10}-\phi^0_{0, 10}]}{E\left(\left[\psi_{C_{1, 1}}-\psi_{C_{0, 1}}\right]\left\{1-\rho(X)\right\}+e_{10}(X)\psi_{1-R}\right)}. \label{eq: est_eqn}
\psi^0_{10}=\frac{E[\phi^0_{1, 10}-\phi^0_{0, 10}]}{E\left[\lambda_{10}\right]}. \label{eq: est_eqn}
\end{equation}
One estimator for the treatment effect $\tau^0_{10}=E[Y(1)-Y(0)|R=0, U=10]$ follows from \ref{eq: est_eqn}  by setting $\mathbb{P}_n(\hat{\varphi}_{\psi^0_{10}})=0$ and solving for $\hat{\psi}^0_{10}$:
\begin{equation}
%\hat{\tau}_{\text{EIF}} = \frac{\mathbb{P}_n[\hat{\phi}^0_{1, 10}-\hat{\phi}^0_{0, 10}]}{\mathbb{P}_n\left(\left[\hat{\psi}_{C_{1, 1}}-\hat{\psi}_{C_{0, 1}}\right]\left\{1-\hat{\rho}(X)\right\}+\hat{e}_{10}(X)(\hat{\psi}_{1-R})\right)}. \label{eq: eif_estimator}
\hat{\tau}_{\text{EIF}} = \frac{\mathbb{P}_n[\hat{\phi}^0_{1, 10}-\hat{\phi}^0_{0, 10}]}{\mathbb{P}_n\left[\hat{\lambda}_{10}\right]}, \label{eq: eif_estimator}
\end{equation}
where $\hat{\phi}^0_{1, 10}$, $\hat{\phi}^0_{0, 10}$, and $\hat{\lambda}_{10}$ are plug-in estimators of their respective population quantities.

Estimator (\ref{eq: eif_estimator}) has desirable theoretical properties, summarized in Theorems \ref{consistency_theorem} and \ref{eif_asymp_theorem}. However, estimating both $\varphi_{\psi^0_{10}}$ and its constituent nuisance functions using the same sample requires further assumptions that limit the complexity of the functions' estimators and limits. \citet{dahabreh_biometrics} and \citet{jiang_multiply_2022} do this by assuming the nuisance functions and estimates are members of a Donsker class. We propose using sample splitting to avoid such assumptions. These techniques ensure that each individual's contribution to estimating $\varphi_{\psi^0_{10}}$ depends only on nuisance function estimates constructed using other individuals' data. Our sample-splitting procedure is analogous to the one in \citet{kennedy_sharp_2020}.

Like \citet{zeng_efficient_2023}, in our asymptotic results we assume for simplicity that the nuisance functions are estimated from a separate, independent sample. In practice, and in our simulation study, we use sample splitting with more than two folds. We establish $\hat{\tau}_{EIF}$'s asymptotic properties in two theorems. The first, Theorem \ref{consistency_theorem}, focuses on consistency of $\hat{\tau}_{EIF}$ under different combinations of nuisance function misspecification. This robustness of consistency alone is sometimes termed ``weak'' robustness because it provides no guarantee of the rate at which $\hat{\tau}_{EIF}$ converges to $\psi^0_{10}$ \citep{wager_notes}. It is, however, still helpful for understanding the properties of $\hat{\tau}_{EIF}$.

\begin{theorem}[Consistency Properties of EIF-Based Estimator]
\label{consistency_theorem}
Suppose estimated nuisance functions $\hat{\pi}(x)$, $1-\hat{\pi}(x)$, $\hat{p}_a(x)$, $1-\hat{p}_a(x)$, $\hat{\rho}(x)$, and $1-\hat{\rho}(x)$ are all greater than some $\delta>0$ for all $x$ in the support of $X$ and $a=0, 1$ for all $n$. Suppose further that each nuisance function $\hat{f}(x)\in \{\hat{\pi}(x), \hat{\rho}(x), \hat{p}_a(x), \hat{\mu}_{ac}(x): a=0,1; c=0,1\}$ converges in probability to a limit $\tilde{f}(x)$ for all $x$ in the support of $X$. If at least one of the following conditions hold for our nuisance function estimators, then $\hat{\tau}_{\text{EIF}}$ is consistent for $\psi^0_{10}$:
\begin{enumerate}
\item $\hat{\pi}(x)$, $\hat{\rho}(x)$, and $\hat{p}_a(x)$ are consistent for their true values for all $x$ in the support of $X$
\item $\hat{\pi}(x)$, $\hat{\rho}(x)$, and $\hat{\mu}_{ac}(x)$ are consistent for their true values for all $x$ in the support of $X$
\item $\hat{p}_a(x)$ and $\hat{\mu}_{ac}(x)$ are consistent for their true values for all $x$ in the support of $X$.
\end{enumerate}
\end{theorem}

\noindent Thus, three distinct conditions of correct specification of nuisance parameter models result in consistency of $\hat{\tau}_{EIF}$ for $\tau^0_{10}$, so one can derive a variety of estimators of $\tau^0_{10}$ by selecting which nuisance parameter models in $\hat{\tau}_{EIF}$ to set to 0.

Theorem \ref{eif_asymp_theorem} below gives conditions under which $\hat{\tau}_{EIF}$ is asymptotically normal and converges at a $\sqrt{n}$-rate:  $\sqrt{n}$ convergence of $\hat{\tau}_{EIF}$ depends on the \textit{product} of the bias of certain nuisance parameters which themselves converge at $\sqrt{n}$ rates. Dependence on the product of biases rather than the biases alone is sometimes called ``strong'' robustness \citep{wager_notes}.

\begin{theorem}[Asymptotic Properties of EIF-Based Estimator]
\label{eif_asymp_theorem}
Suppose nuisance functions ($\hat{\pi}$, $\hat{p}_a$, $\hat{\mu}_{ac}$, $\hat{\rho}$) are estimated from a separate, independent sample. Suppose further that $\hat{\pi}(x)$, $1-\hat{\pi}(x)$, $\hat{p}_a(x)$, $1-\hat{p}_a(x)$, $\hat{\rho}(x)$, and $1-\hat{\rho}(x)$ are all greater than some $\delta>0$ for all $x$ in the support of $X$ and $a=0, 1$ and that there exists $C\in \mathbb{R}$ such that $\{|\mu_{ac}(x)|, |\hat{\mu}_{ac}(x)|\}<C$ for all $x\in \mathcal{X}$, for each $n$. Further, suppose the following consistency conditions hold for the uncentered influence functions: $||\hat{\phi}^*_{1, 10}-\phi^0_{1, 10}||=o_p(1)$, $||\hat{\phi}^*_{0, 10}-\phi^0_{0, 10}||=o_p(1)$, and $||\hat{\lambda}_{10}-\lambda_{10}||=o_p(1)$,
%\begin{enumerate}
%\item \label{un_eif_consistent} $||\hat{\phi}^*_{1, 10}-\phi^0_{1, 10}||=o_p(1)$
%\item $||\hat{\phi}^*_{0, 10}-\phi^0_{0, 10}||=o_p(1)$
%\item $||[\hat{\psi}_{C_{1, 1}}-\hat{\psi}_{C_{0, 1}}]\{1-\hat{\rho(X)}\}+\hat{e}_{10}(X)\hat{\psi}_{1-R}-\left(\left[\psi_{C_{1, 1}}-\psi_{C_{0, 1}}\right]\left\{1-\rho(X)\right\}+e_{10}(X)\psi_{1-R}\right)||=o_p(1)$
%\item $||\hat{\lambda}_{10}-\lambda_{10}||=o_p(1)$,
%\end{enumerate}
where $||X||=\left\{E|X|^2\right\}^{1/2}$ for a random variable $X$. These three expressions are the uncentered EIFs for various components of our ratio parameter $\psi^0_{10}$.
Finally, let
\begin{align*}
R^1_n &= ||\mu_{ac}-\hat{\mu}_{ac}||\left\{||p_a-\hat{p}_a||+||p_{1-a}-\hat{p}_{1-a}||\right\},\\
R^2_n &= ||\mu_{ac}-\hat{\mu}_{ac}||\left\{||\pi-\hat{\pi}||+||\rho-\hat{\rho}||\right\}, \text{ and }\\
R^3_n &= \left\{||p_a-\hat{p}_a||+||p_{1-a}-\hat{p}_{1-a}||\right\}\left\{||\pi-\hat{\pi}||+||\rho-\hat{\rho}||\right\}
\end{align*} for $a=0, 1$. Then
\begin{equation}
\hat{\tau}_{\text{EIF}}-\tau^0_{10}=\mathbb{P}_n\left(\varphi_{\psi^0_{10}}\right)+O_p\left(R^1_n+R^2_n+R^3_n\right)+o_p(1/\sqrt{n}).
\end{equation}
If, further, $R^1_n=o(1/\sqrt{n})$, $R^2_n=o(1/\sqrt{n})$, and $R^3_n=o(1/\sqrt{n})$, then  
\begin{equation}
\hat{\tau}_{\text{EIF}}-\tau^0_{10}=\mathbb{P}_n\left(\varphi_{\psi^0_{10}}\right)+o_p(1/\sqrt{n})
\label{eq: eif_asymp_expansion}
\end{equation}
and, therefore, the EIF-Based Estimator in (\ref{eq: eif_estimator}) is $\sqrt{n}$-consistent, asymptotically normal, and its asymptotic variance achieves the non-parametric efficiency bound.
\end{theorem}

\noindent Theorem \ref{eif_asymp_theorem} gives the robustness properties of $\hat{\tau}_{\text{EIF}}$, showing that its consistency generally and $\sqrt{n}$-convergence specifically hold under a variety of asymptotic behaviors of the nuisance functions, including misspecification. We demonstrate these properties below in simulation studies. 

The asymptotic properties of estimators relying on two datasets depend on further assumptions about their sample sizes. \citet{zhang_double_2023} give a clarifying discussion of these issues in semi-supervised inference. Their work is especially helpful because it makes clear that, for semi-supervised inference, the random variable $R$ is simply a labeling convention to discriminate between samples of one source versus another. This framing is also useful in the context of generalizability and transportability. There, too, $R$ is not a fundamental quantity that distinguishes instrumentally who \textit{could} and \textit{could not} participate in a trial; instead, it is just a mechanism to define certain distributions of (possibly unobserved) covariates as distinguished between the trial and target samples. \citet{colnet_causal_2024} also discuss this idea by distinguishing $R$ from a variable $S$ that refers to being \textit{sampled} into the trial.

In proving Theorem \ref{eif_asymp_theorem}, we made the simplifying assumption that the ratio of the trial sample size, $n_1=\sum_{i=1}^n R_i$, to the target sample size, $n_0=\sum_{i=1}^n (1-R_i)$, converges to a constant $c$ as these sizes grow to infinity. Our simulation studies explore our estimator's properties for different values of $c$.

For inference, Theorem \ref{eif_asymp_theorem} implies that the asymptotic variance of $\hat{\tau}_{\text{EIF}}$ is $E\left[\left\{\varphi_{\psi^0_{10}}\right\}^2\right]$. We can estimate this quantity using the sample average of the square of the estimated centered EIF. An alternative approach---applicable to $\hat{\tau}_{\text{EIF}}$, $\hat{\tau}_{\text{IPW}}$, and $\hat{\tau}_{\text{OM}}$---is a non-parametric bootstrap.

\section{Simulation Study} \label{sec: simulation}

The simulation study illustrates the asymptotic properties in Theorem \ref{eif_asymp_theorem} and investigates efficiency trade-offs when using $\hat{\tau}_{\text{EIF}}$ in relatively small samples. To facilitate comparison with the motivating work of \citet{jiang_multiply_2022}, we build on their simulation setup with new features specific to our case.

\subsection{Simulation setup}

Table \ref{table: sim_settings} summarizes the main simulation settings, where TP ($\pi$) is the treatment probability model, OM ($\mu_{ac}$) is the outcome model, PP ($\rho$) is the participation probability model, and PS ($p_{a}$) is the principal score model.

\begin{table}[ht]
\centering
\begin{tabular}{c c}
Parameter & Possible Settings\\ \midrule \\
Ratio of Trial Sample Size to Total Sample Size & (1/21, 1/3, 1/2)\\ \\
Trial Sample Size & (500, 1000, 2000, 50000)\\ \\
\multirow{5}{10em}{Model Misspecification} & TP, OM, PP, and PS Correctly Specified \\
& TP, PP, and PS Correctly Specified\\
& PS and OM Correctly Specified\\
& TP, OM, and PP Correctly Specified\\
& No Models Correctly Specified\\
\end{tabular}
\caption{Parameters Varied in the Simulation Study}
\label{table: sim_settings}
\end{table}

For each of the 12 possible total sample sizes $n_{total}$, we simulate $n_{total}$ draws from the joint distribution of covariates. Each draw includes five independent covariates, $X_1$, \dots, $X_5$, where $X_1, \dots, X_4$ $~\sim~\text{Unif}(-2, 2)$ and $X_5~\sim~\text{Bernoulli}(0.55)$. We follow \citet{jiang_multiply_2022} in constructing two sets of derived variables to define model misspecification:  $\tilde{C}_1, \dots, \tilde{C}_4$ where $\tilde{C}_i=(X_i^2-1)/\sqrt{2}$ and $C_1, \dots, C_4$ where $C_i=X_i-0.25$. 

As noted after Theorem \ref{eif_asymp_theorem}, the asymptotic results hold when the relative sample sizes of the trial and target samples converge to a constant. To simulate data in a way that reflects this assumption, we fixed the ratio of the trial to total sample sizes at one of $1/21, 1/3$, or $1/2$. After simulating covariates marginally as described above, we assign them to either the treatment or target sample according to the model $P(R=1|X)~=~\text{expit}\left(\beta_{0}+\tilde{C}_1+\tilde{C}_2+\tilde{C}_3\right)$, where $\beta_0$ is chosen to reflect the desired marginal proportion using the method outlined in \citet{balancing_intercept}.

Similar to \citet{jiang_multiply_2022}, we specified nuisance functions as follows:
\begin{align*}
p_a(X)&=\text{expit}\left((2/5)\times(2a-(a-1)\times(\tilde{C}_1+\tilde{C}_2+\tilde{C}_3+\tilde{C}_4)\right),\\
\pi(X) &= \text{expit}\left((2/5)\times (\tilde{C}_1+\tilde{C}_2+\tilde{C}_3+\tilde{C}_4)\right), \text{ and}\\
\mu_{ac}(X) &= ((1+a+c)/4)(\tilde{C}_1+\tilde{C}_2+\tilde{C}_3+\tilde{C}_4+X_5).
\end{align*}
We induced model misspecification in our simulation by erroneously substituting $C_i$ for $\tilde{C}_i$ as applicable when fitting each regression function.

\subsection{Results}

We compared the performance of three estimators: an OM-based estimator motivated by Theorem \ref{om_identification_thm}, an IPW-based estimator motivated by Theorem \ref{ipw_identification_thm}, and, finally, the EIF-based estimator  in Equation (\ref{eq: eif_estimator}). Here, we use 10-fold sample splitting to construct the EIF-based estimator. For both the EIF- and IPW-based estimator, we normalize weights to give H\'ajek-style estimators. See, e.g., \citet{chattopadhyay_balancing_2020} for more information on normalization. Results presented in this section and in the Supplementary Material are averaged over 1000 simulation runs for each simulation setting. The plots below show each estimator's bias for different degrees of misspecification. In Figure \ref{fig:sim_study_res}, we present results only for the simulation setting where the trial to overall sample size ratio is 1/21.

\begin{figure}[ht]
    \includegraphics[width=\textwidth]{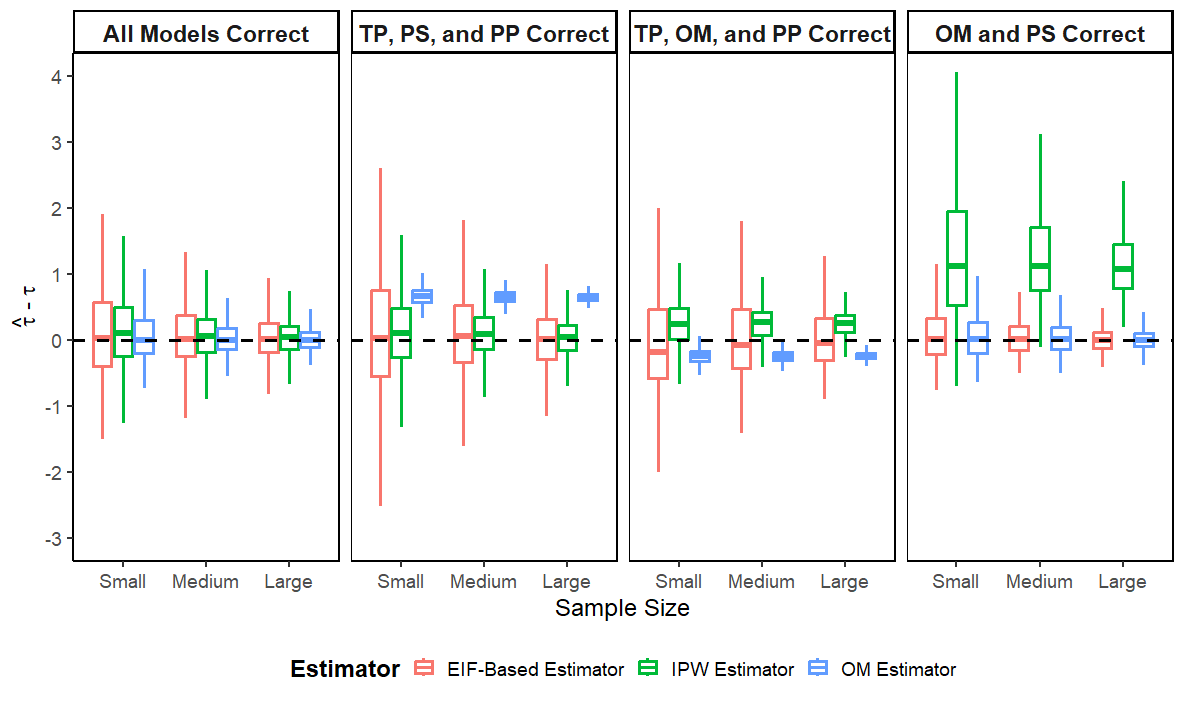}
    \caption{Selected bias results from simulation study. Each panel summarizes the error in the EIF-based, IPW, and OM estimators across 1000 simulation runs under different misspecification regimes.}
    \label{fig:sim_study_res}
\end{figure}

In the following, ``small" refers to simulations where the trial sample size is 500, ``medium" where the trial sample size is 1000, and ``large" where the trial sample size is 2000. Since the ratio of the trial to total sample size is 1/21 for the results in Figure \ref{fig:sim_study_res}, the respective target sample sizes are 10,000, 20,000, and 40,000. This is meant to mimic a scenario in which a very large EHR-derived database is coupled with data from a relatively large trial for the purpose of transportability. 

The key takeaway from Figure \ref{fig:sim_study_res} is that the EIF-based estimator is consistent in all model misspecification scenarios. As expected given their respective theorems, the OM- and IPW-based estimators are asymptotically unbiased in certain scenarios but not all. Supplementary Section \ref{sub: sim_res} presents full bias and root mean square error (RMSE) results. Our findings echo those in Kang and Schafer's 2007 work on so-called doubly-robust estimators \citep{kang_demystifying_2007}. At small to moderate sample sizes, the high variability of the many weights in the EIF-based estimators dramatically increases RMSE relative to a simpler estimator with fewer estimated weights. This variability diminishes at larger sample sizes, consistent with Theorem \ref{eif_asymp_theorem}'s asymptotic results. We note, however, that in cases where the IPW and OM estimators are consistent, in our simulation setting they are correctly specified parametric models, which may have lower variance than a {nonparametrically} efficient estimator.

Supplementary Section \ref{sub: sim_res} gives results for coverage of 95\% confidence intervals based on Theorem \ref{eif_asymp_theorem}. The standard error used to construct such intervals is computed using the empirical variance of the estimated EIF. As shown, intervals constructed  this way have relatively poor coverage, especially when the treatment probability and participation probability models are misspecifed. We suspect the poor coverage results from practical positivity violations of the kind explored in \citet{robust_var_tmle}. We hope that future work will develop more targeted variance estimators.

\section{Analysis of Hotspotting Data}
\label{sec: hotspotting_analysis}

We illustrate our methods by transporting complier average treatment effects from participants in the ``Health Care Hotspotting" RCT to patients in a Midwestern U.S. academic health center. As described earlier, the Hotspotting RCT randomly assigned ``superutilizer" patients with high healthcare needs to either standard of care or targeted, intensive follow-up after an inpatient admission. While the original trial resulted in a null finding, a recent secondary analysis among ``high engagers" found statistically significant beneficial intervention effects \citep{yang_hospital_2023}. Our analysis aims to transport such ``high engager" effects from the hotspotting RCT population to patients in a separate health system.

\subsection{Estimand Construction for the Hotspotting Intervention}

Our methods can target a variety of treatment effects at the intersection of a target population and well-defined principal strata. The running example throughout this work has been treatment effects among ``compliers," i.e., patients who would take their assigned treatment in both the control and intervention arms. However, the hotspotting RCT secondary analysis focused only on compliance behavior in the intervention arm, as patients assigned to control were unable to participate in hotspotting. We follow this approach and define our estimand of interest as the average treatment effect among members of the target population who would be high-engagers in the intervention arm. We use ``engagement" in lieu of ``compliance" both to distinguish the hotspotting trial from traditional pharmaceutical RCTs and to more accurately describe the treatment received in complex, community-based healthcare interventions. 

The relevant estimand is therefore similar to that in \citet{qu_general_2020} for effects among patients who adhere to active intervention alone. Such an estimand conditions only on $C(1)=1$ and averages over different adherence behavior in the control group. We assume in the hotspotting data that ``superutilizers" assigned to control have no access to the active intervention. This assumption has been termed ``strong" monotonicity or one-sided noncompliance (\citet{Imbens_Rubin_2015, jiang_multiply_2022}). Under strong monotonocity, conditioning on $C(1)=1$ is equivalent to conditioning on both $C(1)=1$ and $C(0)=0$. That is,
\begin{equation}
E[Y(1)-Y(0)|C(1)=1] = E[Y(1)-Y(0)|C(1)=1, C(0)=0]. \label{eq: ate_equiv}
\end{equation}

This implies in particular that the observed stratum of participants with $A=1$ and $C=1$ coincide with the unobserved strata of compliers, given that there is no variability in those participants' adherence behavior under counterfactual assignment to control. Therefore, while we define our estimand for the hotspotting analysis as
\begin{equation*}
\tau^0_{1} = E[Y(1)-Y(0)|R=0, C(1)=1]
\end{equation*}
we know by (\ref{eq: ate_equiv}) that $\tau^0_{1}=\tau^0_{10}$ defined in (\ref{eq: target_estimand_compliance}). The same identification formulas applicable to $\tau^0_{10}$ are thus applicable to $\tau^0_{1}$ with the simplification that $
e_{10}(X) = p_{1}(X)$.
The primary outcome in the Hotspotting RCT was hospital readmission 180 days following an index admission. The effect defined in (\ref{eq: ate_equiv}) therefore corresponds to a risk difference in percentage points.

\subsection{Identification Assumptions for Hotspotting Data}

Under one-sided noncompliance, some groups defined by observed $A$ and $C$ are no longer mixtures of latent strata. We summarize this in the table below, adapted from \citet{jiang_multiply_2022}:

\begin{table}[h!]
\centering
\begin{tabular}{c|c|c}
 & $C=0$ & $C=1$\\
    \hline
    $A=0$ & $C(0)=0, C(1)\in (1, 0)$ & N/A\\
    \hline
     $A=1$ & $C(1)=0, C(0)=0 $ & $C(1)=1, C(0)=0$
\end{tabular}
\caption{Relationship between potential and observed compliance under strong monotonicity.}
\label{table: one_sided_mixtures}
\end{table}
\noindent The absence of the principal stratum $U=11$ obviates the need for principal ignorability assumptions on mean potential outcomes under assignment to treatment. In fact, several identification assumptions needed more generally can be significantly weakened in this setting. The weaker variants of our original Assumptions \ref{consistency}-\ref{mean_exch} are:

\begin{weakassumption}{monotonicity}
$C(0)=0$
\end{weakassumption}

\begin{weakassumption}{princignorability}
$E[Y(0)|U=00, R=1, X]=E[Y(0)|U=10, R=1, X]$, and
\end{weakassumption}

\begin{weakassumption}{stratexch}
$R \indep C(1)|X$.
\end{weakassumption}

Assumptions \ref{monotonicity}$^*$-\ref{stratexch}$^*$ replace Assumptions \ref{monotonicity}-\ref{stratexch} and these weaker assumptions allow alternative identification results specific to one-sided noncompliance. Lemma \ref{one_sided_lemma} is one such result. Other identification results can be proven similarly. The EIFs for these identified parameters also follow from our more general results by letting $p_0(X)=0$ wherever that nuisance function appears.

\subsection{Estimates of Generalized Principal Causal Effects for Hotspotting Data}
\label{hotspotting_estimates_section}

This section presents analysis results for three estimation procedures, which target different estimands. First, we apply our EIF-based estimator to estimate the effect of hotspotting on 180-day readmission rates among engagers in the target population. As a baseline, we also apply the EIF-based estimator of \citet{jiang_multiply_2022} to estimate that same effect in the trial population. Finally, we use the distillation method as applied in \citet{yang_hospital_2023}. The distillation estimator does not correspond to a single causal quantity but can be roughly interpreted as ``the effect of hotspotting among those most likely to engage." To avoid problems induced by a small number of individuals with extreme participation probabilities, the combined trial and target dataset was trimmed to remove cases with $P(R=1|X)$ estimates below 0.10 and above 0.90 \citep{trimming_article}. Section \ref{descriptive_stats_hotspotting} provides illustrations of overlap by participation probability and a descriptive comparison of the trial and target samples.

Nuisance functions in the EIF-based methods were estimated using the SuperLearner algorithm applied to penalized logistic regression and random forest models \citep{super_learner}. When applying distillation, we used the same gradient-boosting machine algorithm as \citet{yang_hospital_2023} to estimate principal scores and construct the distilled sample of high-engagers. However, the EIF-based estimators---both transported and untransported---target average risk differences, whereas distillation in \citet{yang_hospital_2023} quantifies the treatment effect of hotspotting as a conditional odds ratio from a logistic regression model. We instead applied g-computation with logistic regression to estimate marginal risk differences in the distilled population and used the delta method to obtain confidence intervals. This was implemented using the \texttt{marginaleffects} package in R \citep{marginal_effects_preprint}. Confidence intervals from the EIF-based methods were also based on asymptotic approximations summarized in Theorem \ref{eif_asymp_theorem} of this work and Theorem 4 of \citet{jiang_multiply_2022}. 

\begin{table}[ht]
\centering
\resizebox{\textwidth}{!}{%
\begin{tabular}{lrrr}
  \toprule
 \textbf{Population} & \textbf{180-Day TE} & \textbf{90-Day TE} & \textbf{30-Day TE}\\
 \midrule 
  Engagers in the Target Population & 8.73 (-3.45, 20.91) & 7.68 (-4.92, 20.29) & -10.50 (-21.26, 0.25) \\ \\
  Engagers in the Trial Population & 6.94 (0.48, 13.40) & 4.93 (-1.78, 11.64) & -4.73 (-10.60, 1.13) \\ \\
  Distilled Sample of the Trial Population & 6.73 (-0.73, 14.18) & 4.98 (-2.72, 12.69) & 0.31 (-6.56, 7.18)\\ \bottomrule
\end{tabular}%
}
\caption{Estimates of the treatment effect (TE) of hotspotting on readmission and corresponding 95\% confidence intervals. Estimates are risk differences with positive values indicating an increase in the probability of readmission.}
\label{tab:hotspotting_results}
\end{table}

The results are in Table \ref{tab:hotspotting_results}. While most of the estimated risk differences are positive---indicating a negative effect of hotspotting among engagers---few are estimated with sufficient precision to reject the null hypothesis of no effect. We suspect that differences between our application of distillation and the results of \citet{yang_hospital_2023} are driven by differences in our definition of engagement, which cannot be mimicked exactly in the public-use dataset \citet{hotspotting_dataset} used in this work. Another important difference is our targeting of the marginal effect on the risk difference scale versus the conditional odds ratio as in \citet{yang_hospital_2023}. Such marginal effects are expected to differ from conditional effects \citep{conditional_vs_marginal}. We do note that all of the estimated effects show increases in benefit (or decrease in harm) for earlier-assessed timepoints among engagers in the trial and target populations. This echoes results shown in Figure 2 of  \citet{yang_hospital_2023}.

Another important challenge in transportability problems is the potential for data mismatch between the trial and target population samples. The methods presented here require all covariates $X$ be included in both data sources. This is an important limitation for this analysis, as many of the rich demographic characteristics in the hotspotting dataset are unavailable in the EHR data from the target population. These unavailable variables include many of the covariates found to be most strongly related to engagement in \citet{yang_hospital_2023}, including arrest history, family and social support, and housing status. Various methods have been developed to address such mismatch in transportability, including \citet{steingrimsson_missing}, \citet{rudolph_efficient_2024}, and \citet{zeng_efficient_2023}. The latter two in particular may help to improve precision by allowing nuisance function estimation using all covariates available in the trial data. An important avenue for future research will be applying insights from these works to our setting.

\section{Discussion}
\label{sec: discussion}

This paper developed a causal framework for transporting treatment effects from randomized trials to target populations when those trials are subject to non-adherence. The principal stratification framework of \citet{frangakis_rubin_2002} allows us to specify a well-defined population---the intersection of latent compliers with the target population---amenable to generalization and transportation. This approach also makes it possible to transport results to populations with latent adherence patterns differing from those who would comply under both treatment and placebo. For instance, \citet{qu_general_2020} define causal estimands among those who would adhere under assignment to experimental treatment, marginalizing over varying counterfactual adherence under assignment to placebo control. 

We contrast transportability of principal stratum-specific treatment effects with transportability of the complier average causal effect (CACE) under instrumental variable (IV) assumptions \citep{rudolph_robust_2017}. The latter requires the exclusion restriction (ER) assumption common to IV settings which, as discussed by \citet{jiang_multiply_2022}, can be replaced by the principal ignorability conditions in Assumption \ref{princignorability}. Whether the ER or principal ignorability assumptions are more plausible is situation-dependent. Thus, our work adds to \citet{rudolph_robust_2017} an alternative framework that investigators can use in  transporting complier average effects depending on which assumption fits their case better. Moreover, the assumptions and identification strategies used here are relevant to causal estimands in populations defined by principal strata corresponding to any intercurrent event, and need not be directly related to adherence. 

This flexibility in both estimands and assumptions can be extended to other data types, e.g., a set of clinical trials as in meta-analysis. For example, \citet{zhou_cace_meta} provide methods to estimate the CACE in a meta-analysis under the exclusion restriction. An important challenge to such methods is the extent to which the population of compliers (and, indeed, participants) varies across trials. Future work may extend our approach---in conjunction with recent advances in causally-interpretable meta-analysis by \citet{dahabreh_biometrics}---to address these challenges by shifting inferential focus to a single population of compliers in the target population.

While theoretically grounded estimates derived from our parameter's EIF may be unbiased and efficient asymptotically, high variability driven by its component estimators may induce poor behavior in small samples. Future work could tackle these estimation problems, for example, we might avoid nuisance function estimation by using a balancing approach, where trial participant outcomes are weighted in terms of their similarity to high-probability compliers in the target population \citep{chattopadhyay_balancing_2020}. Alternatively, approaches based on targeted maximum likelihood estimation have asymptotic properties similar to the standard EIF-based estimator's while potentially improving small-sample performance \citep{tmle_book}.  Our work defines causal quantities and gives identification assumptions necessary to connect those quantities to observed data; future research will build on the estimators presented here to make these methods yet more relevant to practical, data-driven challenges faced by clinicians and policy-makers.

\section{Acknowledgements}

The authors would like to thank Dawn Wiest and Aaron Truchil of the Camden Coalition for helpful discussions regarding the hotspotting data. The analyses and conclusions in this paper do not necessarily reflect the views of the Camden Coalition or its staff. Compilation of the EHR data in Section \ref{sec: hotspotting_analysis} was supported by the University of Minnesota’s NIH Clinical and Translational Science Award: UM1TR004405; we especially thank Talia Wiggen of the University of Minnesota Clinical and Translational Science Institute's Best Practices Integrated Informatics Core for her contributions to that effort.

\ifSubfilesClassLoaded{% <<<<<<<<<<<<<<<
  \bibliography{bibliography}% <<<<<<<<<<<<<<
}{}

\end{document}
\bibliography{bibliography}
\newpage

\setcounter{section}{0}
\renewcommand*{\theHsection}{S.\the\value{section}}
\renewcommand{\thesection}{S.\arabic{section}}

\setcounter{page}{1}
\renewcommand*{\thepage}{S\arabic{page}}

\maketitle

\section{Supplementary Material Organization}

The proofs in this section are relatively lengthy, and we first provide a brief roadmap:

\begin{enumerate}
\item Section \ref{supp_identification} gives proofs for Theorems \ref{ident_theorem}, \ref{ipw_identification_thm}, and \ref{om_identification_thm}. These apply various assumptions to express our causal parameter $\tau^0_{10}$ as a function of the observed data.
\item Section \ref{supp_eif_derivation} derives the EIF for the identified parameter $\psi^0_{10}$ as expressed in (\ref{eq: identification_res}). This quantity is important because it defines the minimum variance achievable among regular, asymptotically linear estimators for $\psi^0_{10}$.
\item Section \ref{consistency_proofs} illustrates the consistency results for our EIF-based estimator of $\psi^0_{10}$ summarized in Theorem \ref{consistency_theorem}.
\item Section \ref{ral_proofs} characterizes the asymptotic properties laid out in Theorem \ref{eif_asymp_theorem}. In particular, we show that $\hat{\tau}_{EIF}$'s influence function coincides with the efficient influence function derived in Section \ref{supp_eif_derivation}.
\item Section \ref{one_sided_identification} identifies generalized treatment effects among compliers under one-sided noncompliance. This setting is relevant to our motivating Hotspotting RCT data analysis.
\item Section \ref{sub: sim_res} gives more detailed results from our simulation experiments.
\end{enumerate}

\newpage

\section{Identification}
\label{supp_identification}

\subsection{Proof of Theorem \ref{ident_theorem}}
\label{plug_in_identification_proof}

Here, we want to show that
\begin{equation*}
\tau^0_{10} = \frac{E\left\{\left[p_1(X)-p_0(X)\right](1-\rho(X))\left[\mu_{11}(X)-\mu_{00}(X)\right]\right\}}{E\left\{\left[p_1(X)-p_0(X)\right](1-\rho(X))\right\}}
\end{equation*}
To start, we prove several lemmas, each of which correspond to identification of separate components of $\tau^0_{10}$. These proofs draw on results from \citet{jiang_multiply_2022} and our own identification assumptions:
\begin{lemma}
\label{denom_identification}
Under our identification assumptions, we identify the marginal probability of compliance in the target population as:
\begin{equation*}
P(U=10, R=0)=E[(p_1(X)-p_0(X))(1-\rho(X))].
\end{equation*}
\end{lemma}
\begin{proof}[Proof of Lemma \ref{denom_identification}]
We have that
\begin{align*}
P(U=10, R=0) &= E[I(U=10, R=0)]\\
&= E[E\{I(U=10, R=0)|X\}]\\
&= E[E\{I(U=10)|X\}E\{I(R=0)|X\}] & \text{by Assumption \ref{stratexch}}\\
&= E[P(U=10|X)(1-\rho(X)]\\
&= E[(p_1(X)-p_0(X))(1-\rho(X))]. & \text{by Assumptions \ref{treatment_ignorability} and \ref{monotonicity}}
\end{align*}
\end{proof}
\begin{lemma}
\label{te_1_identification}
Under our identification and positivity assumptions, we identify the outcome under assignment to treatment among compliers in the target population as:
\begin{equation*}
E[Y(1)|U=10, R=0] = \frac{E[\mu_{11}(X)(p_1(X)-p_0(X))(1-\rho(X))]}{E[(p_1(X)-p_0(X))(1-\rho(X))]}
\end{equation*}
\end{lemma}
\begin{proof}[Proof of Lemma \ref{te_1_identification}]
We then apply Lemma \ref{denom_identification} to identify $E[Y(1)|U=10, R=0]$, the expected potential outcome under assignment to treatment among compliers in the target population. Again, many of the following steps draw directly from proofs given in \citet{jiang_multiply_2022}.
\begin{align*}
& E[Y(1)|U=10, R=0] \\ &= E[E\{Y(1)|X, U=10, R=0\}|U=10, R=0]\\
&= E[E\{Y(1)|X, U=10, R=1\}|U=10, R=0] & \text{by Assumption \ref{mean_exch}}\\
&= E[E\{Y(1)|X, U=10 \text{ or } U=11, R=1\}|U=10 , R=0] & \text{by Assumption \ref{princignorability}}\\
&= E[E\{Y(1)|X, A=1, U=10 \text{ or } U=11, R=1\}|U=10 , R=0] & \text{by Assumption \ref{treatment_ignorability}}\\
&= E[E\{Y(1)|X, C=1, A=1, R=1\}|U=10, R=0] & \text{by Assumption \ref{consistency}}\\
&= E[E\{Y|X, C=1, A=1, R=1\}|U=10, R=0] & \text{by Assumption \ref{consistency}}\\
&= E[\mu_{11}(X)|U=10, R=0] & \text{by definition}\\
&= \frac{E[I(U=10, R=0)\mu_{11}(X)]}{P(U=10, R=0)}\\
&= \frac{E[I(U=10, R=0)\mu_{11}(X)]}{E[(p_1(X)-p_0(X))(1-\rho(X))]} & \text{by Lemma \ref{denom_identification}}\\
&= \frac{E[\mu_{11}(X)E\{I(U=10, R=0)|X\}}{E[(p_1(X)-p_0(X))(1-\rho(X))]}\\
&= \frac{E[\mu_{11}(X)(p_1(X)-p_0(X))(1-\rho(X))]}{E[(p_1(X)-p_0(X))(1-\rho(X))]},
\end{align*}
where the last equality follows from identical reasoning as that given in Lemma \ref{denom_identification}.
\end{proof}
\begin{lemma}
\label{te_0_identification}
Under our identification and positivity assumptions, we identify the outcome under assignment to control among compliers in the target population as:
\begin{equation*}
E[Y(0)|U=10, R=0] = \frac{E[\mu_{00}(X)(p_1(X)-p_0(X))(1-\rho(X))]}{E[(p_1(X)-p_0(X))(1-\rho(X))]}
\end{equation*}
\end{lemma}
\begin{proof}[Proof of Lemma \ref{te_0_identification}]
The reasoning in this case is almost identical as that of Lemma \ref{te_1_identification}.
\end{proof}

\begin{proof}[Proof of Theorem \ref{ident_theorem}]
That 
\begin{equation*}
\tau^0_{10} = \frac{E\left\{\left[p_1(X)-p_0(X)\right](1-\rho(X))\left[\mu_{11}(X)-\mu_{00}(X)\right]\right\}}{E\left\{\left[p_1(X)-p_0(X)\right](1-\rho(X))\right\}}
\end{equation*}
follows immediately from Lemmas \ref{denom_identification}, \ref{te_1_identification}, and \ref{te_0_identification}. We have:
\begin{align*}
\tau^0_{10} &= E[Y(1)-Y(0)|R=0, U=10]\\
&= E[Y(1)|R=0, U=10]-E[Y(0)|R=0, U=10]\\
&= \frac{E[\mu_{11}(X)(p_1(X)-p_0(X))(1-\rho(X))]}{E[(p_1(X)-p_0(X))(1-\rho(X))]}-\frac{E[\mu_{00}(X)(p_1(X)-p_0(X))(1-\rho(X))]}{E[(p_1(X)-p_0(X))(1-\rho(X))]}\\
&= \frac{E\left\{\left[p_1(X)-p_0(X)\right](1-\rho(X))\left[\mu_{11}(X)-\mu_{00}(X)\right]\right\}}{E\left\{\left[p_1(X)-p_0(X)\right](1-\rho(X))\right\}}
\end{align*}
where the penultimate equality follows from the above Lemmas.
\end{proof}

\subsection{Proof of Theorems \ref{ipw_identification_thm} and \ref{om_identification_thm}}
\label{ipw_identification_proof}
Here, we proceed similarly as the proof of Theorem \ref{ident_theorem}: first giving results for each component of the causal contrast, then combining such lemmas to identify $\tau^0_{10}$. 

\begin{lemma}
\label{lemma_ipw_1}
Under Assumptions  \ref{consistency} through \ref{stratexch}, and assuming
$\rho(x)>0$, $1-\rho(x)>0$, and $p_a(x)>0$ for $a=1, 0$ and all $x$ in the support of $X$, we have
\begin{equation*}
E[Y(1)|U=10, R=0] = \frac{1}{D}E\left[C\cdot A\cdot R \cdot \frac{p_1(X)-p_0(X)}{p_1(X)}\cdot\frac{1}{\pi(X)}\cdot \frac{1-\rho(X)}{\rho(X)}\cdot Y\right]
\end{equation*}
where $\pi(X)=P(A=1|R=1, X)$ and $D=E[(1-\rho(X))(p_1(X)-p_0(X))]$
\end{lemma}

\begin{proof}[Proof of Lemma \ref{lemma_ipw_1}]
\begin{align*}
E[Y(1)|U=10, R=0] &= E[E\{Y|X, C=1, A=1, R=1\}|U=10, R=0]\\
&= E\left[E\left\{\frac{I(C=1, A=1, R=1)}{P(C=1, A=1, R=1|X)}Y|X\right\}|U=10, R=0\right]\\
&= E\left[E\left\{\frac{ARC}{p_1(X)\pi(X)\rho(X)}Y|X\right\}|U=10, R=0\right]\\
&= \frac{E\left[I(U=10, R=0)E\left\{\frac{ARC}{p_1(X)\pi(X)\rho(X)}Y|X\right\}\right]}{P(U=10, R=0)}\\
&= \frac{E\left[\frac{f(X)}{p_1(X)\pi(X)\rho(X)}E\left\{I(U=10, R=0)|X\right\}\right]}{P(U=10, R=0)} \hspace{.5cm} \text{where $f(X)=E[ARYC|X]$}\\
&= \frac{E\left[\frac{f(X)}{p_1(X)\pi(X)\rho(X)}(p_1(X)-p_0(X))(1-\rho(X))\right]}{E[(p_1(X)-p_0(X))(1-\rho(X))]} \hspace{.5cm} \text{ Lemma \ref{denom_identification}}\\
&= \frac{E\left[\frac{(p_1(X)-p_0(X))(1-\rho(X))}{p_1(X)\pi(X)\rho(X)}E\{ARYC|X\}\right]}{E[(p_1(X)-p_0(X))(1-\rho(X))]}\\
&= \frac{E\left[C\cdot A\cdot R\cdot \frac{p_1(X)-p_0(X)}{p_1(X)}\cdot \frac{1}{\pi(X)}\cdot \frac{(1-\rho(X))}{\rho(X)}\cdot Y\right]}{E[(p_1(X)-p_0(X))(1-\rho(X))]}, \hspace{0.5cm} \text{Iterated Expectation}
\end{align*}
which corresponds to the expression given in Lemma \ref{lemma_ipw_1}.
\end{proof} 
As with the other identification results, we have an analogous lemma for the average potential outcomes among compliers in the target population under assignment to control:
\begin{lemma}
\label{lemma_ipw_2}
Under Assumptions  \ref{consistency} through \ref{stratexch}, and assuming
$\rho(x)>0$, $1-\rho(x)>0$, and $p_a(x)>0$ for $a=1, 0$ and all $x$ in the support of $X$, we have
\begin{equation*}
E[Y(0)|U=10, R=0] = \frac{1}{D}E\left[(1-C)\cdot (1-A)\cdot R \cdot \frac{p_1(X)-p_0(X)}{1-p_0(X)}\cdot\frac{1}{1-\pi(X)}\cdot \frac{1-\rho(X)}{\rho(X)}\cdot Y\right]
\end{equation*}
\end{lemma}
\begin{proof}[Proof of Lemma \ref{lemma_ipw_2}]
The proof proceeds almost identically as that of Lemma \ref{lemma_ipw_1}, where we replace references to $C$ and $A$ with $(1-C)$ and $(1-A)$ and references to $p_1(X)$ and $\pi(X)$ with $(1-p_0(X))$ and $(1-\pi(X))$.
\end{proof}

\begin{proof}[Proof of Theorem \ref{ipw_identification_thm}]
The first two identification claims of Theorem \ref{ipw_identification_thm} correspond to Lemmas \ref{lemma_ipw_1} and \ref{lemma_ipw_2}, respectively. That these identification results in turn imply identification of the entire treatment effect follows from identical reasoning as in the proof of Theorem \ref{ident_theorem} after expressing $\tau^0_{10}$ as a contrast of the two causal quantities.
\end{proof}

\begin{proof}[Proof of Theorem \ref{om_identification_thm}]
The identification formulas in Theorem \ref{om_identification_thm} follow directly from the plug-in identification formulas in Theorem \ref{ident_theorem}. Here, we show that the expression on the right-hand-side of Equation (\ref{eq: om_identification}) is equal to the expression on the right hand sides of Equation (\ref{eq: identification_res}). This fact in turn connects that expression to the causal parameter $\tau^0_{10}$. Using this strategy, we have:
\begin{align*}
&\frac{E\left[\{p_1(X)-p_0(X)\}(1-R)(\mu_{11}(X)-\mu_{00}(X))\right]}{E\left[\{p_1(X)-p_0(X)\}(1-R)\right]} \\
&= \frac{E\left[E\{(p_1(X)-p_0(X))(1-R)(\mu_{11}(X)-\mu_{00}(X))|X\}\right]}{E\left[E\{(p_1(X)-p_0(X))(1-R)|X\}\right]}\\
&= \frac{E\left[(p_1(X)-p_0(X))(\mu_{11}(X)-\mu_{00}(X))E\{(1-R)|X\}\right]}{E\left[(p_1(X)-p_0(X))E\{(1-R)|X\}\right]}\\
&= \frac{E\left[(p_1(X)-p_0(X))(\mu_{11}(X)-\mu_{00}(X))(1-\rho(X))\right]}{E\left[(p_1(X)-p_0(X))(1-\rho(X))\right]}\\
&= E[Y(1)-Y(0)|U=10, R=0],
\end{align*}
where the last equality follows from Theorem \ref{ident_theorem}.
\end{proof}

\section{EIF Derivation}
\label{supp_eif_derivation}
This section proves Theorem \ref{eif_thm}, which gives the efficient influence function (EIF) for
\begin{equation*}
\psi^0_{10} = \frac{E\left\{\left[p_1(X)-p_0(X)\right](1-\rho(X))\left[\mu_{11}(X)-\mu_{00}(X)\right]\right\}}{E\left\{\left[p_1(X)-p_0(X)\right](1-\rho(X))\right\}}.
\end{equation*}
As in the statement of the theorem, we first derive the EIFs for 
\begin{equation*}
\psi^0_{1, 10} =  \frac{E\left[\left\{p_1(X)-p_0(X)\right\}\left\{1-\rho(X)\right\}\mu_{11}(X)\right]}{E\left[\left\{p_1(X)-p_0(X)\right\}\left\{(1-\rho(X)\right\}\right]}
\end{equation*}
and
\begin{equation*}
\psi^0_{0, 10} = \frac{E\left[\left\{p_1(X)-p_0(X)\right\}\left\{1-\rho(X)\right\}\mu_{00}(X)\right]}{E\left[\left\{p_1(X)-p_0(X)\right\}\left\{(1-\rho(X)\right\}\right]}
\end{equation*}
separately. We begin with some background on EIFs that will be helpful later.

\subsection*{Preliminaries}

\noindent Let $p(z)$ be the joint density for our observed data $Z=(Y, C, A, R, X)$.  Let $p_\epsilon(z) = p(z, \epsilon)$ be a parametric submodel such that $p(z)=p(z, \epsilon = 0)$. Let
\begin{equation*}
    S_\epsilon(z)=\left.\frac{\partial \log p(z, \epsilon)}{\partial \epsilon}\right\rvert_{\epsilon = 0}.
\end{equation*}
Also, let $\mathcal{H}$ be the Hilbert space of mean-zero, finite-variance functions of $Z$ and $\Lambda$ be the mean-square closure of all parametric submodel tangent spaces, i.e., the set of all parametric submodel tangent spaces along with their limit points. From Theorem 4.4 of \citet{tsiatis_book}, we know that $\mathcal{H}$ and $\Lambda$ coincide in our nonparametric setting. Returning to our parametric submodel, we decompose $S_\epsilon(Z)$ as
\begin{equation*}
S_\epsilon(Z)=S_\epsilon(R)+S_\epsilon(X|R)+S_\epsilon(A|X, R)+S_\epsilon(C|A, X, R)+S_\epsilon(Y|A, R, C, X).
\end{equation*}
This follows by decomposing the log-likelihood of $p(z, \epsilon)$. The parametric submodel tangent space can be written as the direct sum of the tangent spaces associated with each of the score vectors in the above decomposition. Applying the existence of such a decomposition across all parametric submodels, we can apply Theorem 4.5 of \href{https://link.springer.com/book/10.1007/0-387-37345-4}{Tsiatis (2006)} to decompose $\mathcal{H}$ as follows:
\begin{equation*}
    \mathcal{H} = \Lambda_1\oplus\Lambda_2\oplus\Lambda_3\oplus\Lambda_4\oplus\Lambda_5
\end{equation*}
where $\Lambda_i$, gives the mean square closure of parametric submodel tangent spaces for the parametric submodels represented by the $i$th term in the score decomposition given above. For instance, $\Lambda_1$ is the set of $h(Z)\in \mathcal{H}$ such that $E[h^T(Z)h(Z)]<\infty$ and there exists a sequence $B_jS_{\epsilon_j}(R)$ such that
\begin{equation*}
    ||h(Z)-B_jS_{\epsilon_j}(R)||^2\overset{j\rightarrow \infty}{\rightarrow} 0
\end{equation*}
for a sequence of parametric submodels indexed by $j$, where $||g(\cdot)||^2=E[g^T(\cdot)g(\cdot)||$. Our goal is to derive the EIF for $\psi^0_{10}$, denoted $\varphi_{\psi^0_{10}}.$ We know that for parametric models, the efficient influence function is the unique influence function that resides in the tangent space. Since we know that the tangent space in the nonparametric setting is the entire set of mean-zero measurable functions of $Z$, then any influence function we derive for $\psi^0_{10}$ will be the efficient influence function. To ease notation a bit in the following calculations, we let $\psi^0_{10}=\psi$ and $\varphi_{\psi^0_{10}}=\varphi$. We will return to our earlier notation later, when we need to distinguish between multiple parameters and their respective efficient influence functions.
Applying Theorem 3.2 of \citet{tsiatis_book} to our parametric submodel, we know that the influence function $\varphi(Z)$ satisfies
\begin{equation}
    E\left[\varphi(Z)S_\epsilon\right] = \left.\frac{\partial \psi(P_\epsilon)}{\partial \epsilon}\right\rvert_{\epsilon = 0}.
    \label{eq: if_condition}
\end{equation}
Ultimately, we want to show that the EIF given in Equation (\ref{eq: eif_tau}) satisfies \ref{eq: if_condition}. We will do so by first finding the EIF for components of $\psi^0_{10}$ before combining these results. We start with the following general results regarding EIFs:
\begin{lemma}
\label{eif_fraction}
    (From \citet{jiang_multiply_2022}) If a parameter $R$ can be written as a ratio $R=N/D$, then, if $\varphi_N(Z)$ is the EIF for $N$, and $\varphi_D(Z)$ is the EIF for $D$, then the EIF for $R$ is given by
    \begin{equation}
        \varphi_R(Z)=\frac{1}{D}\varphi_N(Z)-\frac{R}{D}\varphi_D(Z).
    \end{equation}
\end{lemma}
\begin{proof}[Proof of Lemma \ref{eif_fraction}]
The proof of this statement is given in Lemma S2 of \citet{jiang_multiply_2022}.
\end{proof}
\noindent For completeness, we also give the following Lemma, which will be useful in later proofs:
\begin{lemma}
\label{EIF_diff}
If a parameter $R$ can be written as a difference $R=N-D$, then, if $\varphi_N(Z)$ is the EIF for $N$, and $\varphi_D(Z)$ is the EIF for $D$, then the EIF for $R$ is given by
    \begin{equation}
        \varphi_R(Z)=\varphi_N(Z)-\varphi_D(Z).
    \end{equation}
\end{lemma}
\begin{proof}[Proof of Lemma \ref{EIF_diff}]
That this holds follows directly from the rules of differentiation. That is, 
\begin{align*}
    E[(\varphi_R(Z)S_\epsilon(Z)] &= E[(\varphi_N(Z)-\varphi_D(Z))S_\epsilon(Z)]\\
    &= E[\varphi_N(Z)S_\epsilon(Z)]-E[\varphi_D(Z)S_\epsilon(Z)]\\
    &= \left.\frac{\partial N(P_\epsilon)}{\partial \epsilon}\right\rvert_{\epsilon=0}-\left.\frac{\partial D(P_\epsilon)}{\partial \epsilon}\right\rvert_{\epsilon=0}\\
    &= \left.\frac{\partial( N(P_\epsilon)-D(P_\epsilon))}{\partial \epsilon}\right\rvert_{\epsilon=0}\\
    &= \left.\frac{\partial R(P_\epsilon)}{\partial \epsilon}\right\rvert_{\epsilon=0},
\end{align*}
which proves the result.
\end{proof}

\subsection{Deriving the EIF of \texorpdfstring{$\psi^0_{1, 10}$}{TEXT}}
\label{eif_deriv_supp}
%Not really sure what the above does but it appears
%necessary to put math mode in a section heading
%https://tex.stackexchange.com/questions/5314/equations-in-section-heading-title

\noindent Now, we turn to deriving the EIF for the component of $\psi^0_{10}$ given by the parameter 
\begin{equation*}
\psi^0_{1, 10} = \frac{E\left[\left\{p_1(X)-p_0(X)\right\}\left\{1-\rho(X)\right\}\mu_{11}(X)\right]}{E\left[\left\{p_1(X)-p_0(X)\right\}\left\{(1-\rho(X)\right\}\right]}
\end{equation*}
To do so, we first let 
\begin{align*}
\psi^0_{1, 10} &= \frac{E\left[\left\{p_1(X)\right\}\left\{1-\rho(X)\right\}\mu_{11}(X)\right]-E\left[\left\{p_0(X)\right\}\left\{1-\rho(X)\right\}\mu_{11}(X)\right]}{E\left[\left\{p_1(X)-p_0(X)\right\}\left\{(1-\rho(X)\right\}\right]}\\
    &= \frac{\psi^1_N(P)-\psi^0_N(P)}{\psi_D(P)}
\end{align*}
and derive the EIFs for $\psi^1_N(P)$, $\psi^0_N(P)$, and $\psi_D(P)$ separately. 

\begin{lemma}[EIF for $\psi^1_N(P)$]
\label{lemma: eif_1_proof}
We claim that the EIF for $\psi^1_N(P)$ is given by
\begin{align}
\varphi^1_N(Z) &= \left\{p_1(X)\right\}\left\{1-\rho(X)\right\}\mu_{11}(X)-\psi^1_N(P) \nonumber\\
    &\hspace{1cm}+ \left[\frac{AR\left\{C-p_1(X)\right\}}{\pi(X)\rho(X)}\right]\left\{1-\rho(X)\right\}\mu_{11}(X) \nonumber\\
    &\hspace{1cm}+ \left\{p_1(X)\right\}\left[(1-R)-\left\{1-\rho(X)\right\}\right]\mu_{11}(X) \nonumber\\
    &\hspace{1cm}+\left\{p_1(X)\right\}\left\{1-\rho(X)\right\}\left[\frac{AR[YC-p_1(X)\mu_{11}(X)]-AR[C-p_1(X)]\mu_{11}(X)}{\pi(X)\rho(X)p_1(X)}\right]. \label{eq: putative_eif}
\end{align}
\end{lemma}

\begin{proof}[Proof of Lemma \ref{lemma: eif_1_proof}]
The proof of this claim below follows closely that of Theorem 3 in  \href{https://arxiv.org/abs/2302.00092}{Zeng et al. (2023)}. First we consider the pathwise derivative $\left.\frac{\partial \psi^1_N(P_\epsilon)}{\partial \epsilon}\right\rvert_{\epsilon=0}$. 
By definition,
\begin{align*}
    \psi^1_N(P) &= \int\sum_{r=0}^1\sum_{c=0}^1\int [yc(1-r)]p(y|A=1, C=1, R=1, x)p(c|A=1, R=1, x)p(r|x)p(x)\mathrm{d}y\mathrm{d}x
\end{align*}
We apply the product rule to compute the derivative for a given parametric submodel $P_\epsilon$:
\begin{align*}
    \frac{\partial \psi^1_N(P_\epsilon)}{\partial \epsilon} &=  \int\sum_{r=0}^1\sum_{c=0}^1\int [yc(1-r)]\frac{\partial p_\epsilon(y|A=1, C=1, R=1, x)}{\partial \epsilon} \\ 
    &\hspace{2.65cm} \times p_\epsilon(c|A=1, R=1, x)p_\epsilon(r|x)p_\epsilon(x)\mathrm{d}y\mathrm{d}x\\
    &\hspace{.25cm}+ \int\sum_{r=0}^1\sum_{c=0}^1\int [yc(1-r)]p_\epsilon(y|A=1, C=1, R=1, x)\frac{\partial p_\epsilon(c|A=1, R=1, x)}{\partial \epsilon} \\
    &\hspace{2.65cm} \times p_\epsilon(r|x)p_\epsilon(x)\mathrm{d}y\mathrm{d}x\\
    &\hspace{.25cm}+ \int\sum_{r=0}^1\sum_{c=0}^1\int [yc(1-r)]p_\epsilon(y|A=1, C=1, R=1, x)p_\epsilon(c|A=1, R=1, x) \\
    &\hspace{2.65cm} \times \frac{\partial p_\epsilon(r|x)}{\partial \epsilon}p_\epsilon(x)\mathrm{d}y\mathrm{d}x\\
    &\hspace{.25cm}+ \int\sum_{r=0}^1\sum_{c=0}^1\int [yc(1-r)]p_\epsilon(y|A=1, C=1, R=1, x)p_\epsilon(c|A=1, R=1, x) \\
    &\hspace{2.65cm} \times p_\epsilon(r|x)\frac{\partial p_\epsilon(x)}{\partial \epsilon}\mathrm{d}y\mathrm{d}x.
\end{align*}
Evaluating this derivative at the truth, i.e., $\epsilon=0$, we have
that $\left.\frac{\partial \psi^1_N(P_\epsilon)}{\partial \epsilon}\right\rvert_{\epsilon=0}$ is equal to 
\begin{align}
    &\hspace{-1.5cm}\int\sum_{r=0}^1\sum_{c=0}^1\int [yc(1-r)]S_\epsilon(y|A=1, C=1, R=1, x)p(y|A=1, C=1, R=1, x) \nonumber \\
    &\hspace{.65cm} \times p(c|A=1, R=1, x)p(r|x)p(x)\mathrm{d}y\mathrm{d}x \nonumber\\
    &\hspace{-.8cm}+ \int\sum_{r=0}^1\sum_{c=0}^1\int [yc(1-r)]p(y|A=1, C=1, R=1, x)S_\epsilon(c|A=1, R=1, x) \nonumber \\ 
    &\hspace{1.75cm} \times p(c|A=1, R=1, x)p(r|x)p(x)\mathrm{d}y\mathrm{d}x \nonumber\\
    &\hspace{-.8cm}+ \int\sum_{r=0}^1\sum_{c=0}^1\int [yc(1-r)]p(y|A=1, C=1, R=1, x)p(c|A=1, R=1, x) \nonumber \\ 
    &\hspace{1.75cm} \times S_\epsilon(r|x)p(r|x)p(x)\mathrm{d}y\mathrm{d}x \nonumber\\
    &\hspace{-.8cm}+ \int\sum_{r=0}^1\sum_{c=0}^1\int [yc(1-r)]p(y|A=1, C=1, R=1, x)p(c|A=1, R=1, x) \nonumber \\ 
    &\hspace{1.75cm} \times p(r|x)S_\epsilon(x)p(x)\mathrm{d}y\mathrm{d}x \nonumber\\
    &\hspace{-1.5cm}= E\left(\left\{p_1(X)\right\}\left\{1-\rho(X)\right\}E\left[YS_\epsilon(Y|C=1, A=1, R=1, X)|C=1, A=1, R=1, X\right]\right) \nonumber\\
    &\hspace{-.5cm} + E\left[\left(\left\{1-\rho(X)\right\}\mu_{11}(X)\right)E[CS_\epsilon(C|A=1, R=1, X)| A=1, R=1, X]\right] \nonumber\\
    &\hspace{-.5cm} + E\left[\left\{p_1(X)\right\}\mu_{11}(X)E[(1-R)S_\epsilon(R|X)|X]\right] \nonumber\\
    &\hspace{-.5cm} + E\left[\left\{p_1(X)\right\}\left\{1-\rho(X)\right\}\mu_{11}(X)S_\epsilon(X)\right].\label{eq: pathwise_deriv}
\end{align}
Recall that we want to show that our putative EIF $\varphi^1_N$ satisfies:
\begin{equation*}
E\left[\varphi^1_N(Z)S_\epsilon(Z)\right] = \left.\frac{\partial \psi_1(P_\epsilon)}{\partial \epsilon}\right\rvert_{\epsilon = 0}.
\end{equation*}
To do so, we evaluate each summand in $E\left[\varphi^1_N(Z)S_\epsilon(Z)\right]$ and show their correspondence to components of (\ref{eq: pathwise_deriv}). 

\subsection*{Obtaining the 4\ts{th} component of (\ref{eq: pathwise_deriv})}
Beginning with the first term in (\ref{eq: putative_eif}), we have
\begin{align*}
    &\hspace{-0.2cm}E\left[\left(\left\{p_1(X)\right\}\left\{1-\rho(X)\right\}\mu_{11}(X)-\psi^1_N(P)\right)S_\epsilon(Z)\right] \\
    &= E\left[\left\{p_1(X)\right\}\left\{1-\rho(X)\right\}\mu_{11}(X)S_\epsilon(Z)\right]-0\\
    &= E\left[\left\{p_1(X)\right\}\left\{1-\rho(X)\right\}\mu_{11}(X)S_\epsilon(Y, C, A, R|X)\right]\\
    &\hspace{1cm}+E\left[\left\{p_1(X)\right\}\left\{1-\rho(X)\right\}\mu_{11}(X)S_\epsilon(X)\right]\\
    &= E\left[\left\{p_1(X)\right\}\left\{1-\rho(X)\right\}\mu_{11}(X)E[S_\epsilon(Y, C, A, R|X)|X]\right]\\
    &\hspace{1cm}+E\left[\left\{p_1(X)\right\}\left\{1-\rho(X)\right\}\mu_{11}(X)S_\epsilon(X)\right]\\
    &= 0 + E\left[\left\{p_1(X)\right\}\left\{1-\rho(X)\right\}\mu_{11}(X)S_\epsilon(X)\right]\\
    &= E\left[\left\{p_1(X)\right\}\left\{1-\rho(X)\right\}\mu_{11}(X)S_\epsilon(X)\right],
\end{align*}
which is the 4\ts{th} component of (\ref{eq: pathwise_deriv}). 

\subsection*{Obtaining the 2\ts{nd} component of (\ref{eq: pathwise_deriv})}

Next, we have
\begin{align*}
    E\left[\left(\left[\frac{AR\left\{C-p_1(X)\right\}}{\pi(X)\rho(X)}\right]\left\{1-\rho(X)\right\}\mu_{11}(X)\right)S_\epsilon(Z)\right].
\end{align*}
We decompose $S_\epsilon(Z) = S_\epsilon(C|A, R, X)+S_\epsilon(Y|C, A, R, X)+S_\epsilon(A, R, X)$. Evaluating each multiplied term separately, we first consider
\begin{align*}
    &E\left[\left(\left[\frac{AR\left\{C-p_1(X)\right\}}{\pi(X)\rho(X)}\right]\left\{1-\rho(X)\right\}\mu_{11}(X)\right)S_\epsilon(A, R, X)\right]\\
    &= E\left[\left(\left[\frac{AR\left\{S_\epsilon(A, R, X)\right\}}{\pi(X)\rho(X)}\right]\left\{1-\rho(X)\right\}\mu_{11}(X)\right)E[C-p_1(X)|A=1, R=1, X]\right]\\
    &= 0.
\end{align*}
Next, we have
\begin{align*}
    &E\left[\left(\left[\frac{AR\left\{C-p_1(X)\right\}}{\pi(X)\rho(X)}\right]\left\{1-\rho(X)\right\}\mu_{11}(X)\right)S_\epsilon(Y|C, A, R, X)\right] \\
    &= E\left[\left(\left[\frac{AR\left\{C-p_1(X)\right\}}{\pi(X)\rho(X)}\right]\left\{1-\rho(X)\right\}\mu_{11}(X)\right)E[S_\epsilon(Y|C, A, R, X)|C, A, R, X]\right]\\
    &= 0.
\end{align*}
Finally, we break up our evaluation of $E\left[\left(\left[\frac{AR\left\{C-p_1(X)\right\}}{\pi(X)\rho(X)}\right]\left\{1-\rho(X)\right\}\mu_{11}(X)\right)S_\epsilon(C|A, R, X)\right]$ into two parts. The first term is
\begin{align*}
    &E\left[\left(\left[\frac{ARC}{\pi(X)\rho(X)}\right]\left\{1-\rho(X)\right\}\mu_{11}(X)\right)S_\epsilon(C|A, R, X)\right]\\
    &= E\left[\left(\left[\frac{ARC}{\pi(X)\rho(X)}\right]\left\{1-\rho(X)\right\}\mu_{11}(X)\right)S_\epsilon(C|A=1, R=1, X)\right]\\
    &= E\left[\left(\left[\frac{AR}{\pi(X)\rho(X)}\right]\left\{1-\rho(X)\right\}\mu_{11}(X)\right)E[CS_\epsilon(C|A=1, R=1, X)| A=1, R=1, X]\right]\\
    &= E\left[\left(\left\{1-\rho(X)\right\}\mu_{11}(X)\right)E[CS_\epsilon(C|A=1, R=1, X)| A=1, R=1, X]\right],
\end{align*}
which is the 2\ts{nd} component of (\ref{eq: pathwise_deriv}). The second term is
\begin{align*}
    &-E\left[\left(\left[\frac{ARp_1(X)}{\pi(X)\rho(X)}\right]\left\{1-\rho(X)\right\}\mu_{11}(X)\right)S_\epsilon(C|A, R, X)\right]\\
    &= -E\left[\left(\left[\frac{ARp_1(X)}{\pi(X)\rho(X)}\right]\left\{1-\rho(X)\right\}\mu_{11}(X)\right)E[S_\epsilon(C|A, R, X)|A, R, X]\right]\\
    &= 0.
\end{align*}

\subsection*{Obtaining the 3\ts{rd} component of (\ref{eq: pathwise_deriv})}
Turning our attention to 
\begin{equation*}
    E\left[\left\{p_1(X)\right\}\left[(1-R)-\left\{1-\rho(X)\right\}\right]\mu_{11}(X)S_\epsilon(Z)\right],
\end{equation*}
we decompose $S_\epsilon(Z) = S_\epsilon(R|X)+S_\epsilon(Y, C, A|R, X)+S_\epsilon(X)$. Considering each term separately, we first evaluate
\begin{align*}
    &E\left[\left\{p_1(X)\right\}\left[(1-R)-\left\{1-\rho(X)\right\}\right]\mu_{11}(X)S_\epsilon(Y, C, A|R, X)\right]\\
    &= E\left[\left\{p_1(X)\right\}\left[(1-R)-\left\{1-\rho(X)\right\}\right]\mu_{11}(X)E[S_\epsilon(Y, C, A|R, X)|R, X]\right]\\
    &= 0.
\end{align*}
Next, we have
\begin{align*}
    &E\left[\left\{p_1(X)\right\}\left[(1-R)-\left\{1-\rho(X)\right\}\right]\mu_{11}(X)S_\epsilon(X)\right]\\
    &= E\left[\left\{p_1(X)\right\}\mu_{11}(X)S_\epsilon(X)E\left[\left[(1-R)-\left\{1-\rho(X)\right\}\right]|X\right]\right]\\
    &= 0.
\end{align*}
Finally, 
\begin{align*}
    &E\left[\left\{p_1(X)\right\}\left[(1-R)-\left\{1-\rho(X)\right\}\right]\mu_{11}(X)S_\epsilon(R|X)\right]\\
    &= E\left[\left\{p_1(X)\right\}(1-R)\mu_{11}(X)S_\epsilon(R|X)\right]-E\left[\left\{p_1(X)\right\}\left[\left\{1-\rho(X)\right\}\right]\mu_{11}(X)S_\epsilon(R|X)\right]\\
    &= E\left[\left\{p_1(X)\right\}(1-R)\mu_{11}(X)S_\epsilon(R|X)\right]-E\left[\left\{p_1(X)\right\}\left[\left\{1-\rho(X)\right\}\right]\mu_{11}(X)E[S_\epsilon(R|X)|X]\right]\\
    &= E\left[\left\{p_1(X)\right\}(1-R)\mu_{11}(X)S_\epsilon(R|X)\right]-0\\
    &= E\left[\left\{p_1(X)\right\}\mu_{11}(X)E[(1-R)S_\epsilon(R|X)|X]\right],
\end{align*}
which is the 3\ts{rd} component of (\ref{eq: pathwise_deriv}).

\subsection*{Obtaining the 1\ts{st} component of (\ref{eq: pathwise_deriv})}
Now we turn to evaluating
\begin{equation*}
    E\left(\left\{p_1(X)\right\}\left\{1-\rho(X)\right\}\left[\frac{AR[YC-p_1(X)\mu_{11}(X)]-AR[C-p_1(X)]\mu_{11}(X)}{\pi(X)\rho(X)p_1(X)}\right]S_\epsilon(Z)\right).
\end{equation*}
We start by breaking up $S_\epsilon(Z)$ as $S_\epsilon(Z) = S_\epsilon(Y|C, A, R, X)+S_\epsilon(C, A, R, X)$. 
Considering each component in turn, we start with
\begin{align*}
    &E\left(\left\{p_1(X)\right\}\left\{1-\rho(X)\right\}\left[\frac{AR[YC-p_1(X)\mu_{11}(X)]-AR[C-p_1(X)]\mu_{11}(X)}{\pi(X)\rho(X)p_1(X)}\right]S_\epsilon(C, A, R, X)\right)\\
    & = E\left(\left\{p_1(X)\right\}\left\{1-\rho(X)\right\} \right. \\
    & \quad\quad \left. \times E\left[\left.\left\{\frac{AR[YC-p_1(X)\mu_{11}(X)]-AR[C-p_1(X)]\mu_{11}(X)}{\pi(X)\rho(X)p_1(X)}\right\}S_\epsilon(C, A, R, X)\right\rvert C, A, R, X\right]\right)\\
    &= E\left(\left\{p_1(X)\right\}\left\{1-\rho(X)\right\}S_\epsilon(C, A, R, X) \right. \\
    & \quad\quad \left. \times E\left.\left\{\frac{AR[YC-p_1(X)\mu_{11}(X)]-AR[C-p_1(X)]\mu_{11}(X)}{\pi(X)\rho(X)p_1(X)}\right\rvert C, A, R, X\right\}\right)\\
    &= E\left(\left\{p_1(X)\right\}\left\{1-\rho(X)\right\}S_\epsilon(C, A, R, X)\left[\frac{AR[\mu_{11}(X)C-p_1(X)\mu_{11}(X)]-AR[C-p_1(X)]\mu_{11}(X)}{\pi(X)\rho(X)p_1(X)}\right]\right)\\
    &= 0.
\end{align*}
Next, 
\begin{align*}
    &E\left(\left\{p_1(X)\right\}\left\{1-\rho(X)\right\}\left[\frac{AR[YC-p_1(X)\mu_{11}(X)]-AR[C-p_1(X)]\mu_{11}(X)}{\pi(X)\rho(X)p_1(X)}\right]S_\epsilon(Y|C, A, R, X)\right)\\
    &= E\left(\left\{p_1(X)\right\}\left\{1-\rho(X)\right\}\left[\frac{ARYC}{\pi(X)\rho(X)p_1(X)}\right]S_\epsilon(Y|C, A, R, X)\right)\\
    &\hspace{.5cm} - E\left(\left\{p_1(X)\right\}\left\{1-\rho(X)\right\}\left[\frac{AR[p_1(X)\mu_{11}(X)]+AR[C-p_1(X)]\mu_{11}(X)}{\pi(X)\rho(X)p_1(X)}\right]S_\epsilon(Y|C, A, R, X)\right)\\
    &= E\left(\left\{p_1(X)\right\}\left\{1-\rho(X)\right\}\left[\frac{ARYC}{\pi(X)\rho(X)p_1(X)}\right]S_\epsilon(Y|C=1, A=1, R=1, X)\right)\\
    &\hspace{.5cm}-E\left(\left\{p_1(X)\right\}\left\{1-\rho(X)\right\} \vphantom{\frac{AR}{AR}}\right.\\ 
    & \quad\quad\quad\quad \left. \times \left[\frac{AR[p_1(X)\mu_{11}(X)]+AR[C-p_1(X)]\mu_{11}(X)}{\pi(X)\rho(X)p_1(X)}\right]E[S_\epsilon(Y|C, A, R, X)|C, A, R, X]\right)\\
    &= E\left(\left\{p_1(X)\right\}\left\{1-\rho(X)\right\}\left[\frac{ARYC}{\pi(X)\rho(X)p_1(X)}\right]S_\epsilon(Y|C=1, A=1, R=1, X)\right)+0\\
    &= E\left(\left\{p_1(X)\right\}\left\{1-\rho(X)\right\}\frac{ARC}{\pi(X)\rho(X)p_1(X)} \right. \\ 
    & \quad\quad\quad \left. \vphantom{\frac{AR}{AR}} \times E\left[YS_\epsilon(Y|C=1, A=1, R=1, X)|C=1, A=1, R=1, X\right]\right)\\
    &= E\left(\left\{p_1(X)\right\}\left\{1-\rho(X)\right\}\frac{ARC}{\pi(X)\rho(X)p_1(X)}g(X)\right)\\
    &= E\left(\left\{p_1(X)\right\}\left\{1-\rho(X)\right\}\frac{g(X)}{\rho(X)\pi(X)p_1(X)}E[R\hspace{1mm}E[A \hspace{1mm}E\left[C|A=1, R=1, X\right]|R=1, X]|X]\right)\\
    &= E\left(\left\{p_1(X)\right\}\left\{1-\rho(X)\right\}\frac{g(X)}{\rho(X)\pi(X)p_1(X)}E[R|X]E[A|R=1, X]E[C|R=1, A=1, X]\right)\\
    &= E\left(\left\{p_1(X)\right\}\left\{1-\rho(X)\right\}\frac{g(X)}{\rho(X)\pi(X)p_1(X)}\rho(X)\pi(X)p_1(X)\right)\\
    &= E\left(\left\{p_1(X)\right\}\left\{1-\rho(X)\right\}E\left[YS_\epsilon(Y|C=1, A=1, R=1, X)|C=1, A=1, R=1, X\right]\right),
\end{align*}
where, in the 5\ts{th} equality, we let $E\left[YS_\epsilon(Y|C=1, A=1, R=1, X)|C=1, A=1, R=1, X\right]=g(X)$
to save space and emphasize that this expectation is random as a function of $X$. This is the 1\ts{st} component of (\ref{eq: pathwise_deriv}).
\end{proof}

\begin{lemma}[EIF for $\psi^0_N(P)$]
\label{lemma: eif_0_proof}
We claim that the EIF for $\psi^0_N(P)$ is given by
\begin{align*}
\varphi^0_N(Z) &= p_0(X)\left\{1-\rho(X)\right\}\mu_{11}(X)-\psi^0_N(P)\\
    &\hspace{1cm}+ \left[\frac{(1-A)R\left\{C-p_0(X)\right\}}{(1-\pi(X))\rho(X)}\right]\left\{1-\rho(X)\right\}\mu_{11}(X)\\
    &\hspace{1cm}+ \left\{p_0(X)\right\}\left[(1-R)-\left\{1-\rho(X)\right\}\right]\mu_{11}(X)\\
    &\hspace{1cm}+\left\{p_0(X)\right\}\left\{1-\rho(X)\right\}\left[\left\{\frac{AR[YC-p_1(X)\mu_{11}(X)]-AR[C-p_1(X)]\mu_{11}(X)}{\pi(X)\rho(X)p_1(X)}\right\}\right].
\end{align*}
\end{lemma}

\begin{proof}[Proof of Lemma \ref{lemma: eif_0_proof}]
The proof of this result follows almost identically to that of Lemma \ref{lemma: eif_1_proof}.
\end{proof}

\begin{lemma}[EIF for $\psi_D(P)$]
\label{lemma: eif_denom}
We claim that the EIF for $\psi_D(P)$ is given by
\begin{align*}
\varphi_D(Z) &= \left\{p_1(X)-p_0(X)\right\}\left\{1-\rho(X)\right\}-\psi_D(P)\\
    &\hspace{1cm}+ \left[\frac{AR\left\{C-p_1(X)\right\}}{\pi(X)\rho(X)}-\frac{(1-A)R\left\{C-p_0(X)\right\}}{(1-\pi(X))\rho(X)}\right]\left\{1-\rho(X)\right\}\\
    &\hspace{1cm}+ \left\{p_1(X)-p_0(X)\right\}\left[(1-R)-\left\{1-\rho(X)\right\}\right].
\end{align*}
\end{lemma}

\begin{proof}[Proof of Lemma \ref{lemma: eif_denom}]
Our proof of this claim follows a similar structure to those above. We start by writing out the pathwise derivative $\left.\frac{\partial \psi_D(P_\epsilon)}{\partial \epsilon}\right\rvert_{\epsilon=0}$, where we have
\begin{equation*}
    \psi_D(P_\epsilon) = \int\sum_{r=0}^1\sum_{c=0}^1 c(1-r)\left[p_{\epsilon}(c|A=1, R=1, x)-p_{\epsilon}(c|A=0, R=1, x)\right]p_\epsilon(r|x)p_{\epsilon}(x)\mathrm{d}x.
\end{equation*}
Thus, 
\begin{align*}
    \frac{\partial \psi_D(P_\epsilon)}{\partial \epsilon} &= \int\sum_{r=0}^1\sum_{c=0}^1 c(1-r)\left[\frac{\partial(p_{\epsilon}(c|A=1, R=1, x)-p_{\epsilon}(c|A=0, R=1, x))}{\partial \epsilon}\right]p_\epsilon(r|x)p_{\epsilon}(x)\mathrm{d}x\\
    &\hspace{1cm} + \int\sum_{r=0}^1\sum_{c=0}^1 c(1-r)\left[p_{\epsilon}(c|A=1, R=1, x)-p_{\epsilon}(c|A=0, R=1, x)\right]\frac{\partial p_\epsilon(r|x)}{\partial \epsilon}p_{\epsilon}(x)\mathrm{d}x\\
    &\hspace{1cm} + \int\sum_{r=0}^1\sum_{c=0}^1 c(1-r)\left[p_{\epsilon}(c|A=1, R=1, x)-p_{\epsilon}(c|A=0, R=1, x)\right]p_\epsilon(r|x)\frac{\partial p_{\epsilon}(x)}{\partial \epsilon}\mathrm{d}x.
\end{align*}
Evaluating at the true model:
\begin{align}
    &\hspace{-0.5cm}\left.\frac{\partial \psi_D(P_\epsilon)}{\partial \epsilon}\right\rvert_{\epsilon=0} \nonumber \\
    &\hspace{-0.5cm}= \int\sum_{r=0}^1\sum_{c=0}^1 c(1-r)\left[S_\epsilon(c|A=1, R=1, x)-S_\epsilon(c|A=0, R=1, x)\right]p_\epsilon(r|x)p_{\epsilon}(x)\mathrm{d}x \nonumber\\
    &+ \int\sum_{r=0}^1\sum_{c=0}^1 c(1-r)\left[p_{\epsilon}(c|A=1, R=1, x)-p_{\epsilon}(c|A=0, R=1, x)\right]S_\epsilon(r|x)p_{\epsilon}(x)\mathrm{d}x \nonumber\\
    & + \int\sum_{r=0}^1\sum_{c=0}^1 c(1-r)\left[p_{\epsilon}(c|A=1, R=1, x)-p_{\epsilon}(c|A=0, R=1, x)\right]p_\epsilon(r|x)S_\epsilon(x)\mathrm{d}x. \label{eq: pathwise_deriv_denom}
\end{align}
As above, we now turn to $E[\varphi_D(Z)S_\epsilon(Z)]$ and evaluate each component of our putative EIF $\varphi_D(Z)$ in turn. 

\subsection*{Obtaining the 1\ts{st} component of (\ref{eq: pathwise_deriv_denom})}
Here, we consider
\begin{equation*}
    E\left\{\left[\frac{AR\left\{C-p_1(X)\right\}}{\pi(X)\rho(X)}-\frac{(1-A)R\left\{C-p_0(X)\right\}}{(1-\pi(X))\rho(X)}\right]\left\{1-\rho(X)\right\}S_\epsilon(Z)\right\}.
\end{equation*}
We decompose 
\begin{equation*}
    S_\epsilon(Z) = S_\epsilon(C|A, R, X)+S_\epsilon(Y|C, A, R, X)+S_\epsilon(A, R, X).
\end{equation*}
Evaluating the first component, we have
\begin{align*}
    &E\left\{\left[\frac{AR\left\{C-p_1(X)\right\}}{\pi(X)\rho(X)}-\frac{(1-A)R\left\{C-p_0(X)\right\}}{(1-\pi(X))\rho(X)}\right]\left\{1-\rho(X)\right\}S_\epsilon(C|A, R, X)\right\}\\
    &= E\left\{\left[\frac{AR\left\{E[(C-p_1(X))S_\epsilon(C|A, R, X)|A, R, X]\right\}}{\pi(X)\rho(X)}\right.\right.\\
    &\hspace{1.5cm}\left.\left.-\frac{(1-A)R\left\{E[(C-p_0(X))S_\epsilon(C|A, R, X)|A, R, X]\right\}}{(1-\pi(X))\rho(X)}\right]\left\{1-\rho(X)\right\}\right\}\\
    &= E\left\{\left[\left\{E[(C-p_1(X))S_\epsilon(C|A=1, R=1, X)|A=1, R=1, X]\right\}\right.\right.\\
    &\hspace{1cm}-\left.\left.\left\{E[(C-p_0(X))S_\epsilon(C|A=0, R=1, X)|A=0, R=1, X]\right\}\right]\left\{1-\rho(X)\right\}\right\}\\
    &= E\{[\{E[CS_\epsilon(C|A=1, R=1, X)|A=1, R=1, X]\}\\
    &\hspace{1cm}-\{E[CS_\epsilon(C|A=0, R=1, X)|A=0, R=1, X]\}]\{1-\rho(X)\}\}-0+0
\end{align*}
which is the the 1\ts{st} component of (\ref{eq: pathwise_deriv_denom}).

Evaluating the second component of the score:
\begin{align*}
    &E\left\{\left[\frac{AR\left\{C-p_1(X)\right\}}{\pi(X)\rho(X)}-\frac{(1-A)R\left\{C-p_0(X)\right\}}{(1-\pi(X))\rho(X)}\right]\left\{1-\rho(X)\right\}S_\epsilon(Y|C, A, R, X)\right\}\\
    &= E\left\{\left[\frac{AR\left\{C-p_1(X)\right\}}{\pi(X)\rho(X)}-\frac{(1-A)R\left\{C-p_0(X)\right\}}{(1-\pi(X))\rho(X)}\right]\left\{1-\rho(X)\right\}E[S_\epsilon(Y|C, A, R, X)|C, A, R, X]\right\}\\
    &= 0.
\end{align*}
Evaluating the third component of the score:
\begin{align*}
    &E\left\{\left[\frac{AR\left\{C-p_1(X)\right\}}{\pi(X)\rho(X)}-\frac{(1-A)R\left\{C-p_0(X)\right\}}{(1-\pi(X))\rho(X)}\right]\left\{1-\rho(X)\right\}S_\epsilon( A, R, X)\right\}\\
    &= E\left\{\left[\frac{AR\left\{C-p_1(X)\right\}}{\pi(X)\rho(X)}\right]\left\{1-\rho(X)\right\}S_\epsilon( A, R, X)\right\} \\
    &\hspace{0.5cm} -E\left\{\left[\frac{(1-A)R\left\{C-p_0(X)\right\}}{(1-\pi(X))\rho(X)}\right]\left\{1-\rho(X)\right\}S_\epsilon( A, R, X)\right\}\\
    &= E\left\{\left[\frac{AR\left\{E[C-p_1(X)|A=1, R=1, X]\right\}}{\pi(X)\rho(X)}\right]\left\{1-\rho(X)\right\}S_\epsilon( A, R, X)\right\}\\
    &\hspace{0.5cm}-E\left\{\left[\frac{(1-A)R\left\{E[C-p_0(X), A=0, R=1, X]\right\}}{(1-\pi(X))\rho(X)}\right]\left\{1-\rho(X)\right\}S_\epsilon( A, R, X)\right\}\\
    &= 0.
\end{align*}

\subsection*{Obtaining the 2\ts{nd} component of (\ref{eq: pathwise_deriv_denom})}
Now we turn to 
\begin{equation*}
    E\left\{\left\{p_1(X)-p_0(X)\right\}\left[(1-R)-\left\{1-\rho(X)\right\}\right]S_\epsilon(Z)\right\}
\end{equation*}
and decompose $S_\epsilon(Z) = S_\epsilon(R|X)+S_\epsilon(Y, C, A|R, X)+S_\epsilon(X)$. Considering each term separately, we first evaluate
\begin{align*}
    &E\left[\left\{p_1(X)-p_0(X)\right\}\left[(1-R)-\left\{1-\rho(X)\right\}\right]S_\epsilon(Y, C, A|R, X)\right]\\
    &= E\left[\left\{p_1(X)-p_0(X)\right\}\left[(1-R)-\left\{1-\rho(X)\right\}\right]E[S_\epsilon(Y, C, A|R, X)|R, X]\right]\\
    &= 0.
\end{align*}
Next, we have
\begin{align*}
    &E\left[\left\{p_1(X)-p_0(X)\right\}\left[(1-R)-\left\{1-\rho(X)\right\}\right]S_\epsilon(X)\right]\\
    &= E\left[\left\{p_1(X)-p_0(X)\right\}S_\epsilon(X)E\left[\left[(1-R)-\left\{1-\rho(X)\right\}\right]|X\right]\right]\\
    &= 0.
\end{align*}
Finally, 
\begin{align*}
    &E\left[\left\{p_1(X)-p_0(X)\right\}\left[(1-R)-\left\{1-\rho(X)\right\}\right]S_\epsilon(R|X)\right]\\
    &= E\left[\left\{p_1(X)-p_0(X)\right\}(1-R)S_\epsilon(R|X)\right]-E\left[\left\{p_1(X)-p_0(X)\right\}\left[\left\{1-\rho(X)\right\}\right]S_\epsilon(R|X)\right]\\
    &= E\left[\left\{p_1(X)-p_0(X)\right\}(1-R)S_\epsilon(R|X)\right]-E\left[\left\{p_1(X)-p_0(X)\right\}\left[\left\{1-\rho(X)\right\}\right]E[S_\epsilon(R|X)|X]\right]\\
    &= E\left[\left\{p_1(X)-p_0(X)\right\}(1-R)S_\epsilon(R|X)\right]-0\\
    &= E\left[\left\{p_1(X)-p_0(X)\right\}E[(1-R)S_\epsilon(R|X)|X]\right],
\end{align*}
which is the 2\ts{nd} component of (\ref{eq: pathwise_deriv_denom}).

\subsection*{Obtaining the 3\ts{rd} component (\ref{eq: pathwise_deriv_denom})}
We focus now on
\begin{equation*}
    E[\left(\left\{p_1(X)-p_0(X)\right\}\left\{1-\rho(X)\right\}-\psi_D(P)\right)S_\epsilon(Z)],
\end{equation*}
and, noting that $S_\epsilon(Z)=S_\epsilon(X)+S_\epsilon(Y, C, A, R|X)$, we have
\begin{align*}
& E[(\{p_1(X)-p_0(X)\}\{1-\rho(X)\}-\psi_D(P))S_\epsilon(Z)] \\
&= E[(\{p_1(X)-p_0(X)\}\{1-\rho(X)\}-\psi_D(P))S_\epsilon(Y, C, A, R|X)]\\
&\hspace{1cm} + E[(\{p_1(X)-p_0(X)\}\{1-\rho(X)\}-\psi_D(P))S_\epsilon(X)]\\
&= E[(\{p_1(X)-p_0(X)\}\{1-\rho(X)\}\\
&\hspace{1cm}-\psi_D(P))E[S_\epsilon(Y, C, A, R|X)|X]]\\
&\hspace{1cm} + E[(\{p_1(X)-p_0(X)\}\{1-\rho(X)\}-\psi_D(P))S_\epsilon(X)]\\
&= 0 + E[(\{p_1(X)-p_0(X)\}\{1-\rho(X)\}-\psi_D(P))S_\epsilon(X)]\\
&= E[(\{p_1(X)-p_0(X)\}\{1-\rho(X)\}-\psi_D(P))S_\epsilon(X)]
\end{align*}
which is the 3\ts{rd} component of (\ref{eq: pathwise_deriv_denom}).
\end{proof}

\noindent Finally, we bring these lemmas together to prove the following result:
\begin{lemma}[EIF of $\psi^0_{1, 10}$]
\label{eif_lemma_1}
The EIF for $\psi^0_{1, 10}$ is given by 
\begin{equation}
 \varphi_{\psi^0_{1, 10}}=\frac{\phi^0_{1, 10}}{E\left[e_{10}(X)\left\{1-\rho(X)\right\}\right]}-\psi^0_{1, 10}\left(\frac{\left[\psi_{C_{1, 1}}-\psi_{C_{0, 1}}\right]\left\{1-\rho(X)\right\}+e_{10}(X)\psi_{1-R}}{E\left[e_{10}(X)(1-\rho(X))\right]}\right)
 \end{equation}
 where the additional notation is as defined in the main manuscript.
 \end{lemma}
 \begin{proof}[Proof of Lemma \ref{eif_lemma_1}]
Recall that we write our parameter of interest as
\begin{align*}
    \psi(P) &= \frac{E\left[\left\{p_1(X)\right\}\left\{1-\rho(X)\right\}\mu_{11}(X)\right]-E\left[\left\{p_0(X)\right\}\left\{1-\rho(X)\right\}\mu_{11}(X)\right]}{(p_1-p_0)(1-\rho)}\\
    &=\frac{\psi^1_N(P)-\psi^0_N(P)}{\psi_D(P)}
\end{align*}
and have derived the efficient influence function for each component in Lemmas \ref{lemma: eif_1_proof}, \ref{lemma: eif_0_proof}, and \ref{lemma: eif_denom} as summarized in the following Table \ref{eif_table}.
\newpage
\begin{table}[h!]
\centering
\begin{tabular}{c|l}
    \textbf{Parameter} & \textbf{Efficient Influence Function}\\
    \hline \hline \\
     $\psi^1_N(P)$ & {$\begin{aligned}& \left\{p_1(X)\right\}\left\{1-\rho(X)\right\}\mu_{11}(X)-\psi^1_N(P)\\
    &\hspace{1cm}+ \left[\frac{AR\left\{C-p_1(X)\right\}}{\pi(X)\rho(X)}\right]\left\{1-\rho(X)\right\}\mu_{11}(X)\\
    &\hspace{1cm}+ \left\{p_1(X)\right\}\left[(1-R)-\left\{1-\rho(X)\right\}\right]\mu_{11}(X)\\
    &\hspace{1cm}+\left\{p_1(X)\right\}\left\{1-\rho(X)\right\}\left[\left\{\frac{AR[YC-p_1(X)\mu_{11}(X)]-AR[C-p_1(X)]\mu_{11}(X)}{\pi(X)\rho(X)p_1(X)}\right\}\right]
     \end{aligned}$} \\ \\
     \hline\\
    $\psi^0_N(P)$ & {$\begin{aligned}
&\left\{p_0(X)\right\}\left\{1-\rho(X)\right\}\mu_{11}(X)-\psi^0_N(P)\\
    &\hspace{1cm}+ \left[\frac{(1-A)R\left\{C-p_0(X)\right\}}{(1-\pi(X))\rho(X)}\right]\left\{1-\rho(X)\right\}\mu_{11}(X)\\
    &\hspace{1cm}+ \left\{p_0(X)\right\}\left[(1-R)-\left\{1-\rho(X)\right\}\right]\mu_{11}(X)\\
    &\hspace{1cm}+\left\{p_0(X)\right\}\left\{1-\rho(X)\right\}\left[\left\{\frac{AR[YC-p_1(X)\mu_{11}(X)]-AR[C-p_1(X)]\mu_{11}(X)}{\pi(X)\rho(X)p_1(X)}\right\}\right]
\end{aligned}$} \\ \\
\hline\\
    $\psi_D(P)$ & {$\begin{aligned}
     &\left\{p_1(X)-p_0(X)\right\}\left\{1-\rho(X)\right\}-\psi_D(P)\\
    &\hspace{1cm}+ \left[\frac{AR\left\{C-p_1(X)\right\}}{\pi(X)\rho(X)}-\frac{(1-A)R\left\{C-p_0(X)\right\}}{(1-\pi(X))\rho(X)}\right]\left\{1-\rho(X)\right\}\\
    &\hspace{1cm}+ \left\{p_1(X)-p_0(X)\right\}\left[(1-R)-\left\{1-\rho(X)\right\}\right]
\end{aligned}$}
\end{tabular}
\caption{EIFs for each component of $\psi^0_{1, 10}$.}
\label{eif_table}
\end{table}
\newpage
\noindent Before applying Lemmas \ref{lemma: eif_1_proof}, \ref{lemma: eif_0_proof}, and \ref{lemma: eif_denom} to write down the entire EIF for $\psi$, we re-write influence functions as summarized in Table \ref{eif_table} in a manner analogous to \citet{jiang_multiply_2022} to facilitate comparisons between the EIF in our setting and in theirs, which does not incorporate a separate target population. 

\subsection*{Rewriting the influence functions as in \citet{jiang_multiply_2022}}

We can rewrite the EIFs above in a manner similar to \citet{jiang_multiply_2022} using quantities $\psi_{f(Y_{a, r}, C_{a, r}, X)}$, $\psi_{f(R)}$, and $\phi^0_{0, 10}$ defined in the main manuscript.

\noindent For example, we can see in Table \ref{eif_table} that the EIF of $\psi^1_N(P)$ is
\begin{align*}
    \varphi^1_N(Z) & = p_1(X)\left\{1-\rho(X)\right\}\mu_{11}(X)-\psi^1_N(P)\\
    &\hspace{1cm}+ \left[\frac{AR\left\{C-p_1(X)\right\}}{\pi(X)\rho(X)}\right]\left\{1-\rho(X)\right\}\mu_{11}(X)\\
    &\hspace{1cm}+ \left\{p_1(X)\right\}\left[(1-R)-\left\{1-\rho(X)\right\}\right]\mu_{11}(X)\\
    &\hspace{1cm}+\left\{p_1(X)\right\}\left\{1-\rho(X)\right\}\left[\left\{\frac{AR[YC-p_1(X)\mu_{11}(X)]-AR[C-p_1(X)]\mu_{11}(X)}{\pi(X)\rho(X)p_1(X)}\right\}\right].
\end{align*}
Rewriting this using the additional notation gives
\begin{align*}
    \varphi^1_N(P) & = \left\{p_1(X)\right\}\left\{1-\rho(X)\right\}\mu_{11}(X)-\psi^1_N(P)\\
    &\hspace{1cm}+ \left(\psi_{C_{1, 1}}-p_1(X)\right)\left\{1-\rho(X)\right\}\mu_{11}(X)\\
    &\hspace{1cm}+ \left\{p_1(X)\right\}(\psi_{1-R})\mu_{11}(X)\\
    &\hspace{1cm}+\left\{p_1(X)\right\}\left\{1-\rho(X)\right\}\left[\left\{\frac{\psi_{Y_{1, 1}, C_{1, 1}}-\psi_{C_{1, 1}}\mu_{11}(X)}{p_1(X)}\right\}\right].
\end{align*}
Combining this expression for $\varphi^1_N$ with an analogous one for $\varphi^0_N$, we have that 
\begin{align}
    & \varphi^1_N(Z)-\varphi^0_N(Z) \nonumber\\
    &= \left\{p_1(X)-p_0(X)\right\}\left\{1-\rho(X)\right\}\mu_{11}(X)-\left(\psi^1_N-\psi^0_N\right)\nonumber\\
    &\hspace{1cm}+\left\{\psi_{C_{1, 1}}-\psi_{C_{0, 1}}-p_1(X)+p_0(X)\right\}\left\{1-\rho(X)\right\}\mu_{11}(X)\nonumber\\
    &\hspace{1cm}+\left\{p_1(X)-p_0(X)\right\}(\psi_{1-R})\mu_{11}(X)\nonumber\\
    &\hspace{1cm}+\left\{p_1(X)-p_0(X)\right\}\left\{1-\rho(X)\right\}\left\{\frac{\psi_{Y_{1, 1}, C_{1, 1}}-\psi_{C_{1, 1}}\mu_{11}(X)}{p_1(X)}\right\}\nonumber\\
    &= \left\{\psi_{C_{1, 1}}-\psi_{C_{0, 1}}\right\}\left\{1-\rho(X)\right\}\mu_{11}(X)-(\psi^1_N-\psi^0_N)\nonumber\\
    &\hspace{1cm}+\left\{1-\frac{p_0(X)}{p_1(X)}\right\}\left\{1-\rho(X)\right\}\left\{\psi_{Y_{1, 1}, C_{1, 1}}-\psi_{C_{1, 1}}\mu_{11}(X)\right\}\nonumber\\
    &\hspace{1cm}+e_{10}(X)(\psi_{1-R})\mu_{11}(X)\nonumber\\
    &= \frac{e_{10}(X)}{p_1(X)}\left\{1-\rho(X)\right\}\psi_{Y_{1, 1}C_{1, 1}}-(\psi^1_N-\psi^0_N)-\mu_{11}(X)(1-\rho(X))\left\{\psi_{C_{0, 1}}-\frac{p_0(X)}{p_1(X)}\psi_{C_{1, 1}}\right\}\nonumber\\
    &\hspace{1cm}+e_{10}(X)(\psi_{1-R})\mu_{11}(X)\nonumber\\
    &= \phi^0_{1, 10}-(\psi^1_N-\psi^0_N), \label{phi_eqn}
\end{align}
That, as stated in the main text, $\phi^0_{1, 10}$ is centered at $\psi^1_N-\psi^0_N$ follows directly from the fact that the EIF $\varphi^1_N(Z)-\varphi^0_N(Z)$ is mean-zero.  Turning to the denominator, $\psi_D$, recall that we wrote its EIF as
\begin{align*}
    \varphi_D(Z)&=\left\{p_1(X)-p_0(X)\right\}\left\{1-\rho(X)\right\}-\psi_D(P)\\
    &\hspace{1cm}+ \left[\frac{AR\left\{C-p_1(X)\right\}}{\pi(X)\rho(X)}-\frac{(1-A)R\left\{C-p_0(X)\right\}}{(1-\pi(X))\rho(X)}\right]\left\{1-\rho(X)\right\}\\
    &\hspace{1cm}+ \left\{p_1(X)-p_0(X)\right\}\left[(1-R)-\left\{1-\rho(X)\right\}\right].
\end{align*}

\noindent Again applying the alternative notation, we can rewrite this quantity as
\begin{align*}
    \varphi_D(Z) &= \left\{p_1(X)-p_0(X)\right\}\left\{1-\rho(X)\right\}-\psi_D(P)\\
    &\hspace{1cm}+ \left[\psi_{C_{1, 1}}-\psi_{C_{0, 1}}-(p_1(X)-p_0(X))\right]\left\{1-\rho(X)\right\}\\
    &\hspace{1cm}+ \left\{p_1(X)-p_0(X)\right\}(\psi_{1-R})\\
    &= \left[\psi_{C_{1, 1}}-\psi_{C_{0, 1}}\right]\left\{1-\rho(X)\right\}+\left\{p_1(X)-p_0(X)\right\}(\psi_{1-R})-\psi_D(P)\\
    &= \lambda_{10}.
\end{align*}
For notational simplicity, we let $N = \psi^1_N(P)-\psi^0_N(P)$ and $\varphi_N(Z)=\varphi^1_N(Z)-\varphi^0_N(Z)$. (We know by Lemma \ref{EIF_diff} that $\varphi_N(Z)$ is the EIF for $N$.) Now, applying Lemma \ref{eif_fraction}, we can rewrite the entire EIF for $\psi(P)$ as:
\begin{align*}
    \varphi(Z) &= \frac{\phi^0_{1, 10}-N}{\psi_D(P)} \\
    &\quad\;-\psi(P)\left(\frac{\lambda_{10}-\psi_D(P)}{\psi_D(P)}\right)\\
    &= \frac{\phi^0_{1, 10}}{\psi_D(P)}-\psi(P)\left(\frac{\lambda_{10}}{\psi_D(P)}\right).
\end{align*}
The above EIF is analogous to that of \citet{jiang_multiply_2022} on p. S22 in their Supplementary Material. 
\end{proof}

\subsection{Deriving the EIF of \texorpdfstring{$\psi^0_{0, 10}$}{TEXT}}

Above, we derived the EIF for one component of our overall parameter

\begin{equation*}
\psi^0_{10} = \frac{E\left\{\left[p_1(X)-p_0(X)\right](1-\rho(X))\left[\mu_{11}(X)-\mu_{00}(X)\right]\right\}}{E\left\{\left[p_1(X)-p_0(X)\right](1-\rho(X))\right\}}
\end{equation*}
given by 
\begin{equation*}
\psi^0_{1, 10} =  \frac{E\left[\left\{p_1(X)-p_0(X)\right\}\left\{1-\rho(X)\right\}\mu_{11}(X)\right]}{E\left[\left\{p_1(X)-p_0(X)\right\}\left\{(1-\rho(X)\right\}\right]}
\end{equation*}
which, under our identification assumptions, corresponds to the potential outcome under assignment to treatment for compliers in the target population. We analogously derive the EIF for 
\begin{equation*}
\psi^0_{0, 10} = \frac{E\left[\left\{p_1(X)-p_0(X)\right\}\left\{1-\rho(X)\right\}\mu_{00}(X)\right]}{E\left[\left\{p_1(X)-p_0(X)\right\}\left\{(1-\rho(X)\right\}\right]},
\end{equation*}
which is provided in Lemma \ref{eif_lemma_2}.
\begin{lemma}[EIF of $\psi^0_{0, 10}$]
\label{eif_lemma_2}
The EIF of  $\psi^0_{0, 10}$ is given by
\begin{align*}
   \varphi_{\psi^0_{0, 10}} = \frac{\phi^0_{0, 10}}{E\left[e_{10}(X)\left\{1-\rho(X)\right\}\right]}-\psi^0_{0, 10}\left(\frac{\lambda_{10}}{E\left[e_{10}(X)(1-\rho(X))\right]}\right). 
\end{align*}
\end{lemma}

\begin{proof}[Proof of Lemma \ref{eif_lemma_2}]
The calculations necessary to derive $\varphi_{\psi^0_{0, 10}}$ are almost identical as those used to derive the EIF of $\psi^0_{1, 10}$, where the only distinction is the inclusion of $\mu_{00}(X)$ in lieu of $\mu_{11}(X)$.
\end{proof}

\subsection{Deriving the EIF of \texorpdfstring{$\psi^0_{10}$}{TEXT}}

\begin{proof}[Proof of Theorem \ref{eif_thm}]

Summarizing previous calculations, we write 
\begin{align*}
\psi^0_{10} &= \frac{E\left\{\left[p_1(X)-p_0(X)\right](1-\rho(X))\left[\mu_{11}(X)-\mu_{00}(X)\right]\right\}}{E\left\{\left[p_1(X)-p_0(X)\right](1-\rho(X))\right\}}\\
&= \psi^0_{1, 10}-\psi^0_{0, 10}.
\end{align*}
We showed above that the EIF for $\psi^0_{1, 10}$ is given by 
\begin{equation*}
 \varphi_{\psi^0_{1, 10}}=\frac{\phi^0_{1, 10}}{E\left[e_{10}(X)\left\{1-\rho(X)\right\}\right]}-\psi^0_{1, 10}\left(\frac{\lambda_{10}}{E\left[e_{10}(X)(1-\rho(X))\right]}\right)
\end{equation*}
and that the EIF for $\psi^0_{0, 10}$ is
\begin{equation*}
\varphi_{\psi^0_{0, 10}} = \frac{\phi^0_{0, 10}}{E\left[e_{10}(X)\left\{1-\rho(X)\right\}\right]}-\psi^0_{0, 10}\left(\frac{\lambda_{10}}{E\left[e_{10}(X)(1-\rho(X))\right]}\right)
\end{equation*}
where the additional notation $\psi_{f(R)}=f(R)-E[f(R)|X]$ has been applied. We can now apply Lemma \ref{EIF_diff} to derive the EIF for the entire parameter $\psi$ as $\varphi_{\psi^0_{1, 10}}-\varphi_{\psi^0_{0, 10}}$:

\begin{equation*}
\varphi_{\psi^0_{10}}=\frac{\phi^0_{1, 10}-\phi^0_{0, 10}-\psi^0_{10}\lambda_{10}}{E\left[e_{10}(X)\left\{1-\rho(X)\right\}\right]}
\end{equation*}
as stated in Theorem \ref{eif_thm} of the main text.

\end{proof}

\section{Proofs of Asymptotic Results}

\subsection{Consistency (Theorem \ref{consistency_theorem})}\label{consistency_proofs}

In this section we ultimately prove Theorem \ref{consistency_theorem}, which summarizes the conditions under which our EIF-based estimator 

\begin{equation*}
\hat{\tau}_{\text{EIF}} = \frac{\mathbb{P}_n[\hat{\phi}^*_{1, 10}-\hat{\phi}^*_{0, 10}]}{\mathbb{P}_n\left(\hat{\lambda}_{10}\right)}
\end{equation*}
is consistent for
\begin{equation*}
\psi^0_{10} = \frac{E\left\{\left[p_1(X)-p_0(X)\right](1-\rho(X))\left[\mu_{11}(X)-\mu_{00}(X)\right]\right\}}{E\left\{\left[p_1(X)-p_0(X)\right](1-\rho(X))\right\}}.
\end{equation*}
Relating $\hat{\tau}_{\text{EIF}}$ to $\psi^0_{10}$, we see that we can break up our calculation of the bias into comparisons of:
\begin{enumerate}
\item $\mathbb{P}_n[\hat{\phi}^*_{1, 10}]$ vs. $E\left\{\left[p_1(X)-p_0(X)\right](1-\rho(X))\left[\mu_{11}(X)\right]\right\}$, 
\item $\mathbb{P}_n[\hat{\phi}^*_{0, 10}]$ vs. $E\left\{\left[p_1(X)-p_0(X)\right](1-\rho(X))\left[\mu_{00}(X)\right]\right\}$, and
\item $\mathbb{P}_n\left[\hat{\lambda}_{10}\right]$ vs. $E\left\{\left[p_1(X)-p_0(X)\right](1-\rho(X))\right\}$
\end{enumerate}
By the continuous mapping theorem (see, e.g., Theorem 5.9 of \citet{boos_stefanski}), establishing consistency of the above three estimators to their desired limits does the same for $\hat{\tau}_{EIF}$. We establish each such consistency result with separate lemmas before proving Theorem \ref{consistency_theorem} directly. For simplicity, we assume below that the various nuisance parameters were estimated using a separate sample than the one whose empirical measure is denoted by $\mathbb{P}_n$. In practice, we suggest proceeding via $K>2$-fold cross-fitting. Our simplification is equivalent to cross-fitting with $K=2$ folds. Also, as in the statement of Theorem~\ref{consistency_theorem}, for any nuisance parameter $\theta$ estimated from a sample of size $n$ by $\hat{\theta}$, we let $\tilde{\theta}$ denote the probability limit of $\hat{\theta}$, i.e., $\hat{\theta}$ is consistent for $\tilde{\theta}$. Of course, it may or may not be the case that $\tilde{\theta}=\theta$.

\begin{lemma}[Consistency of $\mathbb{P}_n(\hat{\phi}^*_{1, 10})$]
\label{lemma: consistency_proof_1}
Under the conditions given in Theorem \ref{consistency_theorem}, we have that $\mathbb{P}_n[\hat{\phi}^*_{1, 10}]$ is consistent for $E\left\{\left[p_1(X)-p_0(X)\right](1-\rho(X))\left[\mu_{11}(X)\right]\right\}$.
\end{lemma}

\begin{proof}[Proof of Lemma \ref{lemma: consistency_proof_1}]
\noindent Suppressing dependence on $X$, we follow the approach of \citet{jiang_multiply_2022} in writing $\mathbb{P}_n[\hat{\phi}^*_{1, 10}]$ as 
\begin{equation*}
\mathbb{P}_n\left[(1-\hat{\rho})\hat{e}_{10}\frac{A}{\hat{\pi}}\frac{R}{\hat{\rho}}\frac{C}{\hat{p}_1}\left\{Y-\hat{\mu}_{11}\right\}+(1-\hat{\rho})\hat{\mu}_{11}\left\{\hat{\psi}_{C_{1, 1}}-\hat{\psi}_{C_{0, 1}}\right\}+\hat{e}_{10}(\hat{\psi}_{1-R})\hat{\mu}_{11}\right].
\end{equation*}
Under our assumptions, this quantity is consistent for 

\begin{equation*}
E\left[(1-\tilde{\rho})\tilde{e}_{10}\frac{A}{\tilde{\pi}}\frac{R}{\tilde{\rho}}\frac{C}{\tilde{p}_1}\left\{Y-\tilde{\mu}_{11}\right\}+(1-\tilde{\rho})\tilde{\mu}_{11}\left\{\tilde{\psi}_{C_{1, 1}}-\tilde{\psi}_{C_{0, 1}}\right\}+\tilde{e}_{10}(\tilde{\psi}_{1-R})\tilde{\mu}_{11}\right].
\end{equation*}
Note, for instance, that 
\begin{align}
&E\left[(1-\tilde{\rho})\tilde{e}_{10}\frac{A}{\tilde{\pi}}\frac{R}{\tilde{\rho}}\frac{C}{\tilde{p}_1}\left\{Y-\tilde{\mu}_{11}\right\}\right]\nonumber\\
&= E\left[(1-\tilde{\rho}(X))\tilde{e}_{10}(X)\frac{A}{\tilde{\pi}(X)}\frac{R}{\tilde{\rho}(X)}\frac{C}{\tilde{p}_1(X)}\left\{Y-\tilde{\mu}_{11}(X)\right\}\right]\nonumber\\[6pt]
&= E\left[(1-\tilde{\rho}(X))\tilde{e}_{10}(X)\frac{A}{\tilde{\pi}(X)}\frac{R}{\tilde{\rho}(X)}\frac{C}{\tilde{p}_1(X)}Y\right]-E\left[(1-\tilde{\rho}(X))\tilde{e}_{10}(X)\frac{A}{\tilde{\pi}(X)}\frac{R}{\tilde{\rho}(X)}\frac{C}{\tilde{p}_1(X)}\tilde{\mu}_{11}(X)\right]\nonumber\\[6pt]
&= E\left[\frac{(1-\tilde{\rho}(X))\tilde{e}_{10}(X)}{\tilde{\pi}(X)\tilde{\rho}(X)\tilde{p}_1(X)}E\left[ARCY|X\right]\right]-E\left[\frac{(1-\tilde{\rho}(X))\tilde{e}_{10}(X)\tilde{\mu}_{11}(X)}{\tilde{\pi}(X)\tilde{\rho}(X)\tilde{p}_1(X)}E\left[ARC|X\right]\right]\nonumber\\[6pt]
&= E\left[\frac{(1-\tilde{\rho}(X))\tilde{e}_{10}(X)}{\tilde{\pi}(X)\tilde{\rho}(X)\tilde{p}_1(X)}E\left[A|X\right]E\left[R|X\right]E\left[C|X, A=1, R=1\right]E\left[Y|X, A=1, C=1, R=1\right]\right]\nonumber\\[6pt]
&\hspace{5mm}-E\left[\frac{(1-\tilde{\rho}(X))\tilde{e}_{10}(X)\tilde{\mu}_{11}(X)}{\tilde{\pi}(X)\tilde{\rho}(X)\tilde{p}_1(X)}E\left[A|X\right]E\left[R|X\right]E\left[C|X, A=1, R=1\right]\right]\nonumber\\[6pt]
&= E\left[\frac{(1-\tilde{\rho}(X))\tilde{e}_{10}(X)}{\tilde{\pi}(X)\tilde{\rho}(X)\tilde{p}_1(X)}\pi(X)\rho(X)p_1(X)\mu_{11}(X)\right]-E\left[\frac{(1-\tilde{\rho}(X))\tilde{e}_{10}(X)\tilde{\mu}_{11}(X)}{\tilde{\pi}(X)\tilde{\rho}(X)\tilde{p}_1(X)}\pi(X)\rho(X)p_1(X)\right]\nonumber\\[6pt]
&= E\left[(1-\tilde{\rho})\tilde{e}_{10}\frac{\pi}{\tilde{\pi}}\frac{\rho}{\tilde{\rho}}\frac{p_1}{\tilde{p}_1}\left\{\mu_{11}-\tilde{\mu}_{11}\right\}\right], \label{eq: bias_simplification_1}
\end{align}
where the third-to-last equality followed from repeated application of our conditional independence assumptions and iterated expectations. We apply analogous calculations below when replacing the expected value of indicator variables with their corresponding true nuisance parameter.
Thus, the asymptotic bias of $\mathbb{P}_n[\hat{\phi}^*_{1, 10}]$ is given by
\begin{align*}
&E\left[(1-\tilde{\rho})\tilde{e}_{10}\frac{A}{\tilde{\pi}}\frac{R}{\tilde{\rho}}\frac{C}{\tilde{p}_1}\left\{Y-\tilde{\mu}_{11}\right\}+(1-\tilde{\rho})\tilde{\mu}_{11}\left\{\tilde{\psi}_{C_{1, 1}}-\tilde{\psi}_{C_{0, 1}}\right\}+\tilde{e}_{10}(\tilde{\psi}_{1-R})\tilde{\mu}_{11}\right]-E[e_{10}(1-\rho)\mu_{11}]\\
&= E\left[(1-\tilde{\rho})\tilde{e}_{10}\frac{\pi}{\tilde{\pi}}\frac{\rho}{\tilde{\rho}}\frac{p_1}{\tilde{p}_1}\left\{\mu_{11}-\tilde{\mu}_{11}\right\}+(1-\tilde{\rho})\tilde{\mu}_{11}\left\{\tilde{\psi}_{C_{1, 1}}-\tilde{\psi}_{C_{0, 1}}\right\}+\tilde{e}_{10}(\tilde{\psi}_{1-R})\tilde{\mu}_{11}-e_{10}(1-\rho)\mu_{11}\right].
\end{align*}
We focus our calculations on the terms within the expectation and, therefore, omit the outer expectation operator below to save space. We can break down the terms in the inner expectation as follows:
\begin{align}
&(1-\tilde{\rho})\tilde{e}_{10}\frac{\pi}{\tilde{\pi}}\frac{\rho}{\tilde{\rho}}\frac{p_1}{\tilde{p}_1}\left\{\mu_{11}-\tilde{\mu}_{11}\right\}+(1-\tilde{\rho})\tilde{\mu}_{11}\left\{\tilde{\psi}_{C_{1, 1}}-\tilde{\psi}_{C_{0, 1}}\right\}+\tilde{e}_{10}(\tilde{\psi}_{1-R})\tilde{\mu}_{11}-e_{10}(1-\rho)\mu_{11} \nonumber\\
&= \underbrace{(1-\tilde{\rho})\tilde{p}_{1}\frac{\pi}{\tilde{\pi}}\frac{\rho}{\tilde{\rho}}\frac{p_1}{\tilde{p}_1}\left\{\mu_{11}-\tilde{\mu}_{11}\right\}+(1-\tilde{\rho})\tilde{\mu}_{11}\left\{\tilde{\psi}_{C_{1, 1}}\right\}+\tilde{p}_{1}(\tilde{\psi}_{1-R})\tilde{\mu}_{11}-p_{1}(1-\rho)\mu_{11}}_{\text{\circled{1}}} \nonumber\\
&\hspace{.25cm}- \underbrace{\left((1-\tilde{\rho})\tilde{p}_{0}\frac{\pi}{\tilde{\pi}}\frac{\rho}{\tilde{\rho}}\frac{p_1}{\tilde{p}_1}\left\{\mu_{11}-\tilde{\mu}_{11}\right\}+(1-\tilde{\rho})\tilde{\mu}_{11}\left\{\tilde{\psi}_{C_{0, 1}}\right\}+\tilde{p}_{0}(\tilde{\psi}_{1-R})\tilde{\mu}_{11}-p_{0}(1-\rho)\mu_{11}\right)}_{\text{\circled{2}}} \label{eq: break_down_bias_1}
\end{align}

\noindent Ultimately, we want to show that both sets of terms are equal to zero under the consistency conditions of Theorem \ref{consistency_theorem}. We start by considering only the ``$p_1$ terms", i.e. \circled{1}, in Equation (\ref{eq: break_down_bias_1}):
\begin{equation}
(1-\tilde{\rho})\tilde{p}_1\frac{\pi\rho p_1}{\tilde{\pi}\tilde{\rho}\tilde{p}_1}\left\{\mu_{11}-\tilde{\mu}_{11}\right\}+(1-\tilde{\rho})\tilde{\mu}_{11}\left\{\tilde{\psi}_{C_{1, 1}}\right\}+\tilde{p}_1\left[(1-\rho)-(1-\tilde{\rho})\right]\tilde{\mu}_{11}-(1-\rho)(p_1)\mu_{11}. \label{eq: num_1_bias_1_unsimp}
\end{equation}
Equation (\ref{eq: num_1_bias_1_unsimp}) can be simplified as:
\begin{align}
&(1-\tilde{\rho})\tilde{p}_1\frac{\pi\rho p_1}{\tilde{\pi}\tilde{\rho}\tilde{p}_1}\left\{\mu_{11}-\tilde{\mu}_{11}\right\}+(1-\tilde{\rho})\tilde{\mu}_{11}\left\{\tilde{\psi}_{C_{1, 1}}\right\}+\tilde{p}_1\left[(1-\rho)-(1-\tilde{\rho})\right]\tilde{\mu}_{11}-(1-\rho)(p_1)\mu_{11}\nonumber\\ \nonumber\\
&= (1-\tilde{\rho})(\tilde{p}_1)\left\{\frac{\pi\rho p_1[\mu_{11}-\tilde{\mu}_{11}]}{\tilde{\pi}\tilde{\rho}\tilde{p}_1}\right\}+(1-\tilde{\rho})\tilde{\mu}_{11}\left\{\frac{\pi\rho[p_1-\tilde{p}_1]}{\tilde{\pi}\tilde{\rho}}\right\}+\tilde{p}_1\tilde{\mu}_{11}\left[(1-\rho)-(1-\tilde{\rho})\right]\nonumber\\
&\hspace{0.65cm} + \tilde{p}_1\tilde{\mu}_{11}(1-\tilde{\rho})-p_1(1-\rho)(\mu_{11})\nonumber\\ \nonumber\\
&= \frac{\pi\rho p_1\left[(1-\tilde{\rho})\mu_{11}-(1-\tilde{\rho})\tilde{\mu}_{11}\right]}{\tilde{\pi}\tilde{\rho}} + \frac{\pi\rho\left[(1-\tilde{\rho})\tilde{\mu}_{11}p_1-(1-\tilde{\rho})\tilde{\mu}_{11}\tilde{p}_1\right]}{\tilde{\pi}\tilde{\rho}}+\left[\tilde{p}_1\tilde{\mu}_{11}(1-\rho)-\tilde{p}_1\tilde{\mu}_{11}(1-\tilde{\rho})\right]\nonumber\\
&\hspace{0.65cm} +\tilde{p}_1\tilde{\mu}_{11}(1-\tilde{\rho})-p_1(1-\rho)(\mu_{11})\nonumber\\ \nonumber\\
&= \frac{\pi\rho p_1 (1-\tilde{\rho})\mu_{11}}{\tilde{\pi}\tilde{\rho}}-\frac{\pi\rho(1-\tilde{\rho})\tilde{\mu}_{11}\tilde{p}_1}{\tilde{\pi}\tilde{\rho}}+\tilde{p}_1\tilde{\mu}_{11}(1-\rho)-p_1(1-\rho)\mu_{11}\nonumber\\ \nonumber\\
&= \frac{\pi\rho p_1 (1-\tilde{\rho})\mu_{11}-\pi\rho(1-\tilde{\rho})\tilde{\mu}_{11}\tilde{p}_1+\tilde{\pi}\tilde{\rho}\tilde{p}_1\tilde{\mu}_{11}(1-\rho)-\tilde{\pi}\tilde{\rho}p_1(1-\rho)\mu_{11}}{\tilde{\pi}\tilde{\rho}}\nonumber\\ \nonumber\\
&= \frac{\left[\pi\rho(1-\tilde{\rho})-\tilde{\pi}\tilde{\rho}(1-\rho)\right]\left[\mu_{11}p_1-\tilde{\mu}_{11}\tilde{p}_1\right]}{\tilde{\pi}\tilde{\rho}}. \label{eq: num_1_bias_1}
\end{align}
Note that the numerator of this term does indeed equal to zero under the consistency conditions of Theorem \ref{consistency_theorem}: the first term under conditions 1 or 2 and the second term under condition 3. Since only one such term needs to go to zero, the robustness in consistency property holds.\\
\noindent Now we turn to the ``$p_0$ terms", i.e. \circled{2}, referenced in Equation (\ref{eq: num_1_bias_1_unsimp}). Expanding this component out, we obtain:
\begin{align}
&(1-\tilde{\rho})\tilde{p}_0\frac{\pi\rho p_1}{\tilde{\pi}\tilde{\rho}\tilde{p}_1}\left\{\mu_{11}-\tilde{\mu}_{11}\right\}+(1-\tilde{\rho})\tilde{\mu}_{11}\left\{\tilde{\psi}_{C_{0, 1}}\right\}+\tilde{p}_0\left[(1-\rho)-(1-\tilde{\rho})\right]\tilde{\mu}_{11}-(1-\rho)(p_0)\mu_{11} \nonumber\\  \nonumber\\
&= (1-\tilde{\rho})\tilde{p}_0\left\{\frac{\pi\rho p_1[\mu_{11}-\tilde{\mu}_{11}]}{\tilde{\pi}\tilde{\rho}\tilde{p}_1}\right\}+(1-\tilde{\rho})\tilde{\mu}_{11}\left\{\frac{(1-\pi)\rho[p_0-\tilde{p}_0]}{(1-\tilde{\pi})\tilde{\rho}}\right\}+\tilde{p}_0\tilde{\mu}_{11}\left[(1-\rho)-(1-\tilde{\rho})\right] \nonumber\\
&\hspace{0.65cm} + \tilde{p}_0\tilde{\mu}_{11}(1-\tilde{\rho})-p_0(1-\rho)(\mu_{11})\nonumber\\ \nonumber\\
&= (1-\tilde{\rho})\tilde{p}_0\left\{\frac{\pi\rho p_1[\mu_{11}-\tilde{\mu}_{11}]}{\tilde{\pi}\tilde{\rho}\tilde{p}_1}\right\}+(1-\tilde{\rho})\tilde{\mu}_{11}\left\{\frac{(1-\pi)\rho[p_0-\tilde{p}_0]}{(1-\tilde{\pi})\tilde{\rho}}\right\}+\tilde{p}_0\tilde{\mu}_{11}(1-\rho)-p_0(1-\rho)\mu_{11} \nonumber\\
&= \underbrace{\tilde{p}_0\left\{\frac{\pi\rho p_1[\mu_{11}-\tilde{\mu}_{11}]}{\tilde{\pi}\tilde{\rho}\tilde{p}_1}\right\}+\tilde{\mu}_{11}\left\{\frac{(1-\pi)\rho[p_0-\tilde{p}_0]}{(1-\tilde{\pi})\tilde{\rho}}\right\}+\tilde{p}_0\tilde{\mu}_{11}-p_0\mu_{11}}_{\text{\circled{2a}}} \nonumber\\
&\hspace{.15cm}-\underbrace{\left(\tilde{\rho}\tilde{p}_0\left\{\frac{\pi\rho p_1[\mu_{11}-\tilde{\mu}_{11}]}{\tilde{\pi}\tilde{\rho}\tilde{p}_1}\right\}+\tilde{\rho}\tilde{\mu}_{11}\left\{\frac{(1-\pi)\rho[p_0-\tilde{p}_0]}{(1-\tilde{\pi})\tilde{\rho}}\right\}+\tilde{p}_0\tilde{\mu}_{11}\rho-p_0\rho\mu_{11}\right)}_{\text{\circled{2b}}}. \label{eq: num_1_bias_2_unsimp}
\end{align}

\noindent Turning to the terms in \circled{2a} of Equation (\ref{eq: num_1_bias_2_unsimp}), we have
\begin{align*}
&\tilde{p}_0\left\{\frac{\pi\rho p_1[\mu_{11}-\tilde{\mu}_{11}]}{\tilde{\pi}\tilde{\rho}\tilde{p}_1}\right\}+\tilde{\mu}_{11}\left\{\frac{(1-\pi)\rho[p_0-\tilde{p}_0]}{(1-\tilde{\pi})\tilde{\rho}}\right\}+\tilde{p}_0\tilde{\mu}_{11}-p_0\mu_{11}\\
&= \tilde{p}_0\left\{\frac{\pi\rho p_1}{\tilde{\pi}\tilde{\rho}\tilde{p}_1}-1\right\}[\mu_{11}-\tilde{\mu}_{11}]+\tilde{p}_0[\mu_{11}-\tilde{\mu}_{11}]\\
&\hspace{0.65cm} + \tilde{\mu}_{11}\left\{\frac{(1-\pi)\rho}{(1-\tilde{\pi})\tilde{\rho}}-1\right\}[p_0-\tilde{p}_0]+\tilde{\mu}_{11}[p_0-\tilde{p}_0]\\
&\hspace{0.65cm}+\tilde{p}_0\tilde{\mu}_{11}-p_0\mu_{11}\\
&= \tilde{p}_0\left\{\frac{\pi\rho p_1}{\tilde{\pi}\tilde{\rho}\tilde{p}_1}-1\right\}[\mu_{11}-\tilde{\mu}_{11}]+\tilde{\mu}_{11}\left\{\frac{(1-\pi)\rho}{(1-\tilde{\pi})\tilde{\rho}}-1\right\}[p_0-\tilde{p}_0]\\
&\hspace{0.65cm} + \tilde{p}_0\mu_{11}-\tilde{p}_0\tilde{\mu}_{11}+\tilde{\mu}_{11}p_0-\tilde{\mu}_{11}\tilde{p}_0+\tilde{p}_0\tilde{\mu}_{11}-p_0\mu_{11}\\
&= \tilde{p}_0\left\{\frac{\pi\rho p_1}{\tilde{\pi}\tilde{\rho}\tilde{p}_1}-1\right\}[\mu_{11}-\tilde{\mu}_{11}]+\tilde{\mu}_{11}\left\{\frac{(1-\pi)\rho}{(1-\tilde{\pi})\tilde{\rho}}-1\right\}[p_0-\tilde{p}_0]\\
&\hspace{0.65cm} + \tilde{p}_0\mu_{11}-\tilde{p}_0\tilde{\mu}_{11}+\tilde{\mu}_{11}p_0-p_0\mu_{11}\\
&= \tilde{p}_0\left\{\frac{\pi\rho p_1}{\tilde{\pi}\tilde{\rho}\tilde{p}_1}-1\right\}[\mu_{11}-\tilde{\mu}_{11}]+\tilde{\mu}_{11}\left\{\frac{(1-\pi)\rho}{(1-\tilde{\pi})\tilde{\rho}}-1\right\}[p_0-\tilde{p}_0] + (\tilde{p}_0-p_0)(\mu_{11}-\tilde{\mu}_{11}).
\end{align*}
As in the previous bias calculation, we can see that each of the three products above are zero under any one of the three conditions in Theorem \ref{consistency_theorem}. Now turning to the terms in \circled{2b} of Equation (\ref{eq: num_1_bias_2_unsimp}):
\begin{align*}
&\tilde{\rho}\tilde{p}_0\left\{\frac{\pi\rho p_1[\mu_{11}-\tilde{\mu}_{11}]}{\tilde{\pi}\tilde{\rho}\tilde{p}_1}\right\}+\tilde{\rho}\tilde{\mu}_{11}\left\{\frac{(1-\pi)\rho[p_0-\tilde{p}_0]}{(1-\tilde{\pi})\tilde{\rho}}\right\}+\rho\tilde{p}_0\tilde{\mu}_{11}-p_0\rho\mu_{11}\\
&= \tilde{p}_0\left\{\frac{\pi\rho p_1[\mu_{11}-\tilde{\mu}_{11}]}{\tilde{\pi}\tilde{p}_1}\right\}+\tilde{\mu}_{11}\left\{\frac{(1-\pi)\rho[p_0-\tilde{p}_0]}{(1-\tilde{\pi})}\right\}+\rho\tilde{p}_0\tilde{\mu}_{11}-p_0\rho\mu_{11}\\
&= \tilde{p}_0\rho\left\{\frac{\pi p_1}{\tilde{\pi}\tilde{p}_1}-1\right\}[\mu_{11}-\tilde{\mu}_{11}]+\tilde{p}_0\rho\mu_{11}-\tilde{p}_0\rho\tilde{\mu}_{11}+\tilde{\mu}_{11}\left\{\frac{(1-\pi)\rho[p_0-\tilde{p}_0]}{(1-\tilde{\pi})}\right\}+\rho\tilde{p}_0\tilde{\mu}_{11}-p_0\rho\mu_{11}\\
&= \tilde{p}_0\rho\left\{\frac{\pi p_1}{\tilde{\pi}\tilde{p}_1}-1\right\}[\mu_{11}-\tilde{\mu}_{11}]+\tilde{\mu}_{11}\rho\left\{\frac{(1-\pi)}{(1-\tilde{\pi})}-1\right\}(p_0-\tilde{p}_0)+\tilde{\mu}_{11}\rho p_0-\tilde{\mu}_{11}\rho\tilde{p}_0+\tilde{p}_0 \rho \mu_{11}-p_0\rho \mu_{11}\\
&= \tilde{p}_0\rho\left\{\frac{\pi p_1}{\tilde{\pi}\tilde{p}_1}-1\right\}[\mu_{11}-\tilde{\mu}_{11}]+\tilde{\mu}_{11}\rho\left\{\frac{(1-\pi)}{(1-\tilde{\pi})}-1\right\}(p_0-\tilde{p}_0)+\rho(\tilde{\mu}_{11}-\mu_{11})(p_0-\tilde{p}_0).
\end{align*}
\noindent We combine all of these results together to obtain the following expression for the asymptotic bias of $\mathbb{P}_n[\hat{\phi}^*_{1, 10}]$:
\begin{align}
&E\left[(1-\tilde{\rho})\tilde{e}_{10}\frac{\pi}{\tilde{\pi}}\frac{\rho}{\tilde{\rho}}\frac{p_1}{\tilde{p}_1}\left\{Y-\tilde{\mu}_{11}\right\}+(1-\tilde{\rho})\tilde{\mu}_{11}\left\{\tilde{\psi}_{C_{1, 1}}-\tilde{\psi}_{C_{0, 1}}\right\}+\tilde{e}_{10}(\tilde{\psi}_{1-R})\tilde{\mu}_{11}-e_{10}(1-\rho)\mu_{11}\right] \nonumber\\
&= E\left[\frac{\left[\pi\rho(1-\tilde{\rho})-\tilde{\pi}\tilde{\rho}(1-\rho)\right]\left[\mu_{11}p_1-\tilde{\mu}_{11}\tilde{p}_1\right]}{\tilde{\pi}\tilde{\rho}}\right] \nonumber\\
& \hspace{0.65cm} - E\left[\tilde{p}_0\left\{\frac{\pi\rho p_1}{\tilde{\pi}\tilde{\rho}\tilde{p}_1}-1\right\}[\mu_{11}-\tilde{\mu}_{11}]+\tilde{\mu}_{11}\left\{\frac{(1-\pi)\rho}{(1-\tilde{\pi})\tilde{\rho}}-1\right\}[p_0-\tilde{p}_0] + (\tilde{p}_0-p_0)(\mu_{11}-\tilde{\mu}_{11})\right] \nonumber\\
& \hspace{0.65cm} + E\left[\tilde{p}_0\rho\left\{\frac{\pi p_1}{\tilde{\pi}\tilde{p}_1}-1\right\}[\mu_{11}-\tilde{\mu}_{11}]+\tilde{\mu}_{11}\rho\left\{\frac{(1-\pi)}{(1-\tilde{\pi})}-1\right\}(p_0-\tilde{p}_0)+\rho(\tilde{\mu}_{11}-\mu_{11})(p_0-\tilde{p}_0)\right],\label{eq: asymp_bias_num_1}
\end{align}
each summand of which is zero under any one of the conditions in Theorem \ref{consistency_theorem}.
\end{proof}

\begin{lemma}[Consistency of $\mathbb{P}_n(\hat{\phi}^*_{0, 10})$]
\label{lemma: consistency_proof_2}
Under the conditions given in Theorem \ref{consistency_theorem}, we have that $\mathbb{P}_n[\hat{\phi}^*_{0, 10}]$ is consistent for $E\left\{\left[p_1(X)-p_0(X)\right](1-\rho(X))\left[\mu_{00}(X)\right]\right\}$.
\end{lemma}

\begin{proof}[Proof of Lemma \ref{lemma: consistency_proof_2}]

Calculations for the asymptotic bias of $\mathbb{P}_n[\hat{\phi}^*_{0, 10}]$ proceed almost identically as above, and we omit details for space considerations. Briefly, we again follow \citet{jiang_multiply_2022} in writing $\mathbb{P}_n[\hat{\phi}^*_{0, 10}]$ as:
\begin{equation*}
\mathbb{P}_n\left[(1-\hat{\rho})\hat{e}_{10}\frac{(1-A)}{1-\hat{\pi}}\frac{R}{\hat{\rho}}\frac{(1-C)}{1-\hat{p}_0}\left\{Y-\hat{\mu}_{00}\right\}+(1-\hat{\rho})\hat{\mu}_{00}\left\{\hat{\psi}_{C_{1, 1}}-\hat{\psi}_{C_{0, 1}}\right\}+\hat{e}_{10}(\hat{\psi}_{1-R})\hat{\mu}_{00}\right]
\end{equation*}
which is consistent for 
\begin{equation*}
E\left[(1-\tilde{\rho})\tilde{e}_{10}\frac{(1-\pi)}{1-\tilde{\pi}}\frac{\rho}{\tilde{\rho}}\frac{(1-p_0)}{1-\tilde{p}_0}\left\{\mu_{00}-\tilde{\mu}_{00}\right\}+(1-\tilde{\rho})\tilde{\mu}_{00}\left\{\tilde{\psi}_{C_{1, 1}}-\tilde{\psi}_{C_{0, 1}}\right\}+\tilde{e}_{10}(\tilde{\psi}_{1-R})\tilde{\mu}_{00}\right]
\end{equation*}
We can rewrite the above estimator as:
\begin{equation*}
\mathbb{P}_n\left[(1-\hat{\rho})\hat{e}_{10}\frac{(1-A)}{1-\hat{\pi}}\frac{R}{\hat{\rho}}\frac{(1-C)}{1-\hat{p}_0}\left\{Y-\hat{\mu}_{00}\right\}+(1-\hat{\rho})\hat{\mu}_{00}\left\{\hat{\psi}_{1-C_{0, 1}}-\hat{\psi}_{1-C_{1, 1}}\right\}+\hat{e}_{10}(\hat{\psi}_{1-R})\hat{\mu}_{00}\right]
\end{equation*}
Now, let $A^*=1-A$, $\pi^*=1-\pi$, $C^*=1-C$,  $p^*_0=1-p_0$, and $p^*_1=1-p_1$. Note first that
\begin{align*}
\psi_{1-C_{0, 1}}&=\frac{(1-A)R\left[(1-C)-(1-p_0(X))\right]}{(1-\pi(X))\rho(X)}+(1-p_0(X)) = \frac{A^*R\left[C^*-p_0^*(X)\right]}{\pi^*(X)\rho(X)}+p_0^*(X)\\
&\text{and}\\
\psi_{1-C_{1, 1}}&=\frac{AR\left[(1-C)-(1-p_1(X))\right]}{(1-\pi(X))\rho(X)}+(1-p_1(X)) = \frac{(1-A^*)R\left[C^*-p_1^*(X)\right]}{(1-\pi^*(X))\rho(X)}+p_1^*(X)
\end{align*}
Additionally, we have that
\begin{equation*}
e_{10}(X)=p_1(X)-p_0(X)=(1-p_0(X))-(1-p_1(X))=p^*_0(X)-p^*_1(X).
\end{equation*}
Putting these pieces together, we can write $\hat{\phi}^*_{0, 10}$ as
\begin{align*}
&(1-\hat{\rho})\left(\hat{p}^*_0(X)-\hat{p}^*_1(X)\right)\frac{A^*}{\hat{\pi}^*}\frac{R}{\hat{\rho}}\frac{C^*}{\hat{p}_0^*}\left\{Y-\hat{\mu}_{00}\right\}\\
&\hspace{0.65cm}+(1-\hat{\rho})\hat{\mu}_{00}\left\{\frac{A^*R\left[C^*-p_0^*(X)\right]}{\pi^*(X)\rho(X)}+p_0^*(X)-\frac{(1-A^*)R\left[C^*-p_1^*(X)\right]}{(1-\pi^*(X))\rho(X)}-p_1^*(X)\right\}\\
&\hspace{0.65cm}+\left(\hat{p}^*_0(X)-\hat{p}^*_1(X)\right)(\hat{\psi}_{1-R})\hat{\mu}_{00}.
\end{align*}
Altogether, we see that $\mathbb{P}_n[\hat{\phi}^*_{0, 10}]$ is structurally equivalent to $\mathbb{P}_n[\hat{\phi}^*_{1, 10}]$. That is, in particular, the formulas we derived above for the asymptotic bias of $\mathbb{P}_n[\hat{\phi}^*_{1, 10}]$
apply equally to $\mathbb{P}_n[\hat{\phi}^*_{0, 10}]$, where we replace $\mu_{11}$ with $\mu_{00}$, $\pi$ with $1-\pi$, and $p_1$ with $1-p_0$.
The important takeaway here is that the above replacements don't require any additional asymptotic conditions on the behavior of our nuisance function estimators. They rely on the same correctness of the outcome, treatment probability, and observed compliance probability models. This echoes the point in \citet{jiang_multiply_2022} that both pieces of the numerator are consistent under $\mathcal{M}_{\text{triple}}$.
\end{proof}

\begin{lemma}[Consistency of $\mathbb{P}_n(\hat{\lambda}_{10})$]
\label{lemma: consistency_proof_3}
Under the conditions given in Theorem \ref{consistency_theorem}, we have that $\mathbb{P}_n\left(\hat{\lambda}_{10}\right)$ is consistent for $E\left\{\left[p_1(X)-p_0(X)\right](1-\rho(X))\right\}$.
\end{lemma}

\begin{proof}[Proof of \ref{lemma: consistency_proof_3}]

As in the previous parts, we note that the bias of the denominator is consistent for 
\begin{align*}
&E\left(\left[\tilde{\psi}_{C_{1, 1}}-\tilde{\psi}_{C_{0, 1}}\right]\left\{1-\tilde{\rho}(X)\right\}+\left\{\tilde{p}_1(X)-\tilde{p}_0(X)\right\}\left[(1-R)-\left\{1-\tilde{\rho}(X)\right\}\right]\right) \\
&\hspace{0.5cm}-E\left\{\left[p_1(X)-p_0(X)\right](1-\rho(X))\right\}\\
&= E\left(\left[\tilde{\psi}_{C_{1, 1}}-\tilde{\psi}_{C_{0, 1}}\right]\left\{1-\tilde{\rho}(X)\right\}+\left\{\tilde{p}_1(X)-\tilde{p}_0(X)\right\}\left[(1-\rho)-\left\{1-\tilde{\rho}(X)\right\}\right] \right. \\
&\left.\hspace{1cm}\vphantom{\tilde{\psi}_{C_{1, 1}}}-\left[p_1(X)-p_0(X)\right](1-\rho(X))\right).
\end{align*}
We first consider terms involving $p_1$. Omitting notation for dependence on $X$ as above and working within the expectation, we have
\begin{align*}
\tilde{\psi}_{C_{1, 1}}(1-\tilde{\rho})+\tilde{p}_1[(1-\rho)-(1-\tilde{\rho})]-p_1(1-\rho) &= \tilde{\psi}_{C_{1, 1}}(1-\tilde{\rho})+\tilde{p}_1[(1-\rho)-(1-\tilde{\rho})]-p_1(1-\rho)\\
&\hspace{0.65cm} + p_1(1-\tilde{\rho})-p_1(1-\tilde{\rho})\\
&= \left[\tilde{\psi}_{C_{1, 1}}-p_1\right](1-\tilde{\rho})+(\tilde{p}_1-p_1)[(1-\rho)-(1-\tilde{\rho})].
\end{align*}
Now, observe that 
\begin{equation*}
\tilde{\psi}_{C_{1, 1}}-p_1 = \frac{\pi\rho[p_1-\tilde{p}_1]}{\tilde{\pi}\tilde{\rho}}+\tilde{p_1}-p_1 = \frac{\pi\rho[p_1-\tilde{p}_1]-\tilde{\pi}\tilde{\rho}(p_1-\tilde{p}_1)}{\tilde{\pi}\tilde{\rho}}=\frac{(\pi\rho-\tilde{\pi}\tilde{\rho})(p_1-\tilde{p}_1)}{\tilde{\pi}\tilde{\rho}}.
\end{equation*}
Thus, altogether, we can write the bias of terms involving $p_1$ as:
\begin{equation*}
\frac{(\pi\rho-\tilde{\pi}\tilde{\rho})(p_1-\tilde{p}_1)}{\tilde{\pi}\tilde{\rho}}(1-\tilde{\rho})+(\tilde{p}_1-p_1)(\tilde{\rho}-\rho).
\end{equation*}
An analogous set of calculations for terms involving $p_0$ yields the following:
\begin{equation*}
\frac{(\tilde{\pi}\tilde{\rho}-\pi\rho)(p_0-\tilde{p}_0)}{(1-\tilde{\pi})\tilde{\rho}}(1-\tilde{\rho})+(\tilde{p}_0-p_0)(\tilde{\rho}-\rho).
\end{equation*}
We combine these two results to obtain the following expression for the asymptotic bias of the denominator:
\begin{align}
&E\left[\frac{(\pi\rho-\tilde{\pi}\tilde{\rho})(p_1-\tilde{p}_1)}{\tilde{\pi}\tilde{\rho}}(1-\tilde{\rho})+(\tilde{p}_1-p_1)(\tilde{\rho}-\rho)\right] \nonumber \\
& -E\left[\frac{(\tilde{\pi}\tilde{\rho}-\pi\rho)(p_0-\tilde{p}_0)}{(1-\tilde{\pi})\tilde{\rho}}(1-\tilde{\rho})+(\tilde{p}_0-p_0)(\tilde{\rho}-\rho)\right]. \label{eq: denom_asymp_bias}
\end{align}
As with the proofs for the robustness in consistency for the other components of $\hat{\tau}_{EIF}$, we can see that each of the terms in Equation \ref{eq: denom_asymp_bias} are zero under any of the three consistency conditions of Theorem \ref{consistency_theorem}.
\end{proof}

\begin{proof}[Proof of Theorem \ref{consistency_theorem}]

Lemmas \ref{lemma: consistency_proof_1}, \ref{lemma: consistency_proof_2}, and \ref{lemma: consistency_proof_3} all establish that each of the three components of $\hat{\tau}_{EIF}$ are consistent under the consistency conditions of Theorem \ref{consistency_theorem}. Viewing each such component $\mathbb{P}_n(\hat{\phi}^*_{1, 10})$, $\mathbb{P}_n(\hat{\phi}^*_{0, 10})$, and $\mathbb{P}_n\left(\hat{\lambda}_{10}\right)$ as separate sequences of random variables indexed by $n$, we can apply the continuous mapping theorem to obtain the desired robustness conditions of $\hat{\tau}_{EIF}$ overall.

\end{proof}

\subsection{Proofs of Asymptotic Properties}
\label{ral_proofs}

Now we turn to Theorem \ref{eif_asymp_theorem}, which gives conditions guaranteeing asymptotic normality of our EIF-based estimator $\hat{\tau}_{EIF}$. As with our proof of consistency, we write $\hat{\tau}_{EIF}$ as 

\begin{equation*}
\hat{\tau}_{\text{EIF}} = \frac{\mathbb{P}_n[\hat{\phi}^*_{1, 10}-\hat{\phi}^*_{0, 10}]}{\mathbb{P}_n\left[\hat{\lambda}_{10}\right]}
\end{equation*} 
and consider the asymptotic behavior of each component of $\hat{\tau}_{EIF}$ in turn. We then combine these results to prove the overall characterization given in Theorem \ref{eif_asymp_theorem}. First, we prove the following general lemma, which draws heavily from the logic applied in the proof of Theorem 5 in \citet{zeng_efficient_2023}. This lemma is what will allow us to connect the separate asymptotic behavior of the components of $\hat{\tau}_{EIF}$ to the larger desired result.

\begin{lemma}
\label{asymptotic_linearity_fraction}
Suppose we can write an estimand $\tau$ as
\begin{equation*}
\tau = \frac{E[N]}{E[D]},
\end{equation*}
where $N$ and $D$ are random variables. Suppose further that we've derived estimators $\mathbb{P}_n(\hat{N})$ and $\mathbb{P}_n(\hat{D})$ for $E[N]$ and $E[D]$, that are asymptotically linear in $\varphi_N$ and $\varphi_D$ and consistent for $E[N]$ and $E[D]$, respectively. In particular, 
$\mathbb{P}_n(\hat{N})=\mathbb{P}_n(\varphi_N)+o_p(1/\sqrt{n})$ and 
$\mathbb{P}_n(\hat{D})=\mathbb{P}_n(\varphi_D)+o_p(1/\sqrt{n})$ where $E[\varphi_N]=E[N]$ and $E[\varphi_D]=E[D]$.
Then 
\begin{equation*}
\hat{\tau}=\frac{\mathbb{P}_n(\hat{N})}{\mathbb{P}_n(\hat{D})}
\end{equation*}
has the following asymptotic behavior:
\begin{equation*}
\hat{\tau}-\tau=\mathbb{P}_n\left(\frac{\varphi_N-\tau\varphi_D}{E[D]}\right)+o_p(1/\sqrt{n}).
\end{equation*}
\end{lemma}

\begin{proof}[Proof of Lemma \ref{asymptotic_linearity_fraction}]
We can write 
\begin{align*}
\hat{\tau}-\tau &= \frac{\mathbb{P}_n(\hat{N})}{\mathbb{P}_n(\hat{D})}-\tau \\
&= \frac{\mathbb{P}_n(\hat{N})-\tau{\mathbb{P}_n(\hat{D})}}{\mathbb{P}_n(\hat{D})}\\
&= \frac{\mathbb{P}_n(\hat{N})-\tau{\mathbb{P}_n(\hat{D})}}{E[D]}+\left(\frac{1}{\mathbb{P}_n(\hat{D})}-\frac{1}{E(D)}\right)\left(\mathbb{P}_n(\hat{N})-\tau{\mathbb{P}_n(\hat{D})}\right).
\end{align*}
The first term can be expressed as
\begin{align*}
\frac{\mathbb{P}_n(\hat{N})-\tau{\mathbb{P}_n(\hat{D})}}{E[D]} &= \frac{\mathbb{P}_n(\varphi_N)+o_p(1/\sqrt{n})-\tau\left[\mathbb{P}_n(\varphi_D)+o_p(1/\sqrt{n})\right]}{E[D]}\\
&= \frac{\mathbb{P}_n(\varphi_N)-\tau \mathbb{P}_n(\varphi_D)}{E[D]}+o_p(1/\sqrt{n}) \hspace{1cm} \text{Slutsky's Theorem}\\
&= \frac{\mathbb{P}_n\left(\varphi_N-\tau\varphi_D\right)}{E[D]}+o_p(1/\sqrt{n})\\
&= \mathbb{P}_n\left(\frac{\varphi_N-\tau\varphi_D}{E[D]}\right)+o_p(1/\sqrt{n}).
\end{align*}
By assumption $\sqrt{n}(\mathbb{P}_n(\hat{D})-E(D))$ is asymptotically normal and $\mathbb{P}_n(\hat{D})$ converges in probability to $E(D)$. Thus, again applying Slutsky's Theorem, we obtain:
\begin{align*}
\left(\frac{1}{\mathbb{P}_n(\hat{D})}-\frac{1}{E(D)}\right) &=-\left(\frac{\mathbb{P}_n(\hat{D})-E(D)}{\mathbb{P}_n(\hat{D})E(D)}\right)\\
&= O_p(1/\sqrt{n}).
\end{align*}
Finally, we further apply the fact that $\sqrt{n}(\mathbb{P}_n(\hat{N})-E(N))$ is asymptotically normal to obtain:
\begin{align*}
\mathbb{P}_n(\hat{N})-\tau{\mathbb{P}_n(\hat{D})} &= \mathbb{P}_n(\hat{N})-\tau E[D]-\tau(\mathbb{P}_n(\hat{D})-E[D])\\
&= \mathbb{P}_n(\hat{N})-E[N]-\tau(\mathbb{P}_n(\hat{D})-E[D])\\
&= O_p(1/\sqrt{n})-\tau O_p(1/\sqrt{n})\\
&= O_p(1/\sqrt{n}).
\end{align*}
Thus, applying results from, e.g., Section 5.5.3 of \citet{boos_stefanski} we have
\begin{align*}
\left(\frac{1}{\mathbb{P}_n(\hat{D})}-\frac{1}{E(D)}\right)\left(\mathbb{P}_n(\hat{N})-\tau{\mathbb{P}_n(\hat{D})}\right)&=O_p(1/\sqrt{n})O_p(1/\sqrt{n})\\
&= O_p(1/n)\\
&= o_p(1/\sqrt{n}).
\end{align*}
This proves our result.
\end{proof}
\noindent We apply Lemma \ref{asymptotic_linearity_fraction} after characterizing the asymptotic linearity of each component of $\hat{\tau}_{EIF}$. First turning to $\mathbb{P}_n[\hat{\phi}^*_{1, 10}]$, we have the following lemma.

\begin{lemma}
\label{lemma: asymp_normal_1}
Under the consistency conditions given in Theorem \ref{eif_asymp_theorem}, then we have that $\mathbb{P}_n[\hat{\phi}^*_{1, 10}]$ is asymptotically linear in $\phi^0_{1, 10}$ and is $\sqrt{n}$-consistent for $E\left\{\left[p_1(X)-p_0(X)\right](1-\rho(X))\mu_{11}(X)\right\}$ .
\end{lemma}

\begin{proof}[Proof of Lemma \ref{lemma: asymp_normal_1}]

For simplicity in what follows, we suppress dependence on $X$ when there is no ambiguity. Note first that we can show that $E[\phi^0_{0, 10}]=E[e_{10}(1-\rho)\mu_{11}]$. Looking back to the definition of $\phi^0_{0, 10}$ given in the main manuscript, we see that its first term as expectation $E[e_{10}(1-\rho)\mu_{11}]$ and its second two terms both have expectation zero. We apply this fact to decompose the error of $\mathbb{P}_n\left[\hat{\phi}^*_{1, 10}\right]$ as follows:
\begin{align}
\mathbb{P}_n\left[\hat{\phi}^*_{1, 10}\right]-E[e_{10}(1-\rho)\mu_{11}]&=\mathbb{P}_n\left[\hat{\phi}^*_{1, 10}\right]-E[\phi^0_{0, 10}]\nonumber\\
&= (\mathbb{P}_n-E)(\hat{\phi}^*_{1, 10}-\phi^0_{0, 10})+(\mathbb{P}_n-E)(\phi^0_{0, 10})+E(\hat{\phi}^*_{1, 10}-\phi^0_{0, 10}). \label{eq: phi_1_decomp}
\end{align}
%where in (\ref{eq: phi_1_decomp}) we've added $(E[\hat{\phi}^*_{1, 10}]-E[\hat{\phi}^*_{1, 10}])$, $(\mathbb{P}_n[\phi^0_{0, 10}]-\mathbb{P}_n[\phi^0_{0, 10}])$, and $(E[\phi^0_{0, 10}]-E[\phi^0_{0, 10}])$. 

Also note that the consistency conditions of Theorem \ref{consistency_theorem} imply that $\|\hat{\phi}^*_{1, 10}-\phi^0_{0, 10}\|=o_p(1)$.

Considering each component of (\ref{eq: phi_1_decomp}) in turn, we have first by our assumed sample splitting approach and subsequent application of Lemma 2 in \citet{kennedy_sharp_2020} that
\begin{align*}
(\mathbb{P}_n-E)(\hat{\phi}^*_{1, 10}-e_{10}(1-\rho)\mu_{11}) &= (\mathbb{P}_n-E)(\hat{\phi}^*_{1, 10}-\phi^0_{1, 10})\\
&= O_p\left(\frac{||\hat{\phi}^*_{1, 10}-\phi^0_{1, 10}||}{\sqrt{n}}\right)\\
&= O_p\left(\frac{o_p(1)}{\sqrt{n}}\right) & \text{by assumptions of Theorem \ref{eif_asymp_theorem}}\\
&= o_p(1/\sqrt{n}).
\end{align*}

Next, since $E(\hat{\phi}^*_{1, 10}-\phi^0_{0, 10})=E(\hat{\phi}^*_{1, 10}-e_{10}(1-\rho)\mu_{11})$, we can apply our bias calculations from Equation \ref{eq: asymp_bias_num_1} in the proof of Lemma \ref{lemma: consistency_proof_1}---in addition to the assumptions of Theorem \ref{eif_asymp_theorem}---to fully characterize the asymptotic behavior of $\mathbb{P}_n[\hat{\phi}^*_{1, 10}]$. Recalling Equation \ref{eq: asymp_bias_num_1}, we know that $E(\hat{\phi}^*_{1, 10}-e_{10}(1-\rho)\mu_{11})$ can be written as 
\begin{align}
&E\left[\frac{\left[\pi\rho(1-\hat{\rho})-\hat{\pi}\hat{\rho}(1-\rho)\right]\left[\mu_{11}p_1-\hat{\mu}_{11}\hat{p}_1\right]}{\hat{\pi}\hat{\rho}}\right] \nonumber\\
& \hspace{0.65cm} - E\left[\hat{p}_0\left\{\frac{\pi\rho p_1}{\hat{\pi}\hat{\rho}\hat{p}_1}-1\right\}[\mu_{11}-\hat{\mu}_{11}]+\hat{\mu}_{11}\left\{\frac{(1-\pi)\rho}{(1-\hat{\pi})\hat{\rho}}-1\right\}[p_0-\hat{p}_0] + (\hat{p}_0-p_0)(\mu_{11}-\hat{\mu}_{11})\right] \nonumber\\
& \hspace{0.65cm} + E\left[\hat{p}_0\rho\left\{\frac{\pi p_1}{\hat{\pi}\hat{p}_1}-1\right\}[\mu_{11}-\hat{\mu}_{11}]+\hat{\mu}_{11}\rho\left\{\frac{(1-\pi)}{(1-\hat{\pi})}-1\right\}(p_0-\hat{p}_0)+\rho(\hat{\mu}_{11}-\mu_{11})(p_0-\hat{p}_0)\right]. \label{eq: expectation_1_decomp}
\end{align}
Each of these terms can be bounded above by repeated application of the Cauchy-Schwarz inequality; in particular, we have that 
\begin{align*}
&E\left[\frac{\left[\pi\rho(1-\hat{\rho})-\hat{\pi}\hat{\rho}(1-\rho)\right]\left[\mu_{11}p_1-\hat{\mu}_{11}\hat{p}_1\right]}{\hat{\pi}\hat{\rho}}\right]\\
&= E\left[\frac{\left[\pi\rho(1-\hat{\rho})-\hat{\pi}\hat{\rho}(1-\rho)\right]\left[\mu_{11}p_1-\hat{\mu}_{11}\hat{p}_1\right]}{\hat{\pi}\hat{\rho}}\right]\\
&= E\left[\frac{\left[(\pi-\hat{\pi})\rho(1-\hat{\rho})+(\rho-\hat{\rho})(\hat{\pi})(1-\hat{\rho})+[((1-\hat{\rho})-(1-\rho)]\hat{\pi}\hat{\rho}\right]\left[(\mu_{11}-\hat{\mu}_{11})p_1+(p_1-\hat{p}_1)\hat{\mu}_{11}\right]}{\hat{\pi}\hat{\rho}}\right]\\
&\leq E\left[\left(\left|(\pi-\hat{\pi})\right|\left|\frac{\rho(1-\hat{\rho})}{\hat{\pi}\hat{\rho}}\right|+\left|(\rho-\hat{\rho})\right|\left|\frac{\hat{\pi}(1-\hat{\rho})}{\hat{\pi}\hat{\rho}}\right|+\left|\rho-\hat{\rho}\right|\left|\frac{\hat{\pi}\hat{\rho}}{\hat{\pi}\hat{\rho}}\right|\right) \right. \\
 & \left. \quad\quad\quad\times \left(\left|\mu_{11}-\hat{\mu}_{11}\right|\left|\frac{p_1}{\hat{\rho}\hat{\pi}}\right|+\left|(p_1-\hat{p}_1)\right|\left|\frac{\hat{\mu}_{11}}{\hat{\pi}\hat{\rho}}\right|\right)\right]\\
&\leq C_1\times \left(||\pi-\hat{\pi}||+||\rho-\hat{\rho}||\right)\left(||\mu_{11}-\hat{\mu}_{11}||+||p_1-\hat{p}_1||\right),
\end{align*}
where the last inequality holds with an application of the Cauchy-Schwarz inequality and the constant condition referenced above. Referring to that condition, we note that $C_1\in \mathbb{R}^+$ is some function of the $C$ and $\delta$ referenced in Theorem \ref{eif_asymp_theorem}.

In a similar way, we can bound the second term in (\ref{eq: expectation_1_decomp}) above by 
\begin{align*}
C_2\times &\left[\left(||\pi-\hat{\pi}||+||\rho-\hat{\rho}||+||p_1-\hat{p}_1||\right)\left(||\mu_{11}-\hat{\mu}_{11}||\right) \right. \\ 
&\left.+\left(||\pi-\hat{\pi}||+||\rho-\hat{\rho}||\right)||p_0-\hat{p}_0||+||\mu_{11}-\hat{\mu}_{11}||\hspace{.5mm}||p_0-\hat{p}_0||\right].
\end{align*}
Finally, the third term in (\ref{eq: expectation_1_decomp}) can be bounded above by
\begin{equation*}
C_3\times \left[\left(||\pi-\hat{\pi}||+||p_1-\hat{p}_1||\right)||\mu_{11}-\hat{\mu}_{11}||+||\pi-\hat{\pi}||\hspace{.5mm}||p_0-\hat{p}_0||+||\mu_{11}-\hat{\mu}_{11}||\hspace{.5mm}||p_0-\hat{p}_0||\right].
\end{equation*}
We can combine like terms in the above bounds to conclude overall that the conditional bias in $\hat{\phi}^*_{1, 10}$ is bounded above by the following expression times $C_4$, where $C_4>0$ is chosen such that $C_4\ge C_1+C_2+C_3$:
\begin{align*}
&||\pi-\hat{\pi}||\hspace{.5mm}||\mu_{11}-\hat{\mu}_{11}||+||\pi-\hat{\pi}||\hspace{.5mm}||p_1-\hat{p}_1||+||\rho-\hat{\rho}||\hspace{.5mm}||\mu_{11}-\hat{\mu}_{11}||+||\rho-\hat{\rho}||\hspace{.5mm}||p_1-\hat{p}_1||\\
&\hspace{.5cm}+||p_1-\hat{p}_1||\hspace{.5mm}||\mu_{11}-\hat{\mu}_{11}||+||\pi-\hat{\pi}||\hspace{.5mm}||p_0-\hat{p}_0||+||\rho-\hat{\rho}||\hspace{.5mm}||p_0-\hat{p}_0|| + ||\mu_{11}-\hat{\mu}_{11}||\hspace{.5mm}||p_0-\hat{p}_0||.\\
&= \underbrace{||\mu_{11}-\hat{\mu}_{11}||\left\{||p_1-\hat{p}_1||+||p_0-\hat{p}_0||\right\}}_{\equiv R^1_n}+\underbrace{||\mu_{11}-\hat{\mu}_{11}||\left\{||\pi-\hat{\pi}||+||\rho-\hat{\rho}||\right\}}_{\equiv R^2_n}\\
&\hspace{.5cm}+\underbrace{\left\{||\pi-\hat{\pi}||+||\rho-\hat{\rho}||\right\}\left\{||p_0-\hat{p}_0||+||p_1-\hat{p}_1||\right\}}_{\equiv R^3_n},
\end{align*}
where we overload the notation used in Theorem \ref{eif_asymp_theorem} to define shortand for the remainder terms specific to $\hat{\phi}^*_{1, 10}$. We can use the above finding to characterize both the rate-robustness properties of $\hat{\phi}^*_{1, {10}}$ and its consistency properties. That is, we now have that
\begin{equation}
\mathbb{P}_n\left[\hat{\phi}^*_{1, 10}\right]-E[e_{10}(1-\rho)\mu_{11}]=(\mathbb{P}_n-E)\left(\phi^0_{1, {10}}\right)+O_p\left(R^1_n+R^2_n+R^3_n\right)+o_p(1/\sqrt{n}). \label{eq: num_1_asymp_linear}
\end{equation}
Then, under the rate robustness conditions given in Theorem \ref{eif_asymp_theorem}, we have that $R^1_n+R^2_n+R^3_n=o_p(1/\sqrt{n})=o_p(1/\sqrt{n})$, which implies that $\hat{\phi}^*_{1, 10}$ is asymptotically normal. 
\end{proof}

\begin{proof}[Proof of Theorem \ref{eif_asymp_theorem}]

We can apply analogous arguments as those given in Lemma \ref{lemma: asymp_normal_1} to prove the asymptotic normality of $\mathbb{P}_n[\hat{\phi}^*_{1, 10}]$ and $\mathbb{P}_n\left(\hat{\lambda}_{10}\right)$ under the conditions of Theorem \ref{eif_asymp_theorem}. Both such arguments combine decompositions similar to Equation (\ref{eq: phi_1_decomp}) with the bias formulas derived in Lemmas \ref{lemma: consistency_proof_2} and \ref{lemma: consistency_proof_3}. These arguments ultimately show that 
\begin{align}
&\mathbb{P}_n\left(\hat{\lambda}_{10}\right)-E[e_{10}(1-\rho)] \nonumber\\
&=(\mathbb{P}_n-E)\left(\lambda_{10}\right)+o_p(1/\sqrt{n}) \label{eq: denom_asymp_linear}
\end{align}
and that
\begin{equation}
\mathbb{P}_n\left[\hat{\phi}^*_{0, 10}\right]-E[e_{10}(1-\rho)\mu_{00}]=(\mathbb{P}_n-E)\left(\phi^0_{0, {10}}\right)+o_p(1/\sqrt{n}) \label{eq: num_0_asymp_linear}
\end{equation}
Since $E[e_{10}(1-\rho)\mu_{11}]=E(\phi^0_{1, 10})$, $E[e_{10}(1-\rho)]=E[\lambda_{10}]$, and $E[e_{10}(1-\rho)\mu_{00}]=E(\phi^0_{0, 10})$, we can add the expectations on the left hand side of (\ref{eq: num_1_asymp_linear}), (\ref{eq: num_0_asymp_linear}), and (\ref{eq: denom_asymp_linear}) to the right hand side of their respective equations to obtain:
\begin{align*}
\mathbb{P}_n\left[\hat{\phi}^*_{1, 10}\right]&=\mathbb{P}_n\left(\phi^0_{1, {10}}\right)+o_p(1/\sqrt{n})\\
\mathbb{P}_n\left[\hat{\phi}^*_{0, 10}\right]&=\mathbb{P}_n\left(\phi^0_{0, {10}}\right)+o_p(1/\sqrt{n})\\
\mathbb{P}_n\left(\hat{\lambda}_{10}\right)&= \mathbb{P}_n\left(\lambda_{10}\right)+o_p(1/\sqrt{n})
\end{align*}
We are now in a position to directly apply Lemma \ref{asymptotic_linearity_fraction}. That is, we have
\begin{equation*}
\hat{\tau}_{EIF}-\tau^0_{10}=\mathbb{P}_n\left(\frac{\phi^0_{1, 10}-\phi^0_{0, 10}-\psi^0_{10}\lambda_{10}}{E\left[e_{10}(X)\left\{1-\rho(X)\right\}\right]}\right)+o_p(1/\sqrt{n})
\end{equation*}
This is equal to Equation (\ref{eq: eif_asymp_expansion}) and illustrates that $\hat{\tau}_{EIF}$ is asymptotically linear with influence function $\varphi_{\psi^0_{10}}$, i.e., it achieves the non-parametric efficiency bound.

\end{proof}

\newpage

\section{Identification for Hotspotting Analysis}
\label{one_sided_identification}

\begin{lemma}\label{one_sided_lemma}
Under Assumptions \ref{consistency}, \ref{treatment_ignorability}, \ref{monotonicity}*, \ref{princignorability}*, \ref{stratexch}*, and \ref{mean_exch}, we can identify $\tau^0_1 = E[Y(1)-Y(0)|R=0, C(1)=1]$ as follows:
\begin{equation}
\tau^0_1 = \frac{E[p_1(X)(\mu_{11}(X)-\mu_{00}(X))(1-\rho(X))]}{E[p_1(X)(1-\rho(X)]}
\end{equation}
\end{lemma}

\begin{proof}[Proof of Lemma \ref{one_sided_lemma}]
\begin{align}
&E[Y(1)-Y(0)|R=0, C(1)=1] \nonumber\\
&= E[E\{Y(1)-Y(0)|X, R=0, C(1)=1, C(0)=0\}|C(1)=1, R=0] \nonumber\\
&= E[E\{Y(1)-Y(0)|X, R=1, C(1)=1, C(0)=0\}|C(1)=1, R=0] & \text{By Assumption \ref{mean_exch}}\nonumber\\
&= E[E\{Y(1)|X, R=1, A=1, C=1\}|C(1)=1, R=0]\nonumber\\
&\hspace{.5cm} -E[E\{Y(0)|X, R=1, A=0, C=0\}|C(1)=1, R=0] & \text{By Assumption \ref{princignorability}*}\nonumber\\
&= E[E\{Y|X, R=1, A=1, C=1\}|C(1)=1, R=0]\nonumber\\
&\hspace{.5cm} -E[E\{Y|X, R=1, A=0, C=0\}|C(1)=1, R=0] & \text{By Assumption \ref{consistency}} \label{eq: strong_mono_id_int}
\end{align}
Then, for any function $f(X)$, we have:
\begin{align*}
E[f(X)|C(1)=1, R=0] &= \frac{E[I(C(1)=1, R=0)f(X)]}{P(C(1)=1, R=0)}\\
&= \frac{E[f(X)E\{I(C(1)=1)|X\}E\{I(R=0)|X\}}{E[E\{I(C(1)=1)|X\}E\{I(R=0)|X\}]} & \text{By Assumption \ref{stratexch}*}\\
&= \frac{E[f(X)p_1(X)(1-\rho(X))}{E[p_1(X)(1-\rho(X))]} & \text{By Assumptions \ref{monotonicity}* and \ref{treatment_ignorability}}
\end{align*}
Applying this result to (\ref{eq: strong_mono_id_int}), we obtain:
\begin{equation}
E[Y(1)-Y(0)|R=0, C(1)=1] = \frac{E[p_1(X)(\mu_{11}(X)-\mu_{00}(X))(1-\rho(X))]}{E[p_1(X)(1-\rho(X)]} \label{eq: strong_mono_id}
\end{equation}
\end{proof}

\newpage

\section{Descriptive Statistics for Hotspotting Analysis}
\label{descriptive_stats_hotspotting}
\begin{table}[h!]
\hspace{-0.5cm}
\resizebox{\textwidth}{!}{%
\begin{tabular}{p{10cm}rr}\toprule
\textbf{Characteristic} & \textbf{Trial Sample} & \textbf{Target Population Sample} \\
& (N=613) & (N=463) \\\toprule 
  \# Hospitalizations 6 months prior to index admission & 1.80 (1.67) & 1.86 (1.44) \\ 
 Length of Index Admission (Days) & 7.01 (5.99) & 7.54 (9.66) \\ 
 Age Category &  &  \\ 
 \quad\quad$<$= 34 yrs & 55 (9.0\%) & 43 (9.3\%) \\ 
 \quad\quad$>$= 35 and $<$ 40 yrs & 33 (5.4\%) & 18 (3.9\%) \\ 
 \quad\quad$>$= 40 and $<$ 45 yrs & 34 (5.5\%) & 19 (4.1\%) \\ 
 \quad\quad$>$= 45 and $<$ 50 yrs & 64 (10.4\%) & 45 (9.7\%) \\ 
 \quad\quad$>$= 50 and $<$ 55 yrs & 75 (12.2\%) & 69 (14.9\%) \\ 
 \quad\quad$>$= 55 and $<$ 60 yrs & 90 (14.7\%) & 101 (21.8\%) \\ 
    \quad\quad$>$= 60 and $<$ 65 yrs & 86 (14.0\%) & 65 (14.0\%) \\ 
    \quad\quad$>$= 65 and $<$ 69 yrs & 82 (13.4\%) & 52 (11.2\%) \\ 
    \quad\quad$>$= 70 and $<$ 75 & 51 (8.3\%) & 30 (6.5\%) \\ 
    \quad\quad$>$= 75 & 43 (7.0\%) & 21 (4.5\%) \\ 
    Male & 315 (51.4\%) & 215 (46.4\%) \\ 
    Race &  &  \\ 
    \quad\quad Non-White & 496 (80.9\%) & 178 (38.4\%) \\ 
    \quad\quad White & 117 (19.1\%) & 285 (61.6\%) \\ 
    Marital Status &  &  \\ 
    \quad\quad Married, Civil Union, or Cohabitating & 170 (27.7\%) & 115 (24.8\%) \\ 
    \quad\quad Single, divorced, or widowed & 443 (72.3\%) & 348 (75.2\%) \\ 
    Diagnoses at or prior to index admission & &\\
    \quad\quad Heart Failure & 272 (44.4\%) & 180 (38.9\%) \\ 
    \quad\quad HIV/AIDS & 15 (2.4\%) & 12 (2.6\%) \\ 
    \quad\quad COPD/Emphysema & 172 (28.1\%) & 154 (33.3\%) \\ 
    \quad\quad Diabetes & 306 (49.9\%) & 246 (53.1\%) \\ 
    \quad\quad Arthritis & 81 (13.2\%) & 70 (15.1\%) \\ 
    \quad\quad Mental Health Condition & 501 (81.7\%) & 416 (89.8\%) \\ 
    \quad\quad Substance Abuse & 322 (52.5\%) & 271 (58.5\%) \\ 
   \bottomrule
\end{tabular}%
}
\caption{Descriptive statistics of hotspotting trial sample and target population sample}
\end{table}
\newpage
\begin{figure}[t!]
\centering
\includegraphics[scale=0.66]{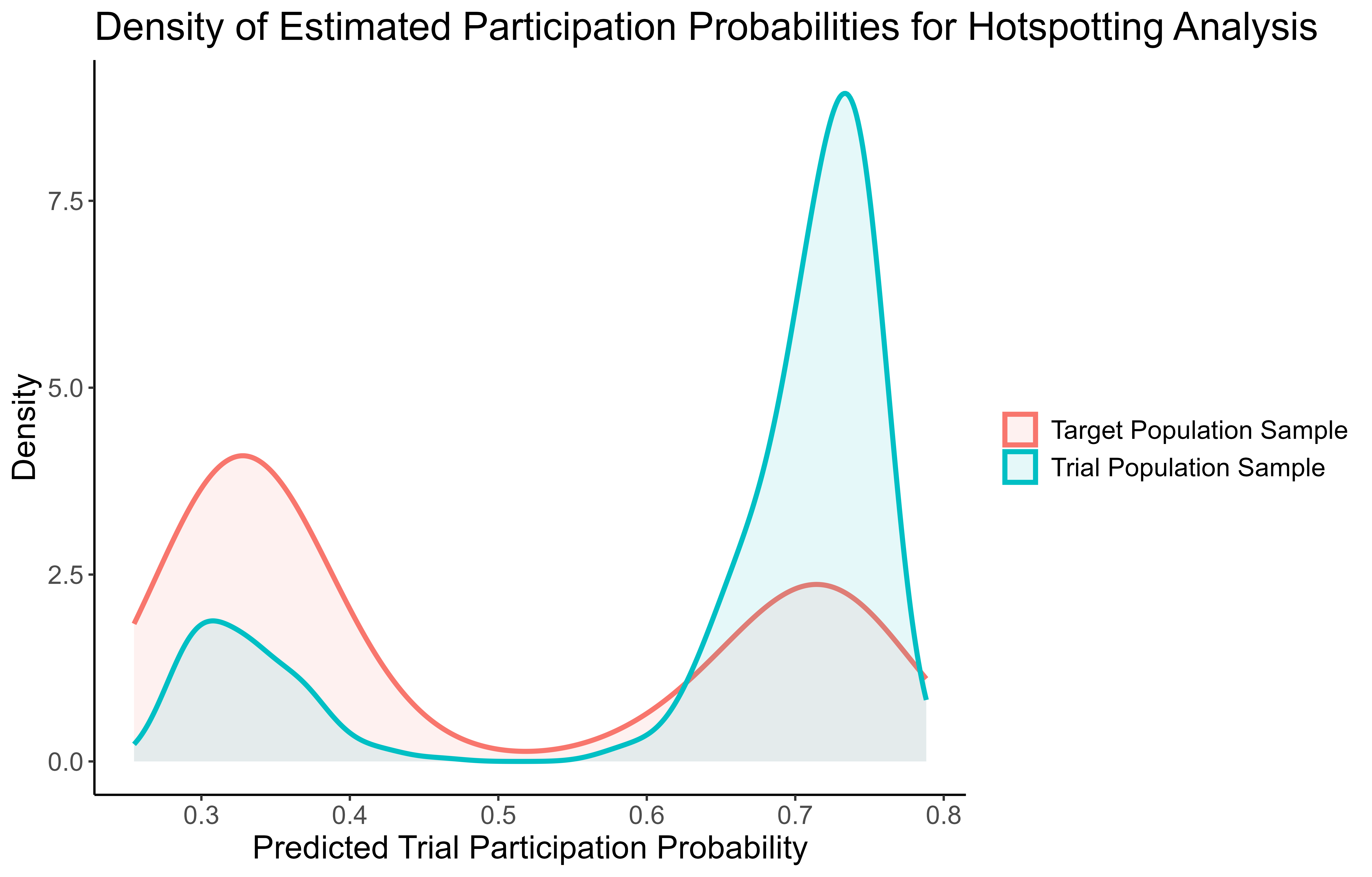}
\label{fig: hotspotting_overlap}
\caption{Kernel density estimates of estimated participation probability densities for the hotspotting data analysis.}
\end{figure}
\section{Additional Simulation Results}
\label{sub: sim_res}
\begin{table}[h!]
\centering
\begin{tabular}{rlllrrr}
  \hline
Ratio & Trial Size & Target Size & Correctly-Specified Models & EIF & IPW & OM \\ 
  \hline
1/2 & 500 & 500 & All &  0.888 & -0.135 & -0.154 \\ 
  1/2 & 1000 & 1000 & All &  0.099 &  0.070 &  0.071 \\ 
  1/2 & 2000 & 2000 & All &  0.025 &  0.030 &  0.025 \\ 
  1/2 & 50000 & 50000 & All &  0.000 & -0.001 & -0.001 \\ 
   \hline
1/2 & 500 & 500 & None &  1.149 &  1.124 &  0.557 \\ 
  1/2 & 1000 & 1000 & None &  1.162 &  1.133 &  0.559 \\ 
  1/2 & 2000 & 2000 & None &  1.164 &  1.136 &  0.557 \\ 
  1/2 & 50000 & 50000 & None &  1.160 &  1.138 &  0.558 \\ 
   \hline
1/2 & 500 & 500 & OM and PS &  0.160 &  1.130 &  0.070 \\ 
  1/2 & 1000 & 1000 & OM and PS &  0.050 &  1.091 &  0.050 \\ 
  1/2 & 2000 & 2000 & OM and PS &  0.011 &  1.043 &  0.013 \\ 
  1/2 & 50000 & 50000 & OM and PS &  0.002 &  1.039 &  0.001 \\ 
   \hline
1/2 & 500 & 500 & TP, OM, and PP &  0.717 & -0.020 & -0.330 \\ 
  1/2 & 1000 & 1000 & TP, OM, and PP & -0.077 & -0.025 & -0.330 \\ 
  1/2 & 2000 & 2000 & TP, OM, and PP &  0.093 & -0.017 & -0.326 \\ 
  1/2 & 50000 & 50000 & TP, OM, and PP &  0.004 & -0.018 & -0.326 \\ 
   \hline
1/2 & 500 & 500 & TP, PS, and PP & -0.385 & -0.567 &  0.602 \\ 
  1/2 & 1000 & 1000 & TP, PS, and PP &  0.102 &  0.058 &  0.560 \\ 
  1/2 & 2000 & 2000 & TP, PS, and PP &  0.054 &  0.016 &  0.562 \\ 
  1/2 & 50000 & 50000 & TP, PS, and PP &  0.003 & -0.001 &  0.557 \\ 
   \hline
\end{tabular}
\caption{Bias with (Trial Size):(Total Sample Size) Ratio of 1/2} 
\end{table}
% latex table generated in R 4.3.2 by xtable 1.8-4 package
% Wed Mar 27 11:32:58 2024
\begin{table}[ht]
\centering
\begin{tabular}{rlllrrr}
  \hline
Ratio & Trial Size & Target Size & Correctly-Specified Models & EIF & IPW & OM \\ 
  \hline
1/21 & 500 & 10000 & All &  5.434 &  0.153 &  0.127 \\ 
  1/21 & 1000 & 20000 & All & -0.054 &  0.152 &  0.096 \\ 
  1/21 & 2000 & 40000 & All &  0.078 &  0.042 &  0.021 \\ 
  1/21 & 50000 & 1000000 & All &  0.004 &  0.001 &  0.001 \\ 
   \hline
1/21 & 500 & 10000 & None &  1.533 &  1.434 &  0.647 \\ 
  1/21 & 1000 & 20000 & None &  1.522 &  1.437 &  0.650 \\ 
  1/21 & 2000 & 40000 & None &  1.506 &  1.433 &  0.656 \\ 
  1/21 & 50000 & 1000000 & None &  1.494 &  1.436 &  0.654 \\ 
   \hline
1/21 & 500 & 10000 & OM and PS &  0.347 &  2.571 &  0.085 \\ 
  1/21 & 1000 & 20000 & OM and PS &  0.082 &  1.058 &  0.088 \\ 
  1/21 & 2000 & 40000 & OM and PS &  0.017 &  1.160 &  0.016 \\ 
  1/21 & 50000 & 1000000 & OM and PS &  0.001 &  1.111 &  0.001 \\ 
   \hline
1/21 & 500 & 10000 & TP, OM, and PP &  0.472 &  0.243 & -0.240 \\ 
  1/21 & 1000 & 20000 & TP, OM, and PP & -0.154 &  0.243 & -0.242 \\ 
  1/21 & 2000 & 40000 & TP, OM, and PP &  0.186 &  0.241 & -0.235 \\ 
  1/21 & 50000 & 1000000 & TP, OM, and PP &  0.005 &  0.227 & -0.237 \\ 
   \hline
1/21 & 500 & 10000 & TP, PS, and PP & -0.019 &  0.142 &  0.665 \\ 
  1/21 & 1000 & 20000 & TP, PS, and PP &  0.148 &  0.130 &  0.653 \\ 
  1/21 & 2000 & 40000 & TP, PS, and PP &  0.047 &  0.048 &  0.653 \\ 
  1/21 & 50000 & 1000000 & TP, PS, and PP &  0.003 &  0.004 &  0.654 \\ 
   \hline
\end{tabular}
\caption{Bias with (Trial Size):(Total Sample Size) Ratio of 1/21} 
\end{table}
% latex table generated in R 4.3.2 by xtable 1.8-4 package
% Wed Mar 27 11:32:58 2024
\begin{table}[ht]
\centering
\begin{tabular}{rlllrrr}
  \hline
Ratio & Trial Size & Target Size & Correctly-Specified Models & EIF & IPW & OM \\ 
  \hline
1/3 & 500 & 1000 & All &  0.291 &  0.190 &  0.193 \\ 
  1/3 & 1000 & 2000 & All &  0.315 &  0.077 &  0.068 \\ 
  1/3 & 2000 & 4000 & All &  0.036 &  0.016 &  0.009 \\ 
  1/3 & 50000 & 100000 & All &  0.003 &  0.002 &  0.002 \\ 
   \hline
1/3 & 500 & 1000 & None &  1.234 &  1.180 &  0.559 \\ 
  1/3 & 1000 & 2000 & None &  1.220 &  1.172 &  0.552 \\ 
  1/3 & 2000 & 4000 & None &  1.209 &  1.172 &  0.552 \\ 
  1/3 & 50000 & 100000 & None &  1.208 &  1.177 &  0.553 \\ 
   \hline
1/3 & 500 & 1000 & OM and PS &  0.084 &  1.261 &  0.118 \\ 
  1/3 & 1000 & 2000 & OM and PS &  0.055 &  1.128 &  0.054 \\ 
  1/3 & 2000 & 4000 & OM and PS &  0.023 &  1.105 &  0.024 \\ 
  1/3 & 50000 & 100000 & OM and PS &  0.000 &  1.069 &  0.001 \\ 
   \hline
1/3 & 500 & 1000 & TP, OM, and PP &  0.015 &  0.057 & -0.294 \\ 
  1/3 & 1000 & 2000 & TP, OM, and PP &  0.213 &  0.051 & -0.292 \\ 
  1/3 & 2000 & 4000 & TP, OM, and PP &  0.055 &  0.051 & -0.294 \\ 
  1/3 & 50000 & 100000 & TP, OM, and PP &  0.006 &  0.046 & -0.292 \\ 
   \hline
1/3 & 500 & 1000 & TP, PS, and PP & -0.166 & -0.089 &  0.567 \\ 
  1/3 & 1000 & 2000 & TP, PS, and PP &  0.064 &  0.079 &  0.556 \\ 
  1/3 & 2000 & 4000 & TP, PS, and PP &  0.054 &  0.029 &  0.553 \\ 
  1/3 & 50000 & 100000 & TP, PS, and PP &  0.004 &  0.001 &  0.553 \\ 
   \hline
\end{tabular}
\caption{Bias with (Trial Size):(Total Sample Size) Ratio of 1/3} 
\end{table}

\begin{table}[ht]
\centering
\begin{tabular}{rlllrrr}
  \hline
Ratio & Trial Size & Target Size & Correctly-Specified Models & EIF & IPW & OM \\ 
  \hline
1/2 & 500 & 500 & All &   27.991 &    6.960 &    7.626 \\ 
  1/2 & 1000 & 1000 & All &    1.309 &    0.343 &    0.334 \\ 
  1/2 & 2000 & 2000 & All &    1.096 &    0.206 &    0.177 \\ 
  1/2 & 50000 & 50000 & All &    0.048 &    0.035 &    0.029 \\ 
   \hline
1/2 & 500 & 500 & None &    1.169 &    1.136 &    0.572 \\ 
  1/2 & 1000 & 1000 & None &    1.172 &    1.139 &    0.566 \\ 
  1/2 & 2000 & 2000 & None &    1.169 &    1.139 &    0.561 \\ 
  1/2 & 50000 & 50000 & None &    1.161 &    1.138 &    0.558 \\ 
   \hline
1/2 & 500 & 500 & OM and PS &    0.983 &    1.269 &    3.553 \\ 
  1/2 & 1000 & 1000 & OM and PS &    0.310 &    1.160 &    0.288 \\ 
  1/2 & 2000 & 2000 & OM and PS &    0.173 &    1.073 &    0.165 \\ 
  1/2 & 50000 & 50000 & OM and PS &    0.032 &    1.040 &    0.030 \\ 
   \hline
1/2 & 500 & 500 & TP, OM, and PP &   31.641 &    0.254 &    0.349 \\ 
  1/2 & 1000 & 1000 & TP, OM, and PP &    7.142 &    0.179 &    0.340 \\ 
  1/2 & 2000 & 2000 & TP, OM, and PP &    0.463 &    0.127 &    0.332 \\ 
  1/2 & 50000 & 50000 & TP, OM, and PP &    0.071 &    0.031 &    0.327 \\ 
   \hline
1/2 & 500 & 500 & TP, PS, and PP &   19.348 &   15.839 &    1.313 \\ 
  1/2 & 1000 & 1000 & TP, PS, and PP &    0.622 &    0.332 &    0.567 \\ 
  1/2 & 2000 & 2000 & TP, PS, and PP &    0.364 &    0.202 &    0.566 \\ 
  1/2 & 50000 & 50000 & TP, PS, and PP &    0.059 &    0.035 &    0.557 \\ 
   \hline
\end{tabular}
\caption{RMSE with (Trial Size):(Total Sample Size) Ratio of 1/2} 
\end{table}
% latex table generated in R 4.3.2 by xtable 1.8-4 package
% Wed Mar 27 11:33:28 2024
\begin{table}[ht]
\centering
\begin{tabular}{rlllrrr}
  \hline
Ratio & Trial Size & Target Size & Correctly-Specified Models & EIF & IPW & OM \\ 
  \hline
1/21 & 500 & 10000 & All &  161.594 &    2.152 &    1.792 \\ 
  1/21 & 1000 & 20000 & All &    6.145 &    2.117 &    1.320 \\ 
  1/21 & 2000 & 40000 & All &    0.498 &    0.285 &    0.195 \\ 
  1/21 & 50000 & 1000000 & All &    0.067 &    0.053 &    0.032 \\ 
   \hline
1/21 & 500 & 10000 & None &    1.561 &    1.451 &    0.663 \\ 
  1/21 & 1000 & 20000 & None &    1.537 &    1.445 &    0.657 \\ 
  1/21 & 2000 & 40000 & None &    1.513 &    1.436 &    0.659 \\ 
  1/21 & 50000 & 1000000 & None &    1.494 &    1.437 &    0.654 \\ 
   \hline
1/21 & 500 & 10000 & OM and PS &   12.632 &   31.122 &    2.802 \\ 
  1/21 & 1000 & 20000 & OM and PS &    0.734 &   10.742 &    0.753 \\ 
  1/21 & 2000 & 40000 & OM and PS &    0.193 &    1.269 &    0.177 \\ 
  1/21 & 50000 & 1000000 & OM and PS &    0.033 &    1.115 &    0.031 \\ 
   \hline
1/21 & 500 & 10000 & TP, OM, and PP &   32.801 &    0.441 &    0.269 \\ 
  1/21 & 1000 & 20000 & TP, OM, and PP &   18.107 &    0.367 &    0.256 \\ 
  1/21 & 2000 & 40000 & TP, OM, and PP &    2.848 &    0.307 &    0.242 \\ 
  1/21 & 50000 & 1000000 & TP, OM, and PP &    0.096 &    0.230 &    0.237 \\ 
   \hline
1/21 & 500 & 10000 & TP, PS, and PP &    6.918 &    3.244 &    0.679 \\ 
  1/21 & 1000 & 20000 & TP, PS, and PP &    2.222 &    0.516 &    0.660 \\ 
  1/21 & 2000 & 40000 & TP, PS, and PP &    0.549 &    0.300 &    0.657 \\ 
  1/21 & 50000 & 1000000 & TP, PS, and PP &    0.086 &    0.055 &    0.654 \\ 
   \hline
\end{tabular}
\caption{RMSE with (Trial Size):(Total Sample Size) Ratio of 1/21} 
\end{table}
% latex table generated in R 4.3.2 by xtable 1.8-4 package
% Wed Mar 27 11:33:28 2024
\begin{table}[ht]
\centering
\begin{tabular}{rlllrrr}
  \hline
Ratio & Trial Size & Target Size & Correctly-Specified Models & EIF & IPW & OM \\ 
  \hline
1/3 & 500 & 1000 & All &    2.467 &    1.490 &    1.846 \\ 
  1/3 & 1000 & 2000 & All &    6.645 &    0.525 &    0.437 \\ 
  1/3 & 2000 & 4000 & All &    0.310 &    0.215 &    0.163 \\ 
  1/3 & 50000 & 100000 & All &    0.053 &    0.039 &    0.029 \\ 
   \hline
1/3 & 500 & 1000 & None &    1.259 &    1.193 &    0.575 \\ 
  1/3 & 1000 & 2000 & None &    1.231 &    1.179 &    0.560 \\ 
  1/3 & 2000 & 4000 & None &    1.215 &    1.175 &    0.555 \\ 
  1/3 & 50000 & 100000 & None &    1.208 &    1.177 &    0.553 \\ 
   \hline
1/3 & 500 & 1000 & OM and PS &    1.763 &    1.591 &    1.022 \\ 
  1/3 & 1000 & 2000 & OM and PS &    0.312 &    1.228 &    0.298 \\ 
  1/3 & 2000 & 4000 & OM and PS &    0.171 &    1.149 &    0.164 \\ 
  1/3 & 50000 & 100000 & OM and PS &    0.031 &    1.070 &    0.029 \\ 
   \hline
1/3 & 500 & 1000 & TP, OM, and PP &    5.962 &    0.275 &    0.315 \\ 
  1/3 & 1000 & 2000 & TP, OM, and PP &    2.163 &    0.208 &    0.302 \\ 
  1/3 & 2000 & 4000 & TP, OM, and PP &    1.244 &    0.146 &    0.299 \\ 
  1/3 & 50000 & 100000 & TP, OM, and PP &    0.079 &    0.054 &    0.292 \\ 
   \hline
1/3 & 500 & 1000 & TP, PS, and PP &    5.602 &    8.099 &    0.623 \\ 
  1/3 & 1000 & 2000 & TP, PS, and PP &    1.152 &    0.440 &    0.563 \\ 
  1/3 & 2000 & 4000 & TP, PS, and PP &    0.385 &    0.213 &    0.556 \\ 
  1/3 & 50000 & 100000 & TP, PS, and PP &    0.063 &    0.037 &    0.553 \\ 
   \hline
\end{tabular}
\caption{RMSE with (Trial Size):(Total Sample Size) Ratio of 1/3} 
\end{table}

\begin{table}[ht]
\centering
\begin{tabular}{rlllr}
  \hline
Ratio & Trial Size & Target Size & Correctly-Specified Models & Estimated Coverage \\ 
  \hline
1/2 & 500 & 500 & All & 0.950 \\ 
  1/2 & 1000 & 1000 & All & 0.961 \\ 
  1/2 & 2000 & 2000 & All & 0.957 \\ 
  1/2 & 50000 & 50000 & All & 0.951 \\ 
   \hline
1/2 & 500 & 500 & None & 0.003 \\ 
  1/2 & 1000 & 1000 & None & 0.000 \\ 
  1/2 & 2000 & 2000 & None & 0.000 \\ 
  1/2 & 50000 & 50000 & None & 0.000 \\ 
   \hline
1/2 & 500 & 500 & OM and PS & 0.903 \\ 
  1/2 & 1000 & 1000 & OM and PS & 0.861 \\ 
  1/2 & 2000 & 2000 & OM and PS & 0.818 \\ 
  1/2 & 50000 & 50000 & OM and PS & 0.813 \\ 
   \hline
1/2 & 500 & 500 & TP, OM, and PP & 0.937 \\ 
  1/2 & 1000 & 1000 & TP, OM, and PP & 0.932 \\ 
  1/2 & 2000 & 2000 & TP, OM, and PP & 0.947 \\ 
  1/2 & 50000 & 50000 & TP, OM, and PP & 0.970 \\ 
   \hline
1/2 & 500 & 500 & TP, PS, and PP & 0.980 \\ 
  1/2 & 1000 & 1000 & TP, PS, and PP & 0.970 \\ 
  1/2 & 2000 & 2000 & TP, PS, and PP & 0.947 \\ 
  1/2 & 50000 & 50000 & TP, PS, and PP & 0.934 \\ 
   \hline
\end{tabular}
\caption{Coverage with (Trial Size):(Total Sample Size) Ratio of 1/2} 
\end{table}
% latex table generated in R 4.3.2 by xtable 1.8-4 package
% Wed Mar 27 11:49:11 2024
\begin{table}[ht]
\centering
\begin{tabular}{rlllr}
  \hline
Ratio & Trial Size & Target Size & Correctly-Specified Models & Estimated Coverage \\ 
  \hline
1/21 & 500 & 10000 & All & 0.940 \\ 
  1/21 & 1000 & 20000 & All & 0.949 \\ 
  1/21 & 2000 & 40000 & All & 0.950 \\ 
  1/21 & 50000 & 1000000 & All & 0.943 \\ 
   \hline
1/21 & 500 & 10000 & None & 0.000 \\ 
  1/21 & 1000 & 20000 & None & 0.000 \\ 
  1/21 & 2000 & 40000 & None & 0.000 \\ 
  1/21 & 50000 & 1000000 & None & 0.000 \\ 
   \hline
1/21 & 500 & 10000 & OM and PS & 0.765 \\ 
  1/21 & 1000 & 20000 & OM and PS & 0.738 \\ 
  1/21 & 2000 & 40000 & OM and PS & 0.707 \\ 
  1/21 & 50000 & 1000000 & OM and PS & 0.694 \\ 
   \hline
1/21 & 500 & 10000 & TP, OM, and PP & 0.940 \\ 
  1/21 & 1000 & 20000 & TP, OM, and PP & 0.918 \\ 
  1/21 & 2000 & 40000 & TP, OM, and PP & 0.918 \\ 
  1/21 & 50000 & 1000000 & TP, OM, and PP & 0.976 \\ 
   \hline
1/21 & 500 & 10000 & TP, PS, and PP & 0.985 \\ 
  1/21 & 1000 & 20000 & TP, PS, and PP & 0.978 \\ 
  1/21 & 2000 & 40000 & TP, PS, and PP & 0.962 \\ 
  1/21 & 50000 & 1000000 & TP, PS, and PP & 0.930 \\ 
   \hline
\end{tabular}
\caption{Coverage with (Trial Size):(Total Sample Size) Ratio of 1/21} 
\end{table}
% latex table generated in R 4.3.2 by xtable 1.8-4 package
% Wed Mar 27 11:49:11 2024
\begin{table}[ht]
\centering
\begin{tabular}{rlllr}
  \hline
Ratio & Trial Size & Target Size & Correctly-Specified Models & Estimated Coverage \\ 
  \hline
1/3 & 500 & 1000 & All & 0.947 \\ 
  1/3 & 1000 & 2000 & All & 0.956 \\ 
  1/3 & 2000 & 4000 & All & 0.952 \\ 
  1/3 & 50000 & 100000 & All & 0.947 \\ 
   \hline
1/3 & 500 & 1000 & None & 0.000 \\ 
  1/3 & 1000 & 2000 & None & 0.000 \\ 
  1/3 & 2000 & 4000 & None & 0.000 \\ 
  1/3 & 50000 & 100000 & None & 0.000 \\ 
   \hline
1/3 & 500 & 1000 & OM and PS & 0.876 \\ 
  1/3 & 1000 & 2000 & OM and PS & 0.814 \\ 
  1/3 & 2000 & 4000 & OM and PS & 0.826 \\ 
  1/3 & 50000 & 100000 & OM and PS & 0.813 \\ 
   \hline
1/3 & 500 & 1000 & TP, OM, and PP & 0.934 \\ 
  1/3 & 1000 & 2000 & TP, OM, and PP & 0.922 \\ 
  1/3 & 2000 & 4000 & TP, OM, and PP & 0.940 \\ 
  1/3 & 50000 & 100000 & TP, OM, and PP & 0.959 \\ 
   \hline
1/3 & 500 & 1000 & TP, PS, and PP & 0.980 \\ 
  1/3 & 1000 & 2000 & TP, PS, and PP & 0.958 \\ 
  1/3 & 2000 & 4000 & TP, PS, and PP & 0.958 \\ 
  1/3 & 50000 & 100000 & TP, PS, and PP & 0.921 \\ 
   \hline
\end{tabular}
\caption{Coverage with (Trial Size):(Total Sample Size) Ratio of 1/3} 
\end{table}

\clearpage
\ifSubfilesClassLoaded{% <<<<<<<<<<<<<<<
  \bibliography{bibliography}% <<<<<<<<<<<<<<

\newcommand{\noop}[1]{}
\begin{thebibliography}{}

\bibitem[Angrist et~al., 1996]{angrist_identification_1996}
Angrist, J.~D., Imbens, G.~W., and Rubin, D.~B. (1996).
\newblock Identification of {Causal} {Effects} {Using} {Instrumental} {Variables}.
\newblock {\em Journal of the American Statistical Association}, 91(434):444--455.

\bibitem[Arel-Bundock et~al., Forthcoming]{marginal_effects_preprint}
Arel-Bundock, V., Greifer, N., and Heiss, A. ({Forthcoming}).
\newblock How to interpret statistical models using {marginaleffects} in {R} and {Python}.
\newblock {\em Journal of Statistical Software}.

\bibitem[Boos and Stefanski, 2013]{boos_stefanski}
Boos, D.~D. and Stefanski, L.~A. (2013).
\newblock {\em Essential {Statistical} {Inference}: {Theory} and {Methods}}.
\newblock Springer Science \& Business Media.

\bibitem[Bornkamp et~al., 2021]{bornkamp_principal_2021}
Bornkamp, B., Rufibach, K., Lin, J., Liu, Y., Mehrotra, D.~V., Roychoudhury, S., Schmidli, H., Shentu, Y., and Wolbers, M. (2021).
\newblock Principal stratum strategy: {Potential} role in drug development.
\newblock {\em Pharmaceutical Statistics}, 20(4):737--751.

\bibitem[Chattopadhyay et~al., 2020]{chattopadhyay_balancing_2020}
Chattopadhyay, A., Hase, C.~H., and Zubizarreta, J.~R. (2020).
\newblock Balancing vs modeling approaches to weighting in practice.
\newblock {\em Statistics in Medicine}, 39(24):3227--3254.

\bibitem[Colnet et~al., 2024]{colnet_causal_2024}
Colnet, B., Mayer, I., Chen, G., Dieng, A., Li, R., Varoquaux, G., Vert, J.-P., Josse, J., and Yang, S. (2024).
\newblock Causal {Inference} {Methods} for {Combining} {Randomized} {Trials} and {Observational} {Studies}: {A} {Review}.
\newblock {\em Statistical Science}, 39(1):165--191.

\bibitem[Crump et~al., 2009]{trimming_article}
Crump, R.~K., Hotz, V.~J., Imbens, G.~W., and Mitnik, O.~A. (2009).
\newblock Dealing with limited overlap in estimation of average treatment effects.
\newblock {\em Biometrika}, 96(1):187–199.

\bibitem[Dahabreh et~al., 2022]{dahabreh_itt}
Dahabreh, I.~J., Robertson, S.~E., and Hern{\'a}n, M.~A. (2022).
\newblock Generalizing and transporting inferences about the effects of treatment assignment subject to non-adherence.
\newblock arXiv:2211.04876 [stat].

\bibitem[Dahabreh et~al., 2023]{dahabreh_biometrics}
Dahabreh, I.~J., Robertson, S.~E., Petito, L.~C., Hern{\'a}n, M.~A., and Steingrimsson, J.~A. (2023).
\newblock Efficient and robust methods for causally interpretable meta-analysis: {Transporting} inferences from multiple randomized trials to a target population.
\newblock {\em Biometrics}, 79(2):1057--1072.

\bibitem[Dahabreh et~al., 2020]{dahabreh_extending_2020}
Dahabreh, I.~J., Robertson, S.~E., Steingrimsson, J.~A., Stuart, E.~A., and Hern{\'a}n, M.~A. (2020).
\newblock Extending inferences from a randomized trial to a new target population.
\newblock {\em Statistics in Medicine}, 39(14):1999--2014.

\bibitem[Degtiar and Rose, 2023]{degtiar_review_2023}
Degtiar, I. and Rose, S. (2023).
\newblock A {Review} of {Generalizability} and {Transportability}.
\newblock {\em Annual Review of Statistics and Its Application}, 10(1):501--524.

\bibitem[Ding, 2023]{ding_causal_intro}
Ding, P. (2023).
\newblock A {First} {Course} in {Causal} {Inference}.
\newblock arXiv:2305.18793 [stat].

\bibitem[Ding and Lu, 2017]{ding_principal_2017}
Ding, P. and Lu, J. (2017).
\newblock Principal stratification analysis using principal scores.
\newblock {\em Journal of the Royal Statistical Society: Series B (Statistical Methodology)}, 79(3):757--777.

\bibitem[Feller et~al., 2017]{feller_principal_2017}
Feller, A., Mealli, F., and Miratrix, L. (2017).
\newblock Principal {Score} {Methods}: {Assumptions}, {Extensions}, and {Practical} {Considerations}.
\newblock {\em Journal of Educational and Behavioral Statistics}, 42(6):726--758.

\bibitem[Finkelstein et~al., 2020a]{finkelstein_health_2020}
Finkelstein, A., Zhou, A., Taubman, S., and Doyle, J. (2020a).
\newblock Health care hotspotting—a randomized, controlled trial.
\newblock {\em New England Journal of Medicine}, 382(2):152--162.

\bibitem[Finkelstein et~al., 2020b]{hotspotting_dataset}
Finkelstein, A., Zhou, A., Taubman, S., and Doyle, J. (2020b).
\newblock {Replication Data for: Health Care Hotspotting — A Randomized, Controlled Trial}.

\bibitem[Follmann, 2000]{follmann_2000}
Follmann, D.~A. (2000).
\newblock On the {Effect} of {Treatment} among {Would}-{Be} {Treatment} {Compliers}: {An} {Analysis} of the {Multiple} {Risk} {Factor} {Intervention} {Trial}.
\newblock {\em Journal of the American Statistical Association}, 95(452):1101--1109.

\bibitem[Frangakis and Rubin, 2002]{frangakis_rubin_2002}
Frangakis, C.~E. and Rubin, D.~B. (2002).
\newblock Principal {Stratification} in {Causal} {Inference}.
\newblock {\em Biometrics}, 58(1):21--29.

\bibitem[Hern{\'a}n and Robins, 2020]{what_if_causal_book}
Hern{\'a}n, M.~A. and Robins, J.~M. (2020).
\newblock {\em Causal Inference: What If}.
\newblock {Boca Raton: Chapman \& Hall/CRC}.

\bibitem[Hudgens and Halloran, 2008]{hudgens_interference}
Hudgens, M.~G. and Halloran, M.~E. (2008).
\newblock Toward {Causal} {Inference} {With} {Interference}.
\newblock {\em Journal of the American Statistical Association}, 103(482):832--842.

\bibitem[{ICH}, 2020]{ich_guidelines}
{ICH} (2020).
\newblock {ICH} {E9} ({R1}) addendum on estimands and sensitivity analysis in clinical trials to the guideline on statistical principles for clinical trials.

\bibitem[Imbens and Rubin, 2015]{Imbens_Rubin_2015}
Imbens, G.~W. and Rubin, D.~B. (2015).
\newblock {\em Instrumental Variables Analysis of Randomized Experiments with One-Sided Noncompliance}, page 513–541.
\newblock Cambridge University Press.

\bibitem[Jiang et~al., 2016]{jiang_principal_2016}
Jiang, Z., Ding, P., and Geng, Z. (2016).
\newblock Principal {Causal} {Effect} {Identification} and {Surrogate} end point {Evaluation} by {Multiple} {Trials}.
\newblock {\em Journal of the Royal Statistical Society: Series B (Statistical Methodology)}, 78(4):829--848.

\bibitem[Jiang et~al., 2022]{jiang_multiply_2022}
Jiang, Z., Yang, S., and Ding, P. (2022).
\newblock Multiply robust estimation of causal effects under principal ignorability.
\newblock {\em Journal of the Royal Statistical Society: Series B (Statistical Methodology)}, 84(4):1423--1445.

\bibitem[Jo and Stuart, 2009]{jo_stuart_principal}
Jo, B. and Stuart, E.~A. (2009).
\newblock On the {Use} of {Propensity} {Scores} in {Principal} {Causal} {Effect} {Estimation}.
\newblock {\em Statistics in medicine}, 28(23):2857--2875.

\bibitem[J{\"u}ni et~al., 2001]{juni_assessing_2001}
J{\"u}ni, P., Altman, D.~G., and Egger, M. (2001).
\newblock Assessing the quality of controlled clinical trials.
\newblock {\em BMJ}, 323(7303):42--46.

\bibitem[Kang and Schafer, 2007]{kang_demystifying_2007}
Kang, J. D.~Y. and Schafer, J.~L. (2007).
\newblock Demystifying {Double} {Robustness}: {A} {Comparison} of {Alternative} {Strategies} for {Estimating} a {Population} {Mean} from {Incomplete} {Data}.
\newblock {\em Statistical Science}, 22(4):523--539.

\bibitem[Kennedy, 2023]{kennedy_review}
Kennedy, E.~H. (2023).
\newblock Semiparametric doubly robust targeted double machine learning: a review.
\newblock http://arxiv.org/abs/2203.06469.

\bibitem[Kennedy et~al., 2020]{kennedy_sharp_2020}
Kennedy, E.~H., Balakrishnan, S., and G{\textquoteright}Sell, M. (2020).
\newblock Sharp instruments for classifying compliers and generalizing causal effects.
\newblock {\em The Annals of Statistics}, 48(4):2008--2030.

\bibitem[Li et~al., 2018]{fan_li_tilting}
Li, F., Morgan, K.~L., and Zaslavsky, A.~M. (2018).
\newblock Balancing {Covariates} via {Propensity} {Score} {Weighting}.
\newblock {\em Journal of the American Statistical Association}, 113(521):390--400.

\bibitem[Lyu et~al., 2023]{lyu_bayesian_2023}
Lyu, T., Bornkamp, B., Mueller-Velten, G., and Schmidli, H. (2023).
\newblock Bayesian {Inference} for a {Principal} {Stratum} {Estimand} on {Recurrent} {Events} {Truncated} by {Death}.
\newblock {\em Biometrics}, 79(4):3792--3802.

\bibitem[Pearl, 2001]{pearl_2001}
Pearl, J. (2001).
\newblock Direct and indirect effects.
\newblock In Breese, J.~S. and Koller, D., editors, {\em Proceedings of the 17th Conference on Uncertainty in Artificial Intelligence}, pages 411--420. pp. 411--420. San Francisco: Morgan Kaufmann Publishers Inc.

\bibitem[Qu et~al., 2020]{qu_general_2020}
Qu, Y., Fu, H., Luo, J., and Ruberg, S.~J. (2020).
\newblock A {General} {Framework} for {Treatment} {Effect} {Estimators} {Considering} {Patient} {Adherence}.
\newblock {\em Statistics in Biopharmaceutical Research}, 12(1):1--18.

\bibitem[Robertson et~al., 2022]{balancing_intercept}
Robertson, S.~E., Steingrimsson, J.~A., and Dahabreh, I.~J. (2022).
\newblock Using {Numerical} {Methods} to {Design} {Simulations}: {Revisiting} the {Balancing} {Intercept}.
\newblock {\em American Journal of Epidemiology}, 191(7):1283--1289.

\bibitem[Robins and Greenland, 1992]{robins_greenland}
Robins, J.~M. and Greenland, S. (1992).
\newblock Identifiability and {Exchangeability} for {Direct} and {Indirect} {Effects}.
\newblock {\em Epidemiology}, 3(2):143.

\bibitem[Rubin, 2006]{rubin_censoring_death}
Rubin, D.~B. (2006).
\newblock Causal {Inference} {Through} {Potential} {Outcomes} and {Principal} {Stratification}: {Application} to {Studies} with {\textquotedblleft}{Censoring}{\textquotedblright} {Due} to {Death}.
\newblock {\em Statistical Science}, 21(3):299--309.

\bibitem[Rudolph and van~der Laan, 2017]{rudolph_robust_2017}
Rudolph, K.~E. and van~der Laan, M.~J. (2017).
\newblock Robust estimation of encouragement design intervention effects transported across sites.
\newblock {\em Journal of the Royal Statistical Society: Series B (Statistical Methodology)}, 79(5):1509--1525.

\bibitem[Rudolph et~al., 2024]{rudolph_efficient_2024}
Rudolph, K.~E., Williams, N.~T., Stuart, E.~A., and Diaz, I. (2024).
\newblock Improving efficiency in transporting average treatment effects.
\newblock https://arxiv.org/abs/2304.00117.

\bibitem[Steingrimsson et~al., 2023]{steingrimsson_missing}
Steingrimsson, J.~A., Barker, D.~H., Bie, R., and Dahabreh, I.~J. (2023).
\newblock Systematically missing data in causally interpretable meta-analysis.
\newblock {\em Biostatistics}, 25(2):289–305.

\bibitem[Tackney et~al., 2023]{conditional_vs_marginal}
Tackney, M.~S., Morris, T., White, I., Leyrat, C., Diaz-Ordaz, K., and Williamson, E. (2023).
\newblock A comparison of covariate adjustment approaches under model misspecification in individually randomized trials.
\newblock {\em Trials}, 24(1).

\bibitem[Tran et~al., 2023]{robust_var_tmle}
Tran, L., Petersen, M., Schwab, J., and Laan, M. J. v.~d. (2023).
\newblock Robust variance estimation and inference for causal effect estimation.
\newblock {\em Journal of Causal Inference}, 11(1).
\newblock Publisher: De Gruyter.

\bibitem[Tsiatis, 2006]{tsiatis_book}
Tsiatis, Anastasios, A. (2006).
\newblock {\em Semiparametric {Theory} and {Missing} {Data}}.
\newblock Springer {Series} in {Statistics}. Springer, New York, NY.

\bibitem[van~der Laan et~al., 2007]{super_learner}
van~der Laan, M.~J., Polley, E.~C., and Hubbard, A.~E. (2007).
\newblock Super learner.
\newblock {\em Statistical Applications in Genetics and Molecular Biology}, 6(1).

\bibitem[Van Der~Laan and Rose, 2011]{tmle_book}
Van Der~Laan, M.~J. and Rose, S. (2011).
\newblock {\em Targeted {Learning}: {Causal} {Inference} for {Observational} and {Experimental} {Data}}.
\newblock Springer {Series} in {Statistics}. Springer, New York, NY.

\bibitem[VanderWeele, 2008]{vanderweele_relations}
VanderWeele, T.~J. (2008).
\newblock Simple relations between principal stratification and direct and indirect effects.
\newblock {\em Statistics \& Probability Letters}, 78(17):2957--2962.

\bibitem[Wager, 2022]{wager_notes}
Wager, S. (2022).
\newblock Lecture notes in causal inference.
\newblock Available at https://web.stanford.edu/~swager/stats361.pdf.

\bibitem[Yang et~al., 2023]{yang_hospital_2023}
Yang, Q., Wiest, D., Davis, A.~C., Truchil, A., and Adams, J.~L. (2023).
\newblock {Hospital Readmissions by Variation in Engagement in the Health Care Hotspotting Trial: A Secondary Analysis of a Randomized Clinical Trial}.
\newblock {\em JAMA Network Open}, 6(9):e2332715--e2332715.

\bibitem[Zeng et~al., 2023]{zeng_efficient_2023}
Zeng, Z., Kennedy, E.~H., Bodnar, L.~M., and Naimi, A.~I. (2023).
\newblock Efficient {Generalization} and {Transportation}.
\newblock http://arxiv.org/abs/2302.00092.

\bibitem[Zhang et~al., 2009]{parametric_princ_strata}
Zhang, J.~L., Rubin, D.~B., and Mealli, F. (2009).
\newblock Likelihood-{Based} {Analysis} of {Causal} {Effects} of {Job}-{Training} {Programs} {Using} {Principal} {Stratification}.
\newblock {\em Journal of the American Statistical Association}, 104(485):166--176.

\bibitem[Zhang et~al., 2023]{zhang_double_2023}
Zhang, Y., Chakrabortty, A., and Bradic, J. (2023).
\newblock Double robust semi-supervised inference for the mean: selection bias under {MAR} labeling with decaying overlap.
\newblock {\em Information and Inference: A Journal of the IMA}, 12(3):2066--2159.

\bibitem[Zhou et~al., 2019]{zhou_cace_meta}
Zhou, J., Hodges, J.~S., Suri, M. F.~K., and Chu, H. (2019).
\newblock A {Bayesian} {Hierarchical} {Model} {Estimating} {CACE} in {Meta}-{Analysis} of {Randomized} {Clinical} {Trials} {With} {Noncompliance}.
\newblock {\em Biometrics}, 75(3):978--987.

\end{thebibliography}
}{}
\end{document}
\end{document}